\documentclass[11pt]{article}

\usepackage{array}
\usepackage{amssymb}
\usepackage{amsmath}
\usepackage{amsfonts}
\usepackage{graphicx}
\usepackage{amsthm, amsfonts}
\usepackage{epsfig}
\usepackage{amsfonts}
\usepackage{enumerate}
\usepackage{amssymb}
\usepackage{makeidx}
\usepackage{amsmath}
\usepackage{epsfig}
\usepackage{subfigure}
\usepackage{here}
\usepackage{enumerate}
\usepackage{setspace}
\usepackage{url}
\usepackage{hyperref}
\usepackage{wasysym}
\usepackage{rotating}
\usepackage{caption}
\usepackage{tabularx}
\usepackage{float}

\newtheorem{theorem}{Theorem}
\newtheorem{corollary}[theorem]{Corollary}
\newtheorem{definition}[theorem]{Definition}
\newtheorem{lemma}[theorem]{Lemma}
\newtheorem{proposition}[theorem]{Proposition}

\def\P{{\rm Pr}}
\begin{document}

\def\l{\lambda}

\def\Implies{\Longrightarrow}

\bibliographystyle{chicago}
\title{How Advance Sales can Reduce Profits: When to Buy, When to Sell, and What Price to Charge}

\author{Amihai Glazer\footnote{Department of Economics, University of California, Irvine.  aglazer@uci.edu }
 \and Refael Hassin\footnote{Statistics and Operations Research, Tel Aviv University
Tel Aviv 69978, Israel. hassin@post.tau.ac.il}  \thanks{This research was supported by the Israel Science Foundation (grant No. 1571/19)}
\and Irit Nowik\footnote{Department of Industrial Engineering and Management, Jerusalem College of Technology, Jerusalem, Israel. nowik@jct.ac.il}}

\date{April 2, 2023}

\maketitle

\begin{abstract}

(1) Problem definition: Consider consumers who prefer to consume a good later rather than earlier. If the price is constant, then we would expect consumers to wait to buy the good. That does not hold if consumers are concerned that others will buy the good early, so that a shortage will later occur. When will consumers arrive when they fear a shortage? What is the profit-maximizing policy of a monopolist? Might the firm lose profits by offering advance sales?

(2) Academic/Practical relevance: The timing of consumer arrivals is much studied. Little consideration, however, has addressed how anticipated shortages affect arrival times. The application is important: managers want to know when consumers will arrive, when they should make the product available, and what price to charge to maximize profits.

(3) Methodology:  We use game theory. We analyze analytically outcomes when a single item is for sale: we give closed solutions for the equilibrium customer behavior and profit-maximizing firm strategy and conduct sensitivity analysis. For generalization concerning more than one unit we give some analytical results and provide many numerical solutions.

(4) Results: When the price is constant over time, then even with no operating cost of doing so, offering advance sales reduces profits. If, however, the firm must offer both advance sales and later sales, then the profit-maximizing price induces all arrivals at the same time (either early or late, depending on the parameters). An increase in the number of units offered for sale increases the profit-maximizing price and increases the firm's expected profit. The equilibrium strategy of consumers can generate some unexpected behavior. The arrival rate may increase with the price of the good. For a given price, an increase in the number of units for sale increases the number of consumers who arrive early.

(5) Managerial implications:
The firm should offer the good only at the time consumers most desire it, and not earlier. Additionally, the profit-maximizing price can be derived from our analysis. This price is not the price which maximizes the expected number of arrivals. 

\bigskip

\end{abstract}

\pagebreak

\section{Introduction}

Consider consumers who fear that demand for a good may exceed supply. A consumer may then seek to buy the good before others do, even if he finds it costly to do so (he may, for example, incur inventory holding costs). What holds for one consumer can apply to all, and so in equilibrium many consumers will buy too early. A firm should recognize that such early buying reduces the price consumers are willing to pay. Therefore, in contrast to work which shows how firms can profit from advance sales, we will show that under plausible conditions advance sales reduce the firm's profits. 

Several other issues arise. Suppose the firm is required to sell in all periods. It can affect when consumers arrive by the price it sets. A low price would induce consumers to arrive early: the consumer surplus from getting the good is high, and so consumers are willing to incur the cost of buying early to get the good. If the price is high, then the cost of buying early can make it not worthwhile to buy the good, and so consumers will arrive late. As we shall see, the seller who must offer the good in all periods maximizes profits by setting the price so that either all consumers arrive in only one period, either all arrive early, or else all arrive late.

(We observe that many stores are not free to choose their opening hours or days, but are instead forced to remain open when they may prefer to be closed. For example, a law firm states that ``Most leases of retail space require the tenant to maintain certain days and hours of operations. A retail store will be expected to be open for business when the other retailers in the shopping center are open."\footnote
{
\url{https://www.stephensonlaw.com/news/tenant-negotiations-for-a-shopping-center-lease/}
}
The Irvine Company, a major landlord in southern California, says this explicitly:``it is a condition of each lease that all tenants keep the entire 
premises continuously open for business during the days and hours established by the Landlord."\footnote
{
\url{https://www.shopirvinecompany.com/merchantmanuals/Crossroads_Manual.pdf
}}

Firms recognizing that consumers may rush to buy too early may offer a good or service for sale only close to the time of its consumption. Thus, Apple made pre-orders for its  iPhone 13 available just seven days before it was available in stores.\footnote
{
\url{https://www.apple.com/newsroom/2021/09/apple-introduces-iphone-13-and-iphone-13-mini/}
}
Moreover, it is quite common for all units (such as movie or concert tickets) to be sold in this advance sale phase.\footnote
{
Sales of tickets, and of some other goods potentially subject to shortages, can induce  speculators to buy in an early period, and resell the good at a high price in a later period if demand turns out to be high. Su (2010) nicely models such behavior.
}
For example, tickets for one of the most anticipated films of 2021, ``Spider-Man: No Way Home" went on sale on November 29, 2021, three weeks before the movie first began showing in movie theaters; within minutes domestic movie ticket sites began to crash as moviegoers rushed to snag seats.\footnote
{
\url{https://www.cnbc.com/2021/11/29/spider-man-no-way-home-ticket-demand-crashed-box-office-sites.html}
}
In both the National Football League (NFL) and the National Basketball Association (NBA) sellout of season tickets is the rule rather than the exception.  Frantic supporters jam phone lines when tickets go on sale.\footnote{
\url{https://onlinelibrary.wiley.com/doi/pdf/10.1002/mde.4090150513}
}

Our paper also sheds light on inventory problems. If the firm is known to offer the good only in a single period it never holds any inventory.  So examining inventory issues requires asking why the firm offers the good in multiple periods. A further contribution of the model is to examine a seller's decision of when to offer the good for sale, and thus to allow examination of how governmental regulations 
affect prices or sales.

The paper shows that even if the firm is forced to offer the good  early, then it still profits from having all consumers who arrive do so at the same time. Either all arrive early or all arrive late, depending on the parameters of the model. In all circumstances (that is, for all parameter values) the firm's profits are smaller if consumers come in multiple periods. Our model has  consumption in period 2 valued more highly than consumption in period 1. But similar effects arise not if the benefit of consuming the good increases over time, but if the cost (other than price) of buying the good declines over time, for example because of consumers' inventory holding costs.

\section{Literature}

Our topic relates to, but differs from, advanced selling, also known as presales. Advanced selling lets consumers ensure that their demand is satisfied by purchasing it before the good is available. This common practice is largely motivated for two reasons (see, for example, Li and Zhang, 2013, and the references therein). First, advanced selling enables the seller to better predict demand and hence better plan production. Second, it lets the seller segment the market and price discriminate between those who buy in advance and those who do not.\footnote
{
Edelstein et al. (2012) also consider ``risk-sharing between purchaser and developer'' when buying a housing unit in the presale market.  This line of research may lead to a discount in the presale phase to compensate consumers for their uncertainty regarding the product, or an increased price as a premium for guaranteed supply (see also Shugan and Xie, 2000, and Yu et al., 2015).
}
Our model differs from this literature. First, we assume that the seller has a fixed amount of inventory for sale, and hence there is no production planning. Second, we assume homogeneous consumers and that the price is fixed over all periods and all consumers, so that price discrimination is impossible. Instead, we ask whether offering to sell the good before the time it is consumed benefits the seller. We assume that a consumer would buy the good early only because he fears that the good will later be unavailable. And because we assume consumers incur a cost of arriving at a store, some may choose not to attempt to buy the good at all.

The effect of a possible shortage on consumer behavior has been considered by Glazer and Hassin (1986) in a deterministic EOQ model where consumers plan their purchase according to published availability periods---possibly purchasing before or after their ideal time of consumption.   Su and Zhang (2009) consider a news vendor model with stochastic demand and strategic consumers who may forego their intended purchase if the probability of a stockout is high and if the fixed search cost is large.  In contrast, we consider advance sale of a product that will be offered for sale at a given time.

Classic inventory papers allow backorders, assuming that customers may buy a good later than at the time they most desire it, and assuming that the delay costs are incurred by the seller (a cost which may include compensation to the waiting consumers). An extensive literature extends the standard model to consider strategic consumers who may postpone their purchase in expectation of a lower price. In a seminal paper, Eppen and Lieberman (1984) model an EOQ-type firm and consumers who consume the good at a constant rate. Consumers incur inventory holding costs; if the price is constant they will not hold inventory. But the firm offers to sell at a lower price at the beginning of each cycle, and consumers can exploit it by stockpiling. Further work considers customers who differ in their valuations of the good; they are indifferent about the time of purchase, or prefer to buy early (e.g., Su, 2007). Their decision of when to buy trades-off securing the purchase by buying early, and hoping to buy it later at a lower price while facing the risk of shortage. In equilibrium, those with higher valuations buy earlier at higher prices, whereas the others wait. Some literature also examines the effect of non-strategic (myopic) consumers in a mixed population. These works include Besanko and Winston (1990), Gallien (2006), Elmaghraby et al. (2008), Aviv and Pazgal (2008), Gallego et al. (2008), Zhang and Cooper (2008), Cachon and Swinney (2009), Kremer at al. (2017), Shum et al. (2017), and Zhang et al. (2019). In contrast, our model has a consumer value the good more highly in the second period. Glazer and Hassin (1982) and Cachon and Feldman (2011) study subscription pricing, which can be a form of advance selling applied to repeat customers, whereas in our model consumers make a single purchase.

An important and insightful paper related to ours is Liu and Ryzen (2008). They show that a firm can benefit from generating stockouts, because potential shortages allow the firm to price discriminate among heterogeneous consumers. This resembles strategic delays in queuing systems, which encourage customers to buy priority (see Afeche, 2013). We do not allow price discrimination so that this factor does not affect our results. Our paper differs in several additional ways. We let the number of potential buyers be stochastic, allow consumers to choose whether and when to arrive, and let a risk of shortage exist in the first period.

Under some conditions, the firm may profit from a period with zero stock, which induces some customers to strategically wait for a price reduction. That benefits the seller by reducing its inventory holding costs (Chen and Zhu 2020). Glazer and Hassin (1986, 1990) extend the choices available to consumers by allowing them to advance or delay their purchase when the good is not available at their ideal times of purchase, while incurring inventory holding costs or waiting costs. The firm, recognizing this strategic behavior, profits from artificially creating  no-sale periods, affecting customers' behavior and reducing its inventory holding costs. In the profit-maximizing solution each inventory cycle consists of an interval with sales followed by an interval with no sales. Two extremes are also possible: continuous sales throughout the cycle, or sales only at replenishment instants (i.e., inventory is never held). Furthermore, the profit-maximizing solution has shorter no-sale periods than the socially-optimal solution. With heterogeneous consumer valuations the firm may sell at discrete points within the inventory cycle; customers arriving between two such points choose between advancing or delaying their purchase (Anily and Hassin 2013) .

Antoniou and Fiocco (2019) investigate a firm's inventory strategy when facing  forward-looking consumers who can store a good in expectation of higher future prices. They show that a retailer unable to commit to future prices can profit from holding inventory. Other work, well surveyed by Antoniou and Fiocco (2019), considers buyer stockpiling in expectation of higher prices. But in that work, unlike in ours, a consumer who buys early does not affect the welfare of a consumer who intends to buy later. DeGraba (1995) finds that a firm may intentionally create a shortage, inducing consumers to buy before they have complete information about the value of the good. Huchzermeier et al. (2002) assume that consumers are heterogeneous and that the firm offers product variety in the form of different package sizes, which enable a discrete form of stockpiling at discounted prices. Chen and Shi (2019) analyze a related model where price changes are also made within the cycle.

Several papers compare profits and consumer welfare under dynamic and static pricing. Because static pricing is a special case of dynamic pricing, it is not surprising that dynamic pricing generates larger profits. The firm's benefit comes from two sources. First, it allows for price discrimination---consumers with a high initial valuation for the good, but with valuation that declines much over time, would be willing to pay a higher price when the good is first offered, even if they realize that the firm will later reduce the price (Golrezaei, Nazerzadeh, and Randhawa 2020). The opposite pattern, of increasing prices over an inventory cycle but lower average prices, can increase consumer welfare: a price which increases over time induces consumers to buy early in an inventory cycle, thereby reducing the firm's inventory holding costs (Stamatopoulos, Chehrazi, and Bassamboo 2019). In these models of dynamic pricing the firm sells goods all the time. Though these models have strategic consumer behavior, they consider a consumer buying after his ideal time. In contrast, we allow a consumer to buy before his ideal time. Moreover, unlike other analyses, we allow consumers to buy early because they fear a shortage in the future.

Our assumption that firms find it costly to change prices, resulting in sticky prices, is common in macroeconomic models. A seminal article making this assumption is Sheshinski and Weiss (1977). They build on inventory models, arguing that changing a price induces a fixed cost (called a menu cost) for the firm, inducing a firm to adopt an $(S,s)$ rule, with prices exhibiting a pattern of inaction followed by large price changes. Sticky prices can arise from managerial costs, which induce a firm selling in different markets or stores to charge a uniform price (DellaVigna and Gentzkow 2017). Evidence shows that the stickiness of prices appears even when demand or cost conditions change, and is more common the less competitive the market (Bils and Klenow 2004). Consumer anger can also restrict the ability of firms to change prices:\footnote
{
\url{https://www.nytimes.com/2007/09/07/technology/07apple.html}
}
\begin{quote}
Apple ... angered many of its most loyal customers by dropping the price of its iPhone to \$400 from \$600 only two months after it first went on sale. They let the company know on blogs, through e-mail messages and with phone calls.

On Thursday, in a remarkable concession, Steven P. Jobs acknowledged that the company had abused its core customers' trust and extended a \$100 store credit to the early iPhone buyers.
\end{quote}

\section{Assumptions}

The central idea we explore is that a consumer may attempt to buy in an early period because he fears that so many other consumers will do the same as to make the good unavailable in period 2. The problem becomes interesting if some inefficiency is involved with such behavior, in the sense that a social planner would have people attempt to buy the good only in period 2. Inefficiency in our model is induced by assuming that a consumer incurs some extra cost, or loss of utility, in buying in period 1 rather than in period 2. The following makes the assumptions explicit. 

\begin{enumerate}

\item There are two periods. The monopolist has $N$ units of a good for sale. Units not sold in period 1 can be sold in period 2. The number $N$ of units on sale in the first period is given; the advance sales information is not used by the firm to plan its future production or costs.

\item The good's price is $P$. And because of the unvarying price,  consumers who arrive early are not motivated to do so because of a reduced price.

\item
A person may be able to buy the good in period 1 or in period 2. But, for a given price, he prefers buying it in period 2. This preference may result from consuming the good in period 2. In that case, we refer to a sale in period 1 as an advanced sale. Alternatively, the early buyer may consume the good early, but it will then be less valuable for him.

\item
The number of potential consumers is stochastic, following the Poisson distribution with intensity $\l$. The Poisson distribution is a natural assumption when the need for the product by each member of a large population independently arises, and the Poisson distribution approximates the binomial distribution.

\item
A consumer incurs a fixed cost, $c$, of arriving at the store. This cost can cause consumers to balk. Consumers are strategic, deciding whether to attempt to buy in period 1, or in period 2, or whether to balk.

\item
Consumers are homogeneous.  Hence, advance selling is not motivated by the firm's desire to price discriminate.

\item
Consumers know the value of getting the good: $V$ if bought in period 2, and $V-K$ if bought in period 1. The reduction in value by $K$ is a penalty for early purchase. This penalty can reflect, for example, a consumer's inventory holding costs. Consumers are uncertain about the number of arrivals in each period and so about the availability of the product. 

\item
Even when arriving in period 1, the advance sale period, obtaining the good is not assured.  However, arriving in period 1 increases the probability of getting the item; the fear of a shortage can cause consumers to arrive early, despite the smaller value of getting the good at that time.

\item In each period, if the number of consumers at the store exceeds the number of units available for sale, each consumer at the store has the same probability of obtaining the good.

\item Both consumers and the monopolist can compute in advance the expected utilities and the expected profits, and make decisions accordingly.

\end{enumerate}

Sections~\ref{1unit} and~\ref{SU} below analyze the model for a single unit ($N=1$); Section~\ref{SU} focuses on an unbounded (i.e., high) arrival rate $\l.$ All results are proven mathematically.  Section~\ref{Munits} analyzes the model for any number of units of the good. Numerical solutions show that all the qualitative results for a single item continue to hold for multiple units. Analytical results show the effects of the amount of supply on the profit-maximizing price and customers' behavior.

\section{Single unit for sale}\label{1unit}

Consider first outcomes when the firm has a single unit of the good to sell. Although all consumers value the good more highly in period 2, if a fraction $q_1>0$ nevertheless arrives in period 1, then the number of consumers arriving in period 1 also has a Poisson distribution with intensity $\l_1 \equiv \l q_1$; and the number of consumers arriving in period 2 has a Poisson distribution with intensity $\l_2 \equiv \l q_2,$ where $q_2$ is the fraction of consumers arriving in period 2 and $q_1+q_2\le1$.

A consumer arriving in period $t$ ($t=1, 2$) has expected utility $U_{t}$. A consumer who never arrives gets zero utility. The expected utilities of a consumer arriving in each period are
\begin{equation}\label{E9000}
U_{1}= -c+\sum_{j=0}^{\infty} \frac{V-K-P}{j+1}\cdot\frac{\l_1^j e^{-\l_1}}{
j!}=
-c+(V-P-K)\frac{\left(1-e^{-\l_1}\right)}{\l_1},
\end{equation}
\begin{equation}\label{E9001}
U_2=-c+e^{-\l_1}(V-P)\frac{\left(1-e^{-\l_2}\right)}{\l_2},
\end{equation}

where $V$ was defined as the value of the good,
 $P$ its price, and $K$ the penalty for buying early. The probability that a given consumer will obtain the good is $\frac{1}{j+1}$. To understand Equation~\eqref{E9000}, note that the expected utility of a consumer arriving in period 1, when $j$ other consumers arrive, is $ \frac{V-K-P}{j+1}.$ \ Similarly for Equation~\eqref{E9001}.
Note that $e^{-\l_1},$ appearing in $U_2,$ is the probability that no consumer arrives in period 1, and hence that the good is still available in period 2. Also note that the expression $1-e^{-\l}\over\l$ is monotonic in $\l.$
In particular, by L'H\^{o}pital's rule $\lim_{\l\to0}{1-e^{-\l}\over\l}=1$. This means that when no one else arrives a consumer arriving in period 1 has $U_{1}=-c+V-K-P$; a consumer arriving in period 2 has
$U_2=-c+V-P.$ \  The firm's expected profit (from both selling periods) is
\begin{equation}\label{Pi1}
\Pi=P(1-e^{-\l_1-\l_2}).
\end{equation}

\subsection{Equilibrium behavior}\label{EB}

 Let $q_t,$ (for $t=1,2$) be the equilibrium fraction of consumers who arrive in period $t$, or equivalently, the equilibrium probability that a consumer arrives in period $t$. Thus $1-q_1-q_2$ is the probability not to arrive at all. Note that both $q_1$ and $q_2$ may equal 0, or may each be positive, or one may equal 1 with the other equaling 0. The sum $q_1+q_2$ may be less than or equal to 1. In total, there are seven possibilities. In each of these possibilities the probabilities $(q_1, q_2)$  constitute an equilibrium, iff some conditions regarding the expected utilities $U_{1}$ and $U_2$ are met, as elaborated next. In all these possibilities listed below, we see the following characteristics: when with positive probability no consumers arrive (equilibrium types 1-4 below) and $0<q_t<1$ (i.e., $0<\l_t<\l$) the corresponding utility satisfies $U_t=0$. The result that $U_t=0$ follows from the indifference property of a Nash equilibrium with mixed strategies: since the utility when not arriving is $0,$ then all other actions that take part in the same equilibrium must also induce zero utility.

\paragraph{Equilibrium types}

\begin{enumerate}

\item A consumer who never arrives would have non-positive expected utility were he to arrive in either period. That is, $q_1=q_2=0 \Implies U_{1},U_2\le 0$.

\item A consumer who does not arrive in period 1, but arrives in period 2 with probability less than 1, has zero utility when arriving in period 2, and would get non-positive utility if he arrived in period 1. That is, $0<q_2<1$ and $q_1=0 \Implies U_2=0$ and $U_{1}\le 0$.

\item A consumer who arrives in period 1, or in period 2, or does not arrive at all, each with positive probability, has expected utility zero for each choice. That is, $0<q_1, q_2<1$ and $q_1+q_2<1 \Implies U_{1}=U_2=0$.

\item  A consumer who arrives in period 1 with probability less than $1$, but never arrives in period 2, must have expected utility zero when arriving in period 1; his expected utility were he to arrive in period 2 would be non-positive. That is, $0<q_1<1$ and $q_2=0 \Implies U_{1}=0$ and $U_2\le 0$.

\item A consumer arriving with positive probability in either period 1 or period 2 must be indifferent between arriving in these  periods. And if he always wants to arrive in some period, his expected utility must be non-negative. That is, $0<q_1, q_2<1$ and $q_1+q_2=1 \Implies U_{1}=U_2 \ge 0$.

\item A consumer arriving in period 1 has larger expected utility than when arriving in period 2 or  never arriving. That is,
$q_1=1 \Implies U_{1}\ge\max\{U_2,0\}$.

\item Lastly, a consumer arriving in period 2 must enjoy utility at least as large when arriving in period 1 or of never arriving. That is, $q_2=1 \Implies U_2 \ge \max\{U_{1},0\}$.
c
\end{enumerate}

Note that $U_t$ can be interpreted as the expected payoff to any consumer arriving in period $t.$ But it can also be interpreted as the average payoff, computed over all consumers.
In particular, $U_2=0$ implies that some of the consumers arriving in period 2 enjoy positive utility, whereas others arriving at the same time get negative utility, such that on average $U_2=0.$

Note also that in equilibrium types 1-4 the sum of the arrival rates $\l_1+\l_2 < \l$. In equilibrium types 5-7 the sum is $\l_1+\l_2 = \l$. This difference becomes significant in the following section, since it means that in equilibrium types 1-4, $\l_1$ and  $\l_2$ may both be low  even when $\l$ is high. Note that the equilibria can be defined and presented by $(q_1, q_2)$ or equivalently, by $(\l_1, \l_2).$ We mostly use $(\l_1, \l_2)$.

We find it convenient to use normalized monetary values by considering $c$ as the unit value. Thus we define  $v \equiv \frac{V}{c}$, $k \equiv \frac{K}{c}$, $p \equiv \frac{P}{c}$,  $\pi \equiv \frac{\Pi}{c}$, and  $u_t \equiv \frac{U_t}{c}$ \ $(t=1, 2).$ In particular, we obtain from~\eqref{E9000} and~\eqref{E9001}

\begin{equation}\label{E9000N}
u_{1}=
-1+(v-p-k)\frac{\left(1-e^{-\l_1}\right)}{\l_1},
\end{equation}
\begin{equation}\label{E9001N}
u_2=-1+e^{-\l_1}(v-p)\frac{\left(1-e^{-\l_2}\right)}{\l_2}.
\end{equation}

Given the price $p$, a consumer's net benefit from buying the good in period 2 is $v-p$. Thus, an increase in $p$ means moving to the left in the $(v,k)$ space. 
We prove that for any given parameters, there is at most one equilibrium of each type (see Lemma~\ref{L1} at the beginning of the Appendix).
The conditions that define the seven types of equilibrium allow us to characterize the regions in which each type of equilibrium exists. For example, in equilibrium type 1, $\l_1=\l_2=0$ and  $u_{1}, u_2\le 0.$
Since \ $\lim_{\l\to 0}{1-e^{-\l}\over\l}=1,$
we obtain in~\eqref{E9001N}
$$\lim_{\l_1,\l_2\to 0}u_{2}= \lim_{\l_1,\l_2\to 0}  -1+e^{-\l_1}(v-p)\frac{\left(1-e^{-\l_2}\right)}{\l_2} = -1+v-p.$$  Thus, in equilibrium type 1, $v-p\leq 1.$
Similarly, we obtain from~\eqref{E9000N} that in the type-1 equilibrium $v-p\leq k+1.$
Taken together, type 1 equilibria exist only in the region where $v-p\leq 1.$
In the same way we obtain characteristics for the other six types of equilibria listed above
 (see Section~\ref{SCHI} in the Appendix).
Figure \ref{f2} below describes the regions corresponding to the different types of equilibria when $\l=1$ . The point at the intersection of types 7, 5, 3, and 2 has $v-p=e/(e-1)$ and $k=1/(e-1)$.  The point at the intersection of types 7, 6, 5, 4, and 3 has $v-p=e$ and $k=e(e-2)/(e-1)$.

\begin{figure}[H]
\centering
\vspace{-5cm}
\includegraphics[scale=0.6]{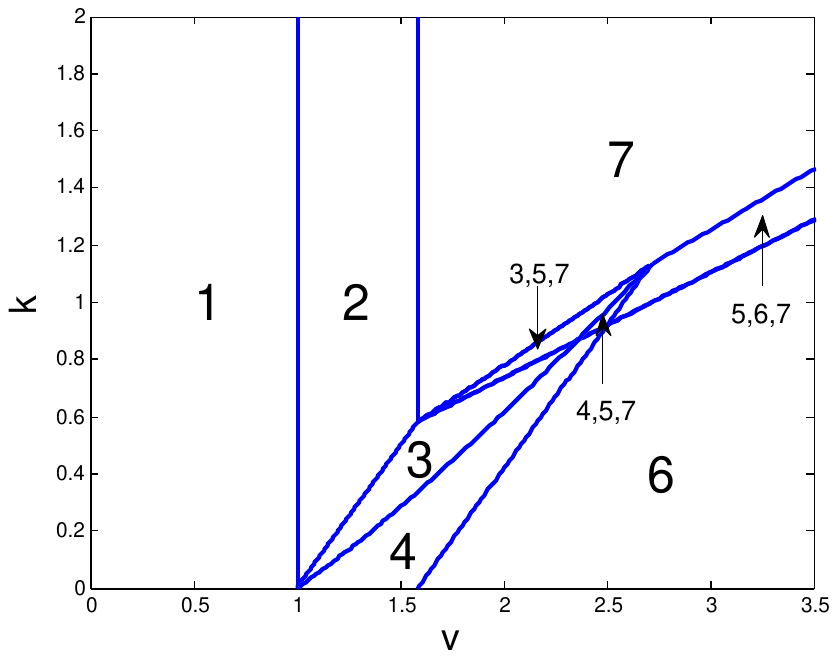}
\vspace{-5cm}
\caption{Types of equilibria in the $(v-p,k)$ space when $\l=1.$ In types 1-4 with positive probability no one arrives (i.e., $\l_0+\l_1< 1$); in type 1 no one arrives, in type 2 all arrivals occur late, and in type 4 all arrivals occur early. Type 3 has mixed arrivals (i.e., early and late arrivals). In types 5-7 all consumers arrive (i.e., $\l_0+\l_1= 1$). Type 5 has mixed arrivals. In type 6 all arrive early. In type 7 all arrive late.}\label{f2}
\end{figure}

As seen in Figure~\ref{f2}, several areas of $(v-p,k)$ have multiple equilibria. Thus, analyzing profit maximization requires making assumptions about which equilibrium arises, or examining multiple outcomes. 

However, Proposition~\ref{P66} (see Section~\ref{AP66} in the Appendix)  shows  that under unbounded potential demand each pair $(v, k)$ has a unique equilibrium point and this equilibrium is of one of the types 1-4. Regions 5-7 vanish because in these types of equilibria 
$\l_1+\l_2=\l.$ Hence, if $\l$ is large, then either $\l_1$ or $\l_2$ (or both) must also be large. But in that case~\eqref{E9000N} and~\eqref{E9001N} imply that either $u_1$ or $u_2<0,$ meaning this could not be an equilibrium.  Thus  types 5-7  are not possible and only  equilibrium types 1-4 (where  $\l_1+\l_2 <  \l$) arise. 

An unbounded demand approximates the case of high potential demand $\l$. Hence from now on, given $(v, k),$ we assume that the demand is high, where by `high' we mean that $\l$ is large enough so that there is a unique equilibrium and it is of type 1-4. For the  case of low demand, we give further analysis in the Appendix (see Section~\ref{AFIN}).


\section{High potential demand}\label{SU}

If $\l_1$ and $\l_2$ are both small (so that $u_1, u_2$ are non-negative), but $\l$ is large, the expected number of consumers $\l-\l_1-\l_2$ who want the good but nevertheless never arrive must also be large. Intuitively, note that if many consumers are expected to arrive the probability that an arriving consumer obtains the good is small. At some point the expected gain is smaller than the arrival cost $c$: utility becomes negative. At that point there are no additional arrivals. This implies that though potential demand is high (i.e., $\l$ is large), only few arrivals may be expected (i.e., $\l_1+\l_2$ is small).
For example, for $v=10$ and $k=2,$ in equilibrium, $\l_1=0$ and $\l_2=2.3,$  although $\l$ is unbounded.  In particular, the
 probability that no one arrives is high, and in that case the firm's profit is $0.$
The partition of the $(v,k)$ plane is shown in Figure \ref{f3}.

\begin{figure}[H]
\centering
\vspace{-5cm}
\includegraphics[scale=0.6]{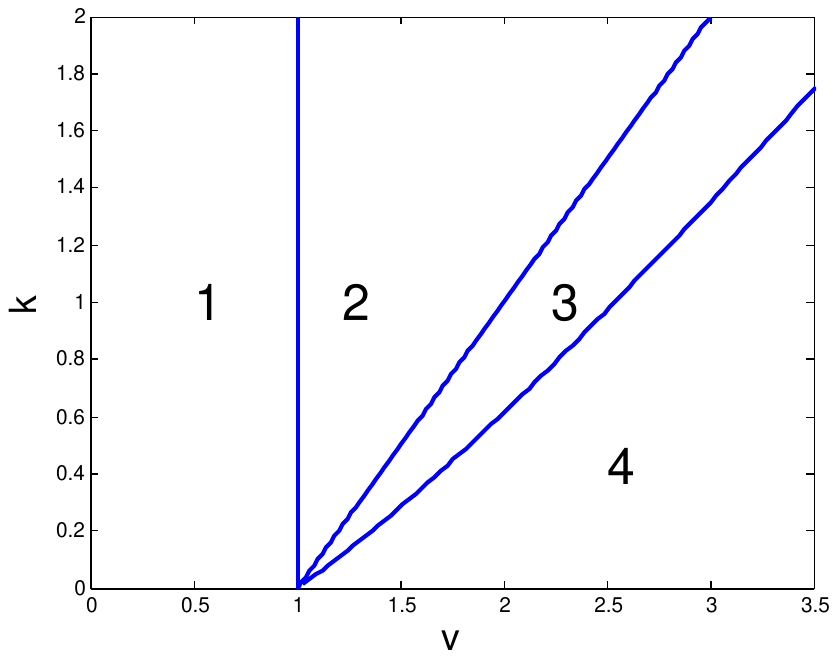}
\vspace{-5cm}
\caption{Types of equilibria in the $(v-p,k)$ space - unbounded potential demand. In type 1 no one arrives, in type 2 all arrivals occur late, and in type 4 all arrivals occur early. In type 3 there are mixed arrivals.}\label{f3}
\end{figure}

As seen, indeed only equilibria of types 1-4 exist and the regions corresponding to the four different equilibrium types do not overlap. Because the regions do not overlap, each $(v, k)$ has only one type of equilibrium. Combining this with Lemma~\ref{L1} in the Appendix, that any $(v, k)$ has at most one equilibrium of each type, we get the uniqueness of the equilibrium.

Recall that Figure~\ref{f2}, presented above, has
 $\l=1.$ As $\l$ increases the Figure changes and gradually  becomes  similar to Figure~\ref{f3}, which relates to an unbounded $\l.$ In particular, note that the lower-left corner of Figure~\ref{f2} resembles Figure~\ref{f3}. The equilibrium, and therefore the utility, depend on the region where we end up with after subtracting the price $p$ from the initial value $v$. As explained in Section~\ref{SCHI} at the beginning of the Appendix, the borders of Regions 1 and 2 are explicit; $0\leq v-p< 1$ constitute Region 1, and $1\leq v-p <k+1$  constitute Region 2. Recall that $v-p=\frac{V}{c}-\frac{P}{c}$. Thus, for example, the condition $0\leq v-p < 1$ that defines Region 1 corresponds to $0\leq V-P< c$, etc.
For $k>0,$  let  $v(k)$ denote the border between Regions 3 and 4. The proof of Proposition~\ref{P66} in Section~\ref{AP66} in the Appendix shows that $v(k)$ is determined by the equation
$$k=v-\frac{\ln{v}}{1-\frac{1}{v}},$$ for which $v=v(k)$ is the unique solution.
In Region 1 (where $0 \leq v-p < 1$), because the search cost exceeds the value of consuming the good, consumers never arrive. In Region 2 (where $1\leq v-p \leq k+1$), the value of an early purchase is smaller than the search cost, so consumers arrive only in period 2. In Region 3 (where  $k+1 < v-p < v(k)$), the arrival rate in period 1 makes the expected utility of  a consumer arriving then zero. But the added benefit, $k$, of buying in period 2 rather than in period 1 is sufficiently large to compensate for the reduced probability of getting the good in period 2. Thus, consumers arrive in both periods. In Region 4 (where $v-p\geq v(k)$), the small value of $k$ means that a large arrival rate in period 1 discourages consumers from arriving in period 2. Table 1 summarizes the results obtained so far.

\

\begin{center}

\begin{tabularx}{0.8\textwidth}
 {
  | >{\raggedright\arraybackslash}X
  | >{\centering\arraybackslash}X
  | >{\raggedleft\arraybackslash}X
   | >{\raggedleft\arraybackslash}X | }
\hline
Region & Borders of region in terms of $v-p$& Equilibrium & Utilities \\
\hline
 1  & $[0, 1)$  & $(0, 0)$ &$u_{1}, u_2\leq 0$ \\
\hline
 2  & $[1, k+1]$  & $(0, \l_2)$   & $u_{1}\leq 0,$  $u_2=0$\\
\hline
3  & $(k+1, v(k))$  & $(\l_1, \l_2)$ & $u_{1}=0,$ $u_2=0$\\
\hline
 4  & $[v(k), \infty)$  & $(\l_1, 0)$ &$u_{1}=0,$ $u_2\leq 0$\\
\hline
\end{tabularx}
\bigskip

Table 1: Definition and characteristics of the regions in terms of the net value  $v-p$.  In Region 1 no one arrives, in Region 2 all arrivals occur late, and in Region 4 all arrivals occur early. In Region 3 arrivals are mixed.

\end{center}

For example, if  $k=2$ and $v-p=2.5$ then $v-p$ lies in Region 2, and so  the unique equilibrium for $(v-p, k)=(2.5, 2),$ \ is \ $(0, \l_2),$ with $u_{1}\leq 0,$  $u_2=0.$
Table 1 shows that in each region whenever $\l_t>0$ the corresponding utility satisfies $u_t=0$. Whenever $\l_t=0$ the corresponding utility satisfies $u_t\leq 0.$ This was already explained in Section~\ref{EB} just before presenting all types of equilibria.

\subsection{Equilibrium arrival rates}\label{SAR}

We are interested in the behavior of $\l_1,\l_2$ and $\l_1+\l_2$ as functions of the net value $v-p$. We would expect that as the net value increases ( i.e., as the price $p$ declines), more people would arrive. It turns out that  this relation indeed holds for a wide range of parameter values, namely in Regions 2 and 4, (see Proposition~\ref{P1} in the Appendix). However, an exception arises in Region 3 (in which arrivals are mixed, i.e.,  $\l_1, \l_2>0$). In that case,  $\l_2$ increases with the price $p$. In other words, as long as we are still in Region 3, an increase in the price increases the expected number of arrivals in period 2. This seems counterintuitive. The reason is that a higher price induces fewer consumers to arrive in period 1, 
 $\l_1$ declines,  (see Proposition~\ref{P1} in the Appendix~\ref{ASAR}), making it more likely that the good is still available in period 2 (since the probability $e^{-\l_1}$ that the good is still available in period 2 increases with the price $p$), and so more consumers arrive in period 2.

Moreover, Theorem~\ref{T1} below shows that in Region 3 (where $\l_2$ increases with the price $p$) the {\textit{total} arrival rate  $\l_1+\l_2$ over both periods 1 and 2 also increases with the price, although $\l_1$ does decline with the price. Hence in Region 3 the increase in the arrival rate in period 2 exceeds the decrease in the arrival rate in period 1. 
Formally,

\begin{theorem}\label{T1}
In equilibrium, in Region 3 (where consumers may arrive in both periods) the aggregate arrival rate $\l_1+\lambda_2$ increases with the price $p.$
\end{theorem}

The proof of Theorem~\ref{T1} appears in the Appendix (see Section~\ref{PT1}). The conclusion is that whenever in equilibrium consumers arrive only in one of the periods, increasing the price leads to fewer arrivals. However, when the equilibrium arrivals are mixed, increasing the price leads to more arrivals in total.

Consider next the equilibrium $(\l_1, \l_2)$ in each of the four regions. The expressions that define their values involve the {\it Lambert function} $W[a],$ also called the omega function or product logarithm. The Lambert function is  the inverse to the function $ze^z$. Namely, if $a=ze^z$ then $W[a]=z.$
Since  the expressions for the equilibrium points $(\l_1, \l_2)$ turn out to be quite complex, they are described in the Appendix (see Theorem~\ref{T10} in Section~\ref{PP3}).

\subsection{Profit maximization}\label{PM}

The seller's expected profit in units of $c$ is the price $p>0$, multiplied by the probability that the good is sold, which is the same as the probability of at least one arrival. Thus,
\begin{equation}\label{E5000}
\pi=p\left(1-e^{-(\l_1+\l_2)}\right),
\end{equation}
where  $(\l_1, \l_2)$ is the unique equilibrium corresponding to $(v-p, k)$; it can be computed by using Theorem~\ref{T10} (see Section~\ref{ASARV} in the Appendix). Our next goal is to find for any given $(v, k),$ the price maximizing the firm's expected profit. Corollary~\ref{MON} follows immediately from Theorem~\ref{T1} presented in the previous section.

\begin{corollary}\label{MON}
Given a pair $(v, k)$ in Region 3, the profit $\pi$ increases with the price $p$ for all $p$ such that $(v-p,k)$ still lies in Region 3.
\end{corollary}
\begin{proof}
By~(\ref{E5000}), \ $\pi=p\left(1-e^{-(\l_1+\l_2)}\right).$ By Theorem~\ref{T1}, in Region 3 \  $\l_1+\l_2$ increases with $p.$ This implies that  $\pi=p\left(1-e^{-(\l_1+\l_2)}\right)$ also increases with $p$.
\end{proof}

Corollary~\ref{MON} is key in proving a central result regarding profit maximization. If the firm must offer the good in both periods 
then the equilibrium with the profit-maximizing price has no mixed arrivals: either all arriving consumers arrive in period 1, or all in period 2:

\begin{theorem}\label{FORCE}
Given $(v, k)$, the price that maximizes expected profit induces all arriving consumers to arrive at one time, either all arrive early or else all arrive late (depending on the parameters of the model).
\end{theorem}
\begin{proof}
The proof is straightforward when the maximum profit is attained in Regions 2 or 4, as by definition (see Table 1), \  $\l_1=0$ in Region 2  (i.e., arrivals occur in period 2 only), and $\l_2=0$ in Region 4 (i.e., consumers arrive only in period 1). But in Region 3 the equilibrium has  mixed arrivals, namely $\l_1, \l_2>0.$  By Corollary~\ref{MON}, the expected profit $\pi$  increases in $p$ in Region 3. Therefore, the profit-maximizing price does not lie in Region 3, but rather may only lie at the border separating  it from  Region 2. But at that point $\l_1=0,$ namely consumers arrive only in period 2. Hence the result.
\end{proof}

We examine next the behavior of $\pi$ as a function of $p.$
Table 1 defined the regions in terms of $v-p.$ The definition of the regions in terms of $p$ are:
\begin{itemize}
\item Region 1 consists of all $p$ satisfying \ $v-1\leq p\leq v$.
\item Region 2 consists of all $p$ satisfying \ $v-k-1\leq p< v-1$.
\item Region 3 consists of all $p$ satisfying \ $v-v(k)<p< v-k-1$.
\item Region 4 consists of all $p$ satisfying \ $0\leq p\leq v-v(k)$,
\end{itemize}
where $v=v(k)$ was defined as the border between Regions 3 and 4 in terms of $v-p.$

Note  that not all four regions exist for all pairs $(v,k)$. For example, if  $v-k-1<0$ then Regions 3 and 4 are empty. Recall that in Region 1 $\pi= 0.$ By Corollary~\ref{MON} in Region 3 profit is monotonically increasing. Finally, it follows from  Proposition~\ref{P5000} in the Appendix that in Regions 2 and 4 profit is strictly concave. Figure~\ref{F300} presents $\pi$ for $v=10$ and $k=2.$  The numbered regions in Figure \ref{F300} correspond to the regions in Figure \ref{f3}. For example, in both figures the region that is labeled as 2 corresponds to Region 2, (i.e., all the pairs $(v, k)$ with equilibrium of type 2).

 \begin{figure}[H]
\begin{center}
\includegraphics[scale=0.35]{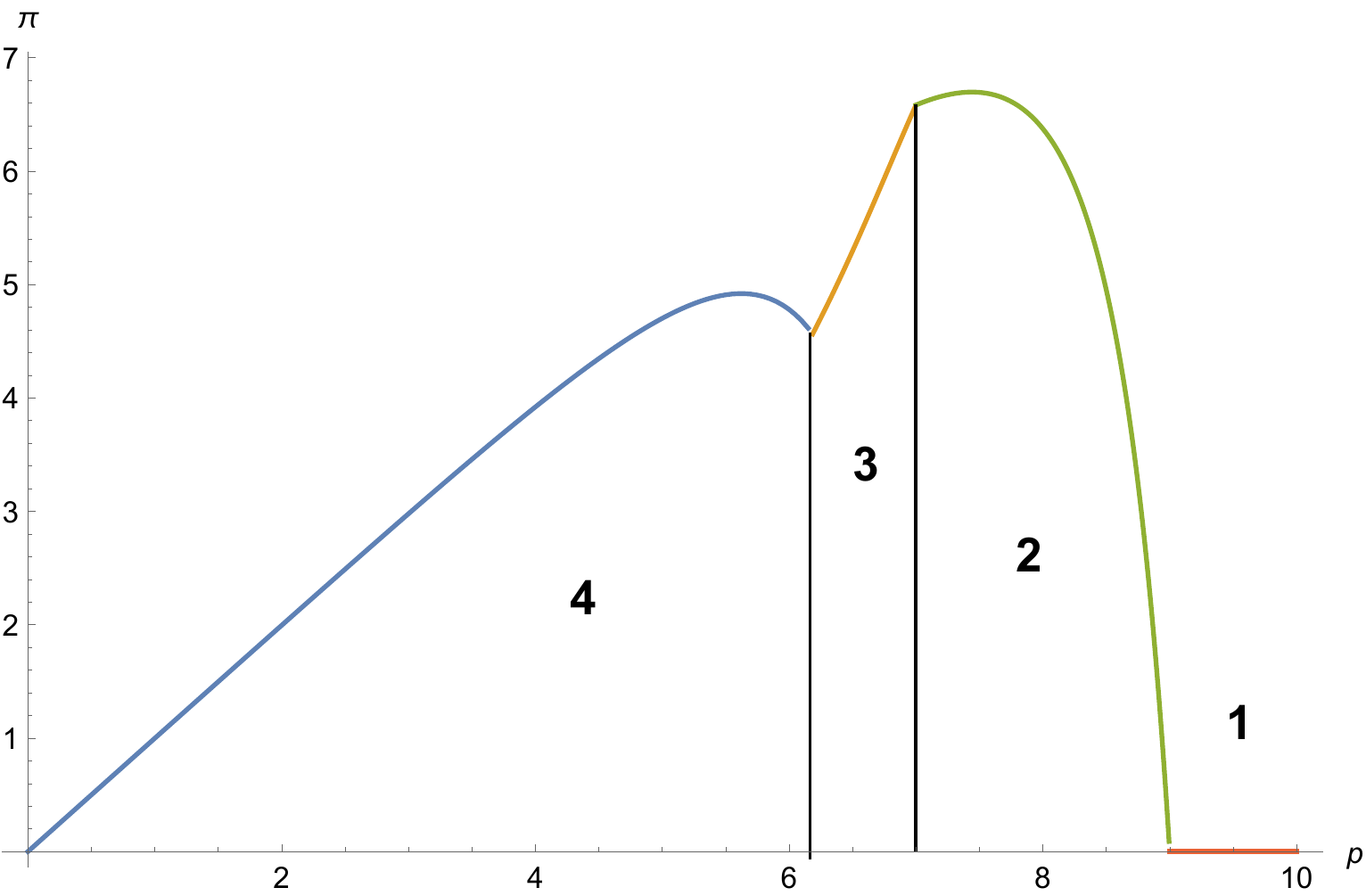}
\caption{$\pi$ as a function of $p,$ when $v=10$ and $k=2.$ In Region 1 no one arrives, in Region 2 all arrivals occur late, and in Region 4 all arrivals occur early. Region 3 has mixed arrivals.}\label{F300}
\end{center}
\end{figure}

Note that in Figure~\ref{f3} the order among Regions 1 to 4 is from left to right, whereas in Figure~\ref{F300} the order is reversed (right to left). The reason is that Figure~\ref{f3} relates to $v-p$ (in the $x$-axis), whereas Figure \ref{F300} relates to $p.$

As seen in the example presented in Figure~\ref{F300}, $\pi$ is indeed concave in Regions 2 and 4, and monotonically increasing in Region 3. In this example the profit-maximizing price lies in Region 2 (where all arrivals occur in period 2). In general, there are 3 candidates for the global maximum of $\pi$: the local maximum in Region 2 or in Region 4, or the right end of Region 3.

Theorem~\ref{T4} below describes for a given pair $(v, k)$ the price $p^*$ that maximizes expected profit $\pi$. Denote the local maximum of Regions 2 and 4 as $p_2$ and $p_4$ respectively. Denote the right end of Region 3 as $p_3$; then  $p_3=v-k-1.$ Denote \  $\pi_i =\pi(p_i), \ i=2,3,4.$  Lastly, let $p^*$ be the profit-maximizing price, and $\pi^*$ the global maximum  of the expected profit $\pi,$ namely $\pi^*=\pi(p^*).$

With $W[\cdot]$ denoting the Lambert function, let $W=W[-(k+1)e^{-(k+1)}]$, and let $v_f$ be the larger root of the function $\pi_4-\pi_3=\left( 1-\frac{v}{k+1}\right)W-\ln{(v-k)}.$

\begin{theorem} \label{T4}
Given the values of $v$ and of $k$, the profit-maximizing (normalized) price $p^*$, the induced arrival rates $(\l_1$ and  $\lambda_2$), and the associated expected (normalized) profit $\pi^*$ satisfy
\begin{enumerate}
\item If $v<1,$ then  $\pi^* =0,$ \  $\l_1=0,$ and  $ \lambda_2=0$.

\item If $1 \leq v\leq e^{W+k+1},$ then $p^*=p_2=v-{\ln v\over 1-{1\over v}},$  $\pi^*=\pi_2=v-1-\ln v,$ \ $\l_1=0, and \ \lambda_2=\ln{v}$.

\item If $e^{W+k+1}< v \leq v_f,$ then $p^*=p_3=v-k-1$, $\pi^*=\pi_3=(v-k-1)\left( 1-e^{-(W+k+1)}\right)$, $\l_1=0, and \ \lambda_2=W+k+1$.

\item If $v> v_f,$ then
$p^*=p_4=v-k-{\ln\left(v-k\right)\over1-{1\over v-k}}$,
$\pi^*=\pi_4=v-k-1-\ln\left(v-k\right)$, $\l_1=\ln{(v-k)}$, and $\lambda_2=0$.
\end{enumerate}

\end{theorem}

See Lemma~\ref{P8}  and Definition~\ref{D100} in the Appendix for elaboration on $v_f$.
The proof of Theorem~\ref{T4} is also in the Appendix (see Section~\ref{APM}).
As can be noted from all the different cases in Theorem~\ref{T4}, given a pair $(v,k)$ such that $v\geq 1,$ the profit-maximizing firm may set the price such that consumers either arrive only in period 2 (cases 2 and 3), or only in period 1 (case 4), but never in both periods. This is consistent with Theorem~\ref{FORCE}.

\begin{corollary}\label{SUP}
For all $(v, k)$ the upper bound of $\pi$ is  $ v-1-\ln{v}.$   This upper bound is realized only when \ $1 \leq v\leq e^{W+k+1}$. In that case $\pi^*=\pi_2 = v-1-\ln{v},$ and all arrivals are in period 2. If $(v, k)$ does not satisfy this condition, then the expected profit of the firm will be strictly smaller than $v-1-\ln{v}.$
\end{corollary}

See the proof of Corollary~\ref{SUP} in Section~\ref{ASUP} in the Appendix.

Note that if $k$ is sufficiently large then the condition $1 \leq v\leq e^{W+k+1}$ (appearing in case 2 of Theorem~\ref{T4}) is indeed satisfied and the upper bound $v-1-\ln{v}$ of $\pi$ is thus obtained.
Therefore a firm which controls the penalty $k$ should set it high enough so as to satisfy the condition.
\subsection{A firm does not profit by offering the good early}\label{SB}

The results in the previous section imply that the firm never profits by offering the good early, and in some cases its expected profit strictly declines by offering the good early, as the following theorem shows.

\begin{theorem}\label {Tonly2}
Profit maximization requires offering the good only at the time consumers most highly value it.
\end{theorem}
\begin{proof}
Consider a model with a single-period, in which the firm does not offer the good early. This is equivalent to assuming that $k$ is infinite which, when $v\geq 1$, leads to the second case in Theorem~\ref{T4} (since  $\lim\limits_{k\to\infty}e^{W+k+1}=\infty$, and so \  $1\leq v<e^{W+k+1}$), thus earning $\pi_2,$ which is the upper bound of the expected profit for all $k$.
Then, in the two-period case (unless $v<1$ which results in $\pi^*=0$),
 if $1 \leq v< e^{W+k+1}$ then offering the good in both periods does not affect profits. But if $v>e^{W+k+1}$  the firm's expected profit is strictly smaller if it offers the good in both periods and gains either $\pi_3$ or $\pi_4$ rather than offering it only in period 2, which is when all consumers most value it and gains the larger profit of $\pi_2$.
\end{proof}

One may expect that if consumers are heterogeneous in valuation and in the penalty then the results presented above do not hold. However, the following section shows that when each homogeneous sub-population is unbounded the heterogeneous case reduces to the  homogeneous case and thus the results still hold.

\subsection{Heterogeneous consumers}
Assume that consumers differ in the (normalized) value $v$ and in the early-arrival penalty $k.$
For simplicity, only two values: $v_1, v_2,$ with \ $v_1> v_2,$ and two penalties $k_1, k_2,$ \ with $k_1> k_2.$ So there are four sub-populations of consumers (corresponding to all possible combinations of $v_i,$ and $ k_j$). Denote by ${\rm SP}_{ij}$ the sub-populations of consumers having $v_i$ and $k_j$.
Let $u_{tij}$ be the normalized expected utility of a consumer from ${\rm SP}_{ij}$ who arrives in period $t.$
Then for each pair $i,j$ \ ($i=1,2, \ \  j=1,2$),
\begin{equation}\label{EEE1}
u_{1ij}= -1+\sum_{m=0}^{\infty} \frac{v_i-k_j-p}{m+1}\cdot\frac{\l_1^j e^{-\l_1}}{
m!}=
-1+(v_i-k_j-p)\frac{\left(1-e^{-\l_1}\right)}{\l_1},
\end{equation}
\begin{equation}\label{EEE2}
u_{2ij}=-1+e^{-\l_1}(v_i-p)\frac{\left(1-e^{-\l_2}\right)}{\l_2},
\end{equation}
where $\l_1, \l_2$ are the total arrival rates (from all sub-populations) in periods 1 and 2 respectively. Assume that demand by each sub-population is high. As in the homogeneous case, and for the same reason, equilibrium types 5-7 are not possible. Let $\l_{tij}$ be the arrival rate from ${\rm SP}_{ij}$ in period $t.$
Note that $v_1> v_2$ and $k_1> k_2$ imply that ${\rm SP}_{12}$ (with $v_1$ and $k_2$) dominates all other ${\rm SP}_{ij}$ in the expected utilities in all periods. Namely, for all  $(i,j)\ne (1,2)$: in period 1  $u_{112}>u_{1ij},$ and in period 2  $u_{212}>u_{2ij}$. (This is seen by substituting  $(v_1, k_2)$ in comparison to substituting $(v_i, k_j)\ne (v_1, k_2)$ \ in Equations~\eqref{EEE1} and~\eqref{EEE2}.)
In addition, as in the homogeneous case, whenever $\l_{tij}>0,$ (namely, arrivals are expected from ${\rm SP}_{ij}$ in period $t$), then $u_{tij}=0,$ and whenever
$\l_{tij}=0,$ (namely, no arrivals are expected from ${\rm SP}_{ij}$ in period $t$), then $u_{tij}\leq 0.$ 
Taken together, in equilibrium either $u_{tij}< u_{t12} < 0$ for all sub-populations, in which case there are no arrivals in period $t,$ or else $u_{tij}< u_{t12} = 0,$ in which case the only arrivals are from the superior sub-population ${\rm SP}_{12}$ (with $v_1$ and $k_2$).

Thus whenever some consumers arrive, these are from ${\rm SP}_{12}$ only, effectively ruling out heterogeneity and returning to the homogeneous case.

\subsection{The effect of the penalty $k$ for consuming early}

By Corollary~\ref{SUP}, when $1\leq v\leq e^{W+k+1},$ the upper bound $v-1-\ln{v}$ of $\pi$ is realized. Otherwise, the expected profit is strictly smaller.  This implies that if the firm must offer the good early the profit maximizing penalty for early purchase, $k$, should be large enough to satisfy $1\leq v\leq e^{W+k+1}.$ A further increase in $k$ does not affect $\pi,$ since $\pi=v-1-\ln{v}$ does not depend on $k.$
Accordingly, for any given $v,$ increasing $k$ such that the profit-maximizing price moves from $p_4$ to $p_3$ and finally to $p_2$ increases profits. Figures~\ref{FNEW1}-\ref{FNEW3N}
 depict the expected profit for $k=0.5, \ 1 , \ 2$.  In all these figures $v=10.$ Note that the upper bound $v-1-\ln{v}=6.697$ is not realized for $k=0.5$ and $k=1$ (see Figures~\ref{FNEW1} and~\ref{FNEW2}), but is realized for $k=2$ (see Figure~\ref{FNEW3N}).

\begin{figure}[H]
\minipage{0.32\textwidth}
\includegraphics[width=\linewidth=0.3]{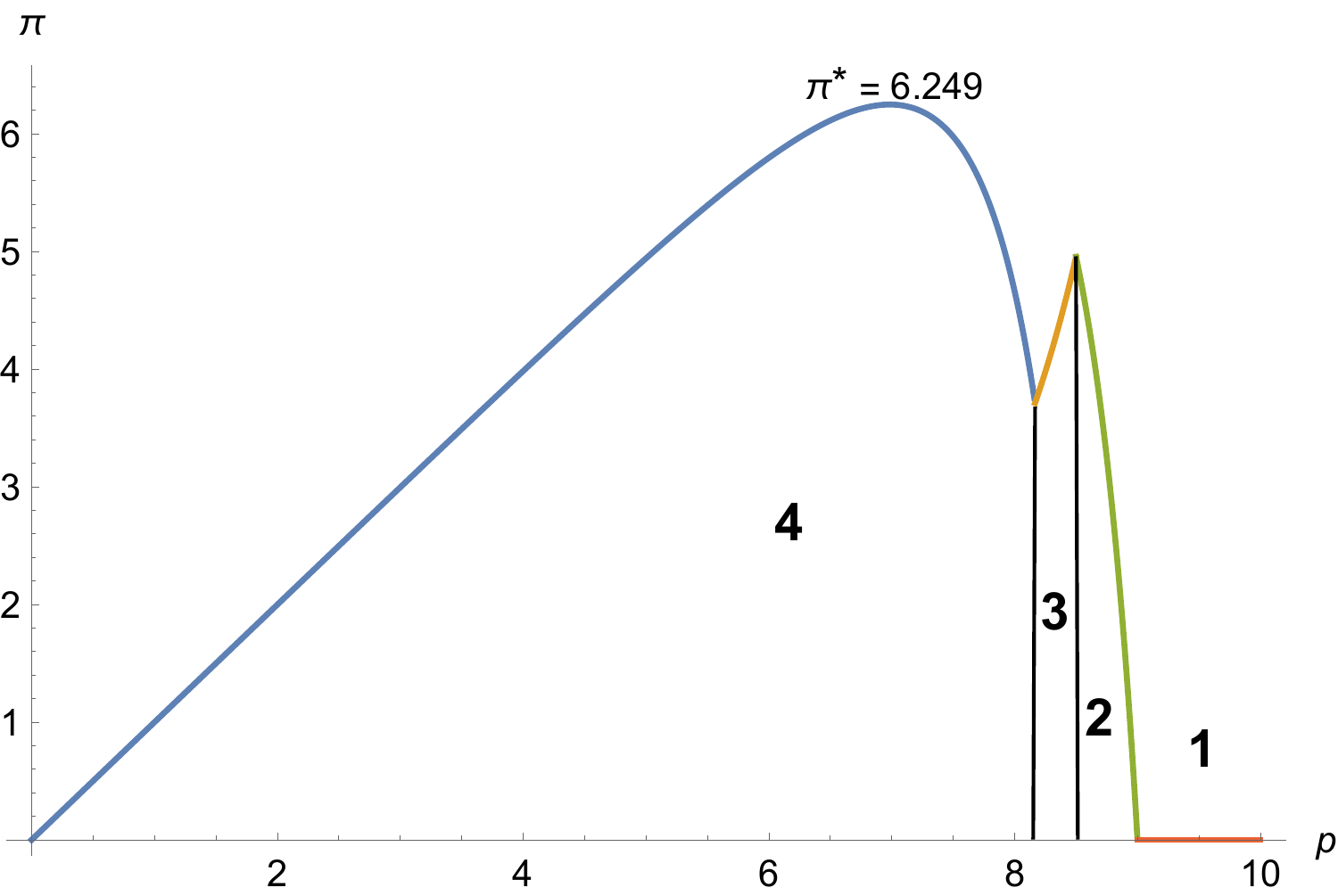}
\caption{$\pi$ when $k=0.5$}\label{FNEW1}
\endminipage\hfill
\minipage{0.32\textwidth}
\includegraphics[width=\linewidth=0.3]{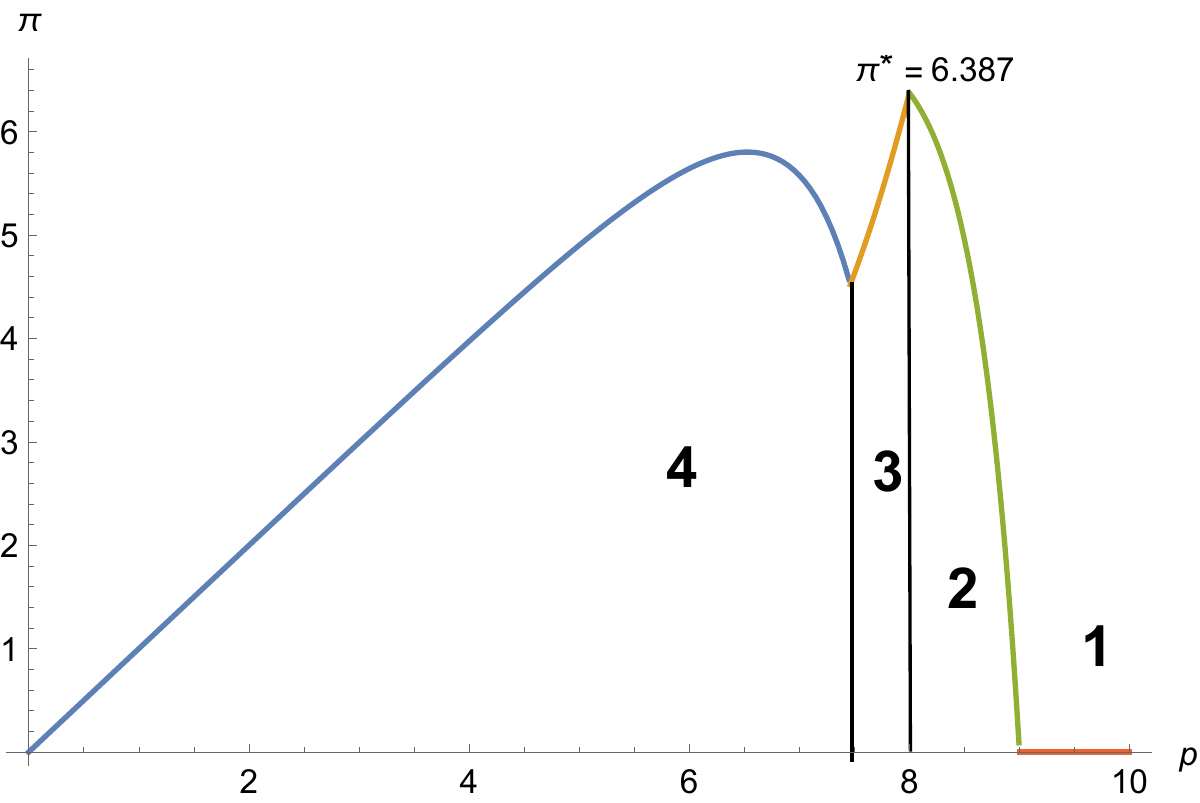}
\caption{$\pi$ when $k=1$}\label{FNEW2}
\endminipage\hfill\minipage{0.34\textwidth}
\includegraphics[width=\linewidth=0.3]{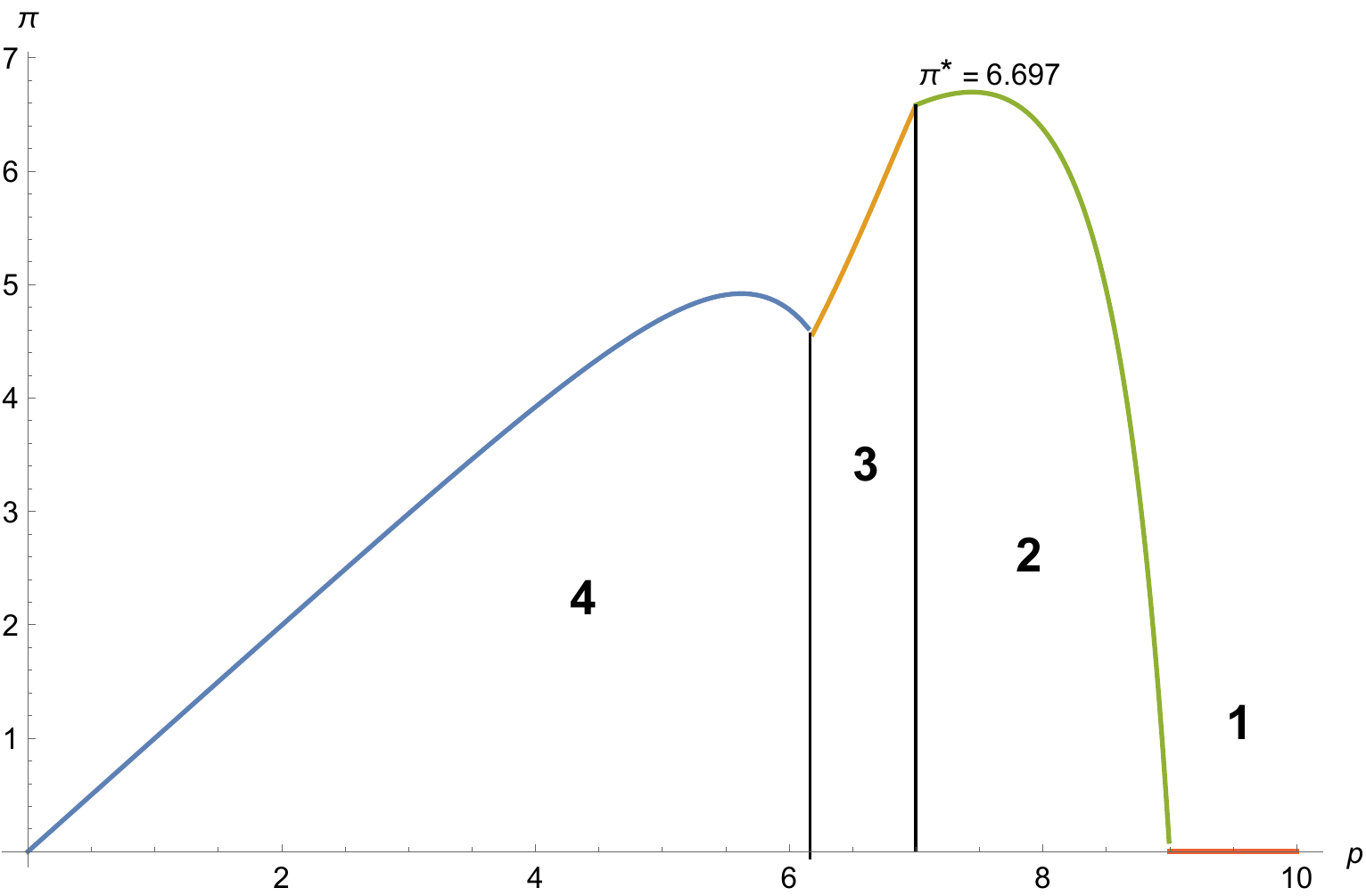}
\caption{$\pi$ when $k=2$}\label{FNEW3N}
\endminipage\hfill
\end{figure}

But what happens when the change in $k$ is small? In other words, if  $(v,k)$ lies in a certain region and $k$ increases just a little so that $(v,k)$ still lies in the same region, how does $\pi^*$ change? Obviously, in Region 1, where $\pi^*\equiv 0,$ and in Region 2, where $\pi^*=v-1-\ln{v},$ \ $\pi^*$ is unaffected by a change in $k$. However, in Regions 3 and 4 $\pi^*$ decreases with $k.$ (For Region 4 this is easy to prove analytically. For Region 3 we establish this numerically).
So increasing the penalty of early arrival without changing the arrival profile (i.e, we are still in the original region), reduces expected profits.

\subsection{The effect of the arrival cost $c$}\label{EC}

The previous analysis normalized all monetary values by considering the arrival cost as the unit value. Now we study how a change in the arrival cost, $c$, affects the arrival rates $(\l_1, \l_2)$, and affects the expected profit. In the normalized case we obtained the upper bound  $\pi^*=v-1-\ln{v}$ (see~Corollary \ref{SUP}). From this we deduce  the non-normalized  upper bound:
$\Pi^*=V-c-c\ln{\frac{V}{c}}.$
As we would expect, given $(v, k),$ both $(\l_1, \l_2)$  and $\pi$ decrease with $c$. Figures~\ref{Fchalf}-~\ref{Fc6} show $\Pi$ as a function of $P$ when $V=10$ and $K=2.$

\begin{figure}[H]
\minipage{0.3\textwidth}
\includegraphics[width=\linewidth=0.2]{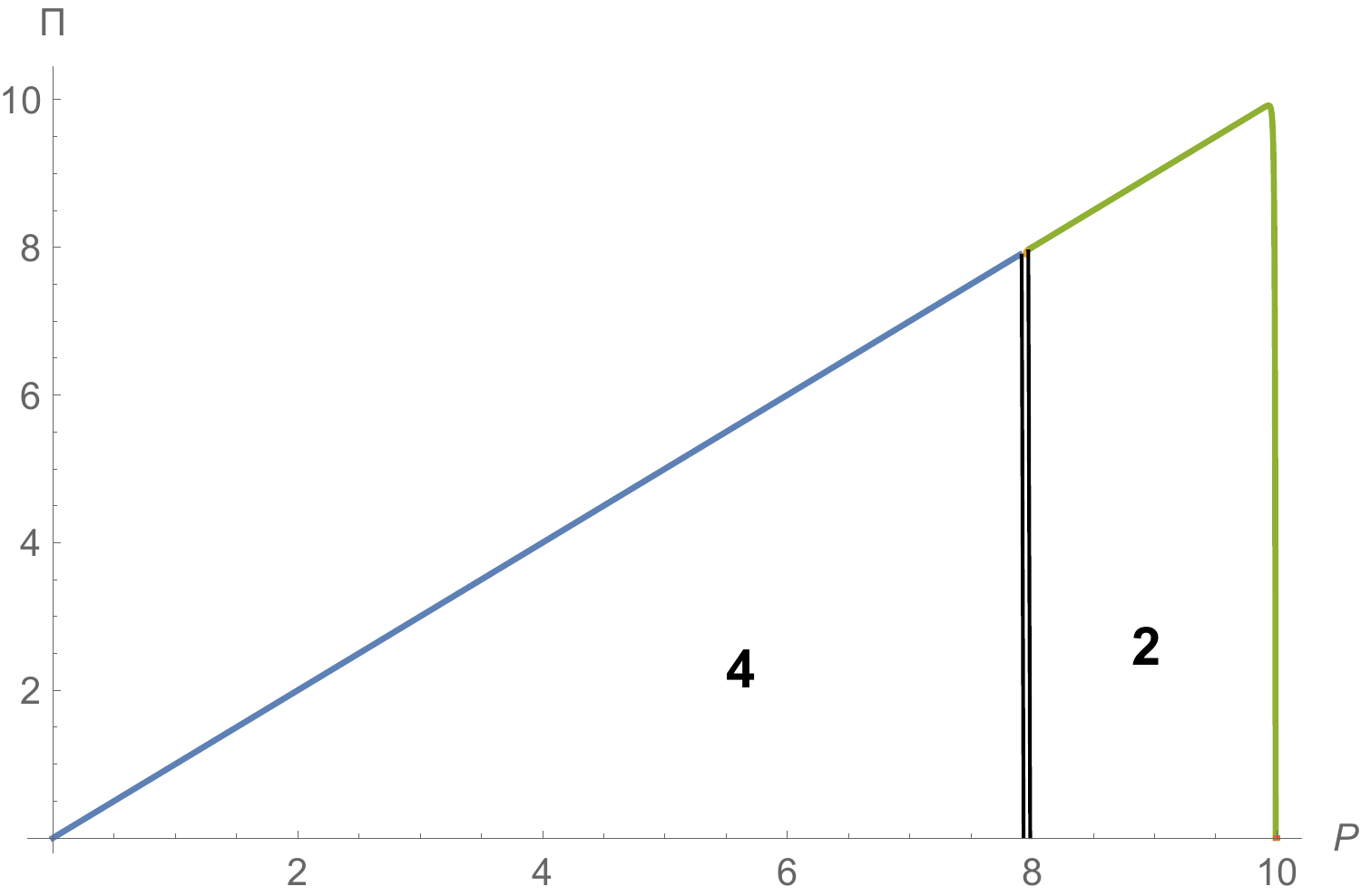}
\caption{\ \ $c=0.01$}\label{Fchalf}
\endminipage\hfill
\minipage{0.3\textwidth}
\includegraphics[width=\linewidth=0.2]{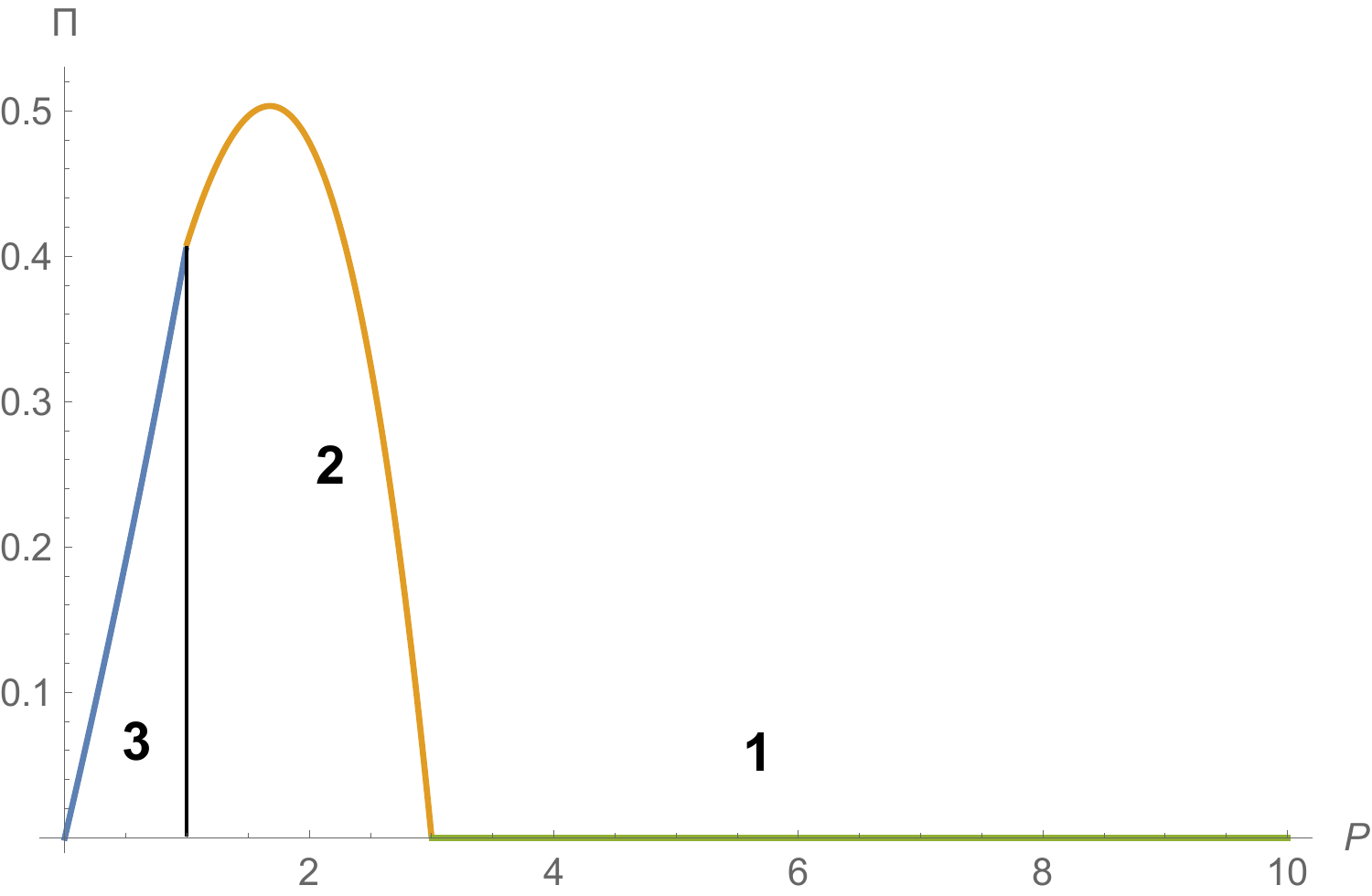}
\caption{\ \ $c=7$}\label{Fc5}
\endminipage\hfill\minipage{0.32\textwidth}
\includegraphics[width=\linewidth=0.3]{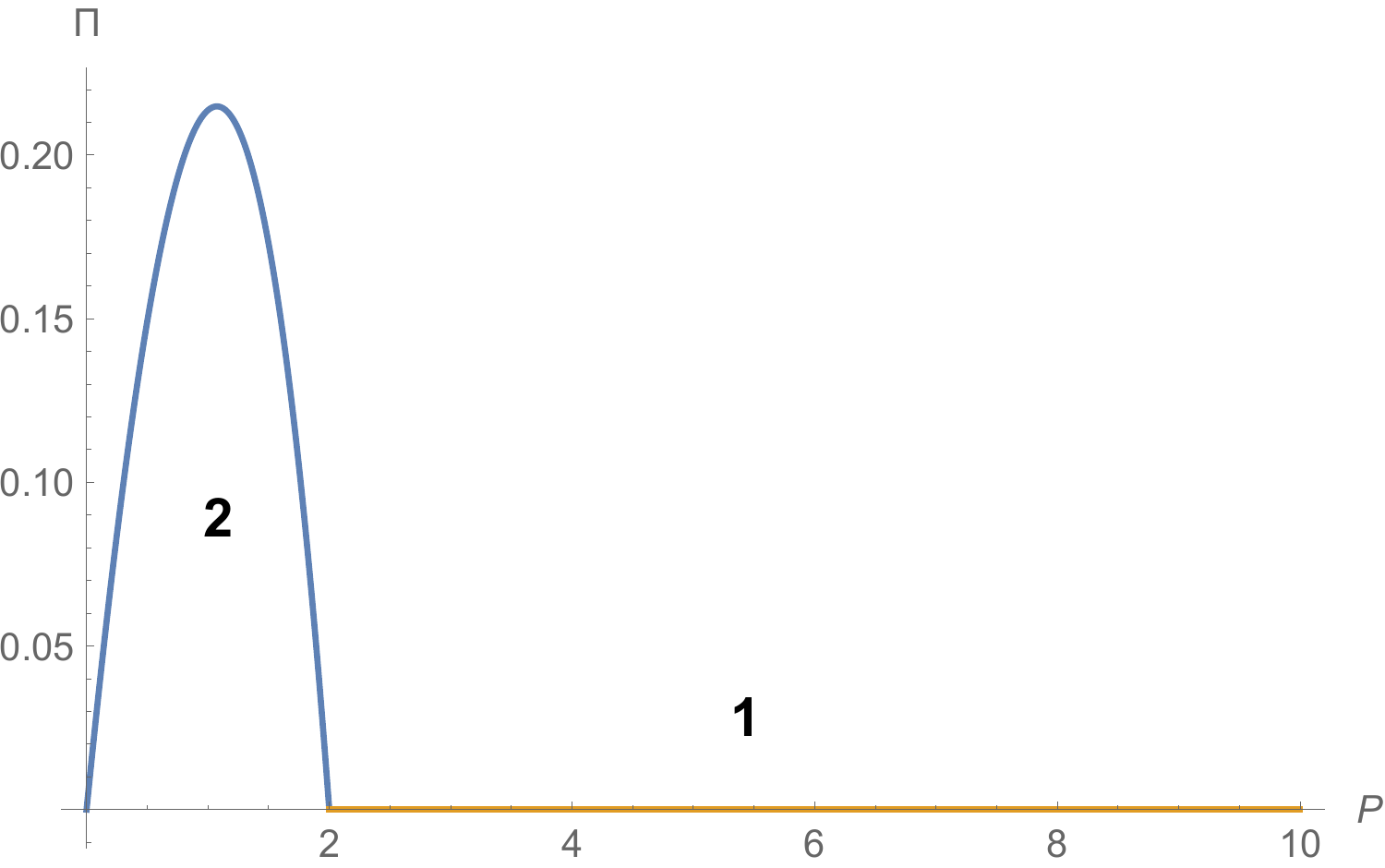}
\caption{\ \ $c=8$}\label{Fc6}
\endminipage\hfill
\end{figure}

The decrease of $\Pi^*,$ from $9.94$ when $c=0.01,$ to $0.216$ when $c=8$, is clearly shown.
It also leads to a gradual shift of early arrivals to later arrivals and finally to no arrivals.
Eliminating the arrival cost  also eliminates the region of mixed arrivals (see Figure~\ref{Fchalf}), so that depending on the value of the parameters either consumers arrive only at the time they most desire the good or only at the earlier time. Outcomes when the arrival cost is zero are relevant when people buy online. See more in Section~\ref{AEC} in the Appendix.

\section{Multiple units for sale}\label{Munits}

The firm has $N$ units of the good. All $N$ units are offered for sale in period 1; units not sold will be offered for sale again in period 2. Denote by $s$  the number of units unsold in period 1. We assume that consumers are not informed about the number of units that are left for sale in period 2. 
Any potential consumer can compute in advance the utilities of coming in period 1 or 2, and decide accordingly if and when to arrive.

We look for an equilibrium with a symmetric consumer strategy. A strategy consists of probabilities $q_1,q_2,$ where $q_t$ is the probability of arriving in period $t$, \
$t=1, 2.$
The number of arrivals in each case, $X_t$, is a random variable with Poisson distribution and parameters $\lambda_t=\lambda q_t$.  Let $U_t$ be the expected benefit to a consumer who arrives in period $t.$ Let $U_{2,s}$ be  the expected benefit to a consumer who arrives in period 2 when $s$ units are left for sale. Let $P_s$ be the probability that $s$ units are for sale in period 2.
 Then, \ \ $U_2=\sum_{s=0}^{N}P_sU_{2,s},$ and

\begin{equation}\label{EN10}
U_{1}=-c+(V-K-P)\left[\P(X_1\le N-1)+\sum_{i=N}^{\infty} \P(X_1=i){N\over i+1}\right].
\end{equation}
 For $s=1,\dots,N$
\begin{equation}\label{EN20}
U_{2,s}=-c+(V-P)\left[\P(X_2\le s-1)+\sum_{i=s}^{\infty} \P(X_2=i){s\over i+1}\right].
\end{equation}
For $s=0,$ \  $U_{2,0}=-c.$
As in the single-unit model, it is convenient to use normalized monetary values by considering $c$ as the unit value. Thus define
$v \equiv \frac{V}{c}$, $k \equiv \frac{K}{c}$, $p \equiv \frac{P}{c}$,  $\pi \equiv \frac{\Pi}{c}$, and  $u_t \equiv \frac{U_t}{c}$ \ $(t=1, 2).$

\subsubsection{Results of the model with a single unit hold for multiple units}\label{GER}

Much of the mathematical analysis becomes too difficult when $N>1,$ hence we used  {\it{Mathematica}} software for our investigations. All the figures in this section were executed for $k=2$, $v=10$ and $N=1,2,4,10, 50$ units. The figures illustrate that all the qualitative results proved for the single-unit case also hold for the multi-unit case. In addition, we solved many other instances of $v, k, N,$ numerically, and always observe the same behavior. Figure~\ref{FNEW8} presents $\l_1+\l_2$ as  a function of the price $p,$ for $N=1,2,4,10,$ (the graph for $N=50$ is not shown because it is in a different scale, but it has the same characteristics as $N=1,2,4,10$).

\begin{figure}[H]
\captionsetup{width=0.9\textwidth}
\begin{center}
\includegraphics[scale=0.4]{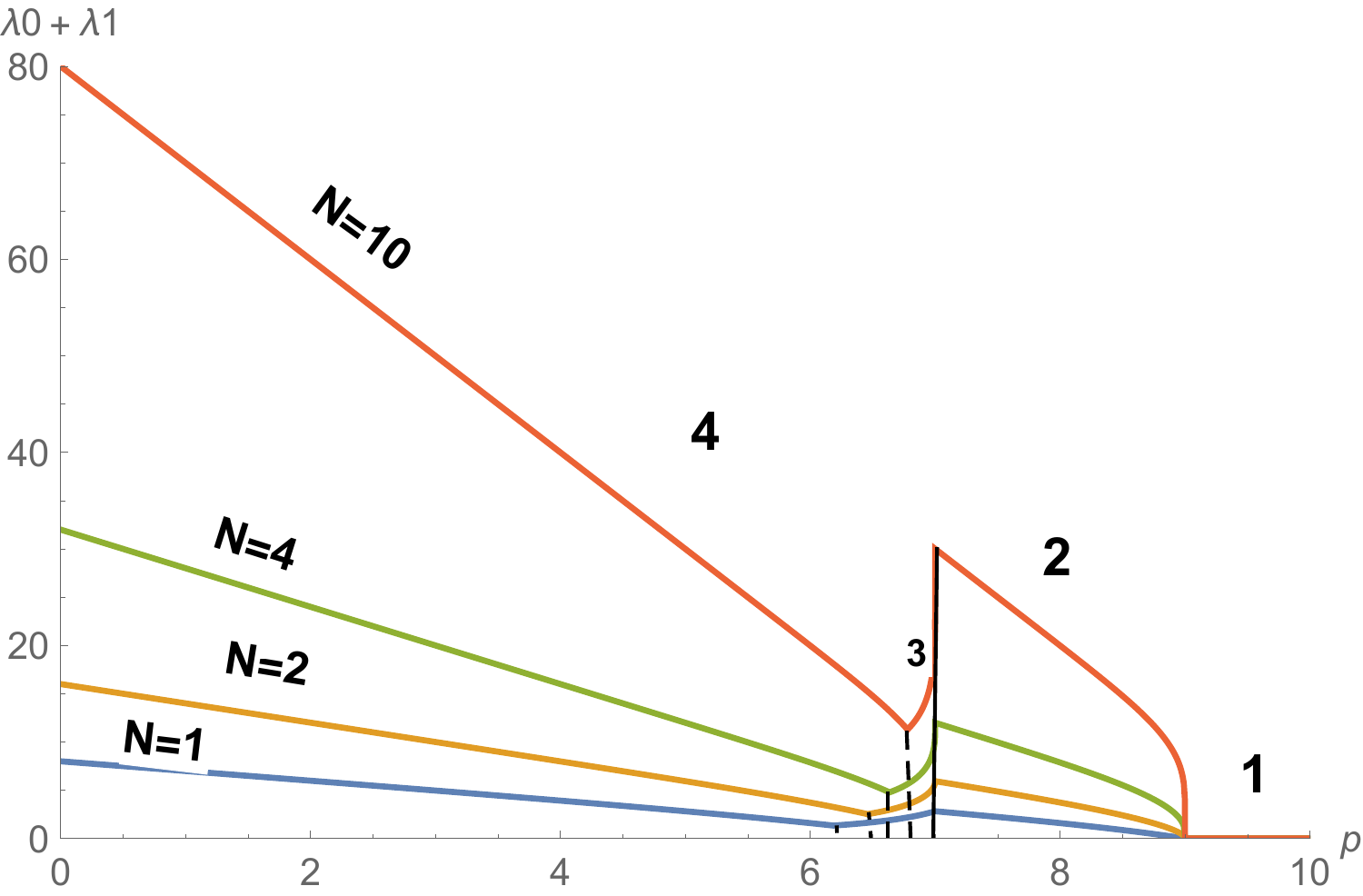}
\caption{ $\l_1+\l_2$ for $v=10$, $k=2$, rising from the lowest line for $N=1$ followed by $N=2,$ \ $N=4$ and concluding with the highest line for $N=10$. Vertical lines denote the borders of Region 3: dashed lines at the left border, which is different for different values of $N,$ and solid line at the right border $v-k-1=7$ for all values of $N.$}\label{FNEW8}
\end{center}
\end{figure}

As seen in the Figure, as in the single-unit case, the expected total number of arrivals $\l_1+\l_2$ decreases with the price in all cases except for the mixed arrival case (Region 3), where it increases with the price.
Figure~\ref{FNEW8}  also shows that as long as the price is not too close to $v-1$ the total expected arrival $\l_1+\l_2$ exceeds the number $N$ of units. For example, for the line of $N=10$ when $p<8.85<9=v-1,$ $\l_1+\l_2>10.$
Thus there is indeed an expected shortage. Proposition~\ref{Lambda2} in Section \ref{APN10} in the Appendix states that the expected number of arrivals when the price approaches  $v-1$ tends to the number of units for sale. Proposition~\ref{PN10} gives the expression for expected profits.

Figures~\ref{FNEW4}-\ref{FNEW7} below present $\pi$ for  $v=10, k=2$  for various values of $N.$ The Figures illustrate that  also for $N>1$ the profit increases monotonically in Region 3. It has at most one local maximum in Region 2 and at most one in Region 4. More details  are in Section~\ref{AGER} in the Appendix.

\begin{figure}[H]
\minipage{0.3\textwidth}
\includegraphics[width=\linewidth=0.2]{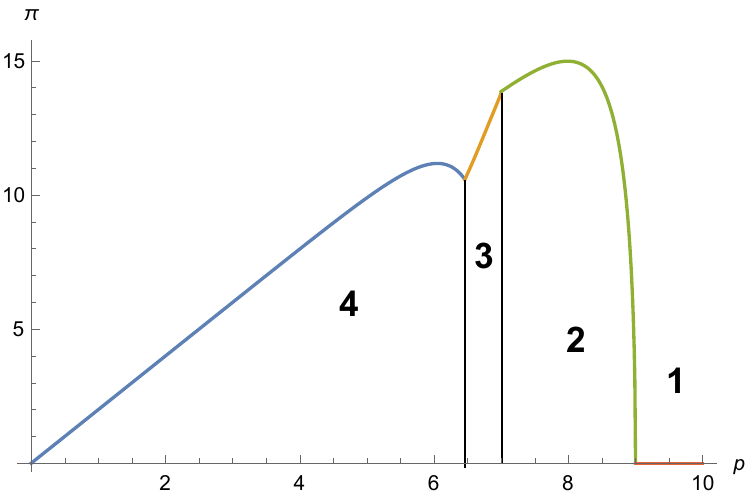}
\caption{\ \ $N=2$}\label{FNEW4}
\endminipage\hfill
\minipage{0.3\textwidth}
\includegraphics[width=\linewidth=0.2]{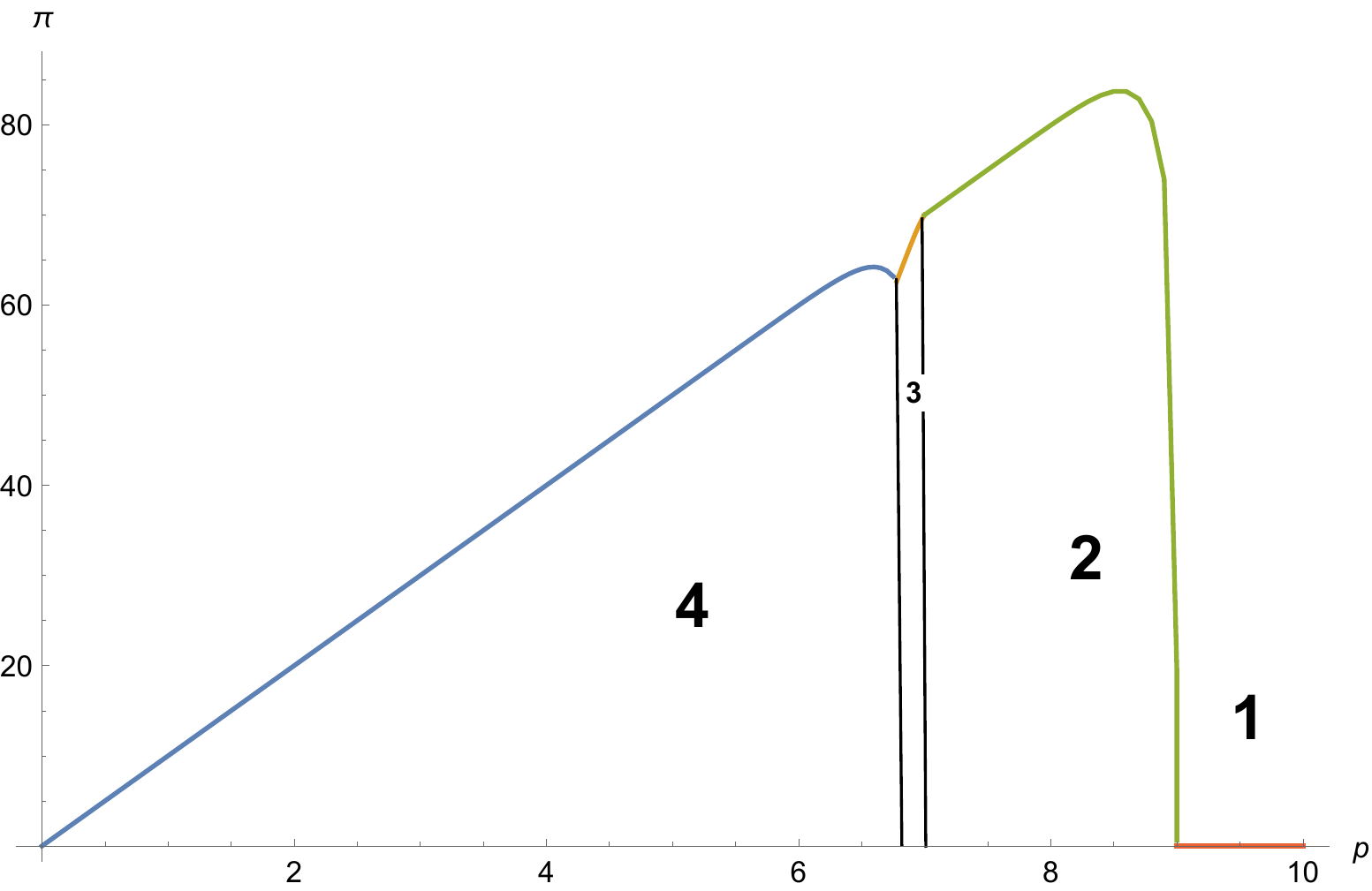}
\caption{\ \  $N=10$}\label{FNEW6}
\endminipage\hfill
\minipage{0.3\textwidth}
\includegraphics[width=\linewidth=0.2]{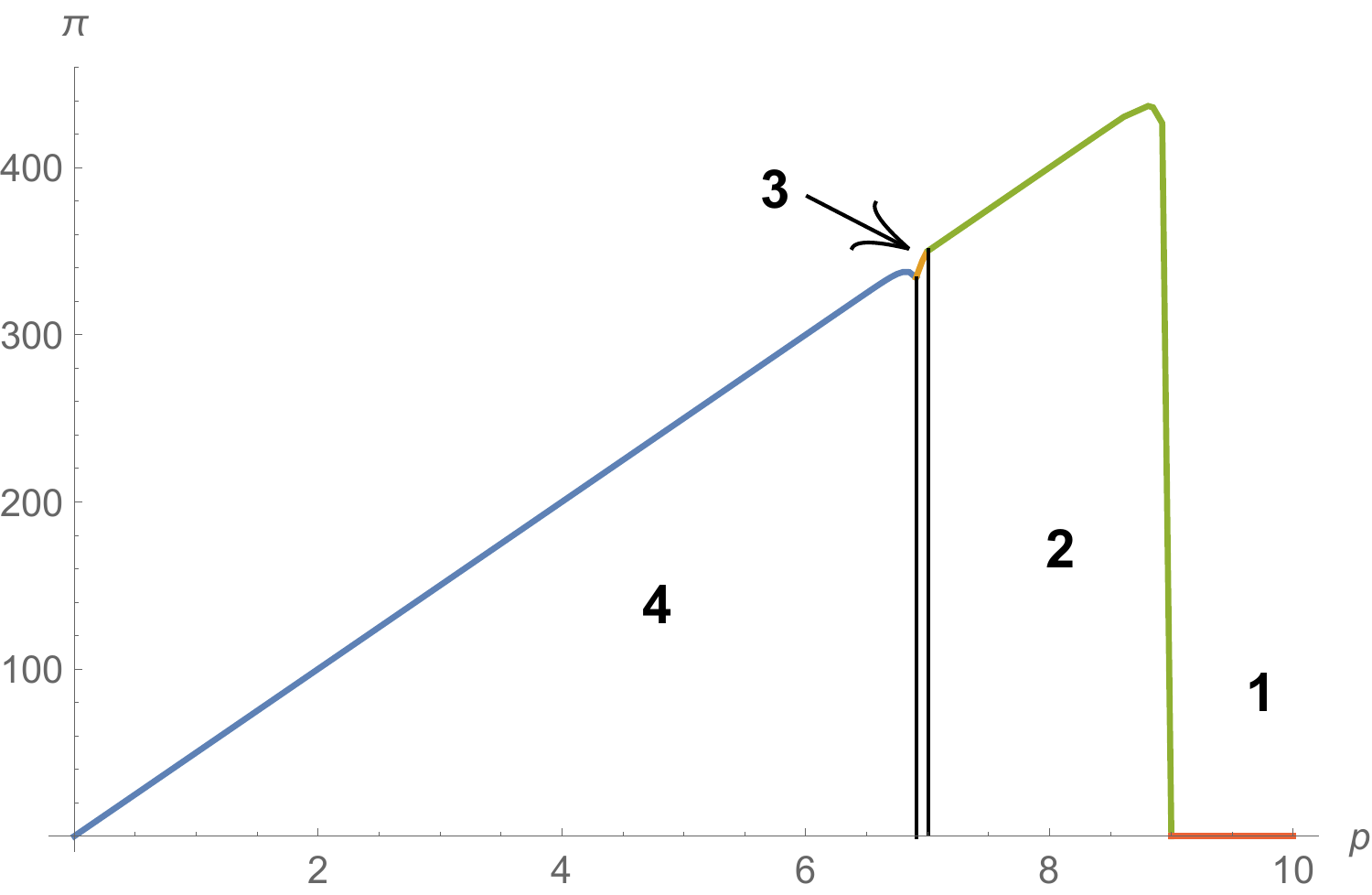}
\caption{\ \ $N=50$}\label{FNEW7}
\endminipage\hfill
\end{figure}

\subsection{The effect of supply on profit-maximizing price and customers' behavior}\label{NRES}

Figure~\ref{FNEW8} shows that at low prices (left side of Region 4) there are  more arrivals than at higher prices. Nevertheless, as we soon show, the maximum expected profit is obtained with the higher prices, where fewer arrivals are expected.

Surprisingly, we find that increased supply (an increase in $N$) enlarges the range of prices in which all consumers arrive early (Region 4), at the expense of the range in which consumers arrive in both periods (Region 3). This is clearly seen in Figures~\ref{FNEW4}-\ref{FNEW7}.  Thus, for a given price, increased supply induces consumers to arrive earlier. One may expect the opposite, that increased supply eases the fear of shortage, resulting in more consumers arriving later, in period 2. The reason is that a potential consumer understands that increased supply in period 1 induces more consumers to arrive in period 1, and thus no units will be left for period 2. This increases the incentive to arrive in period 1.
At the limit, (when $N$ goes to infinity) the range of mixed arrivals (Region 3) vanishes, and Region 4 becomes $[0, v-k-1]=[0, 7].$ Thus, at the limit, the equilibrium behavior of the consumers becomes similar to what occurs with the profit-maximizing price, in the sense that there are no mixed arrivals.

Another finding, (seen in  Figures~\ref{FNEW4}-\ref{FNEW7}) is that the maximum expected profit increases with $N$. Let $\pi^*_N$ be the expected profit when $N$ units are offered for sale in period 1.
The increase of $\pi^*_N$ as a function of $N$ is seen in Figure~\ref{FNEW9}, which shows $\pi^*_N$ as a function of $N$ for $v=10$ and $ k=2.$

\begin{figure}[h!]
\captionsetup{width=0.8\textwidth}
\begin{center}
\includegraphics[scale=0.3]{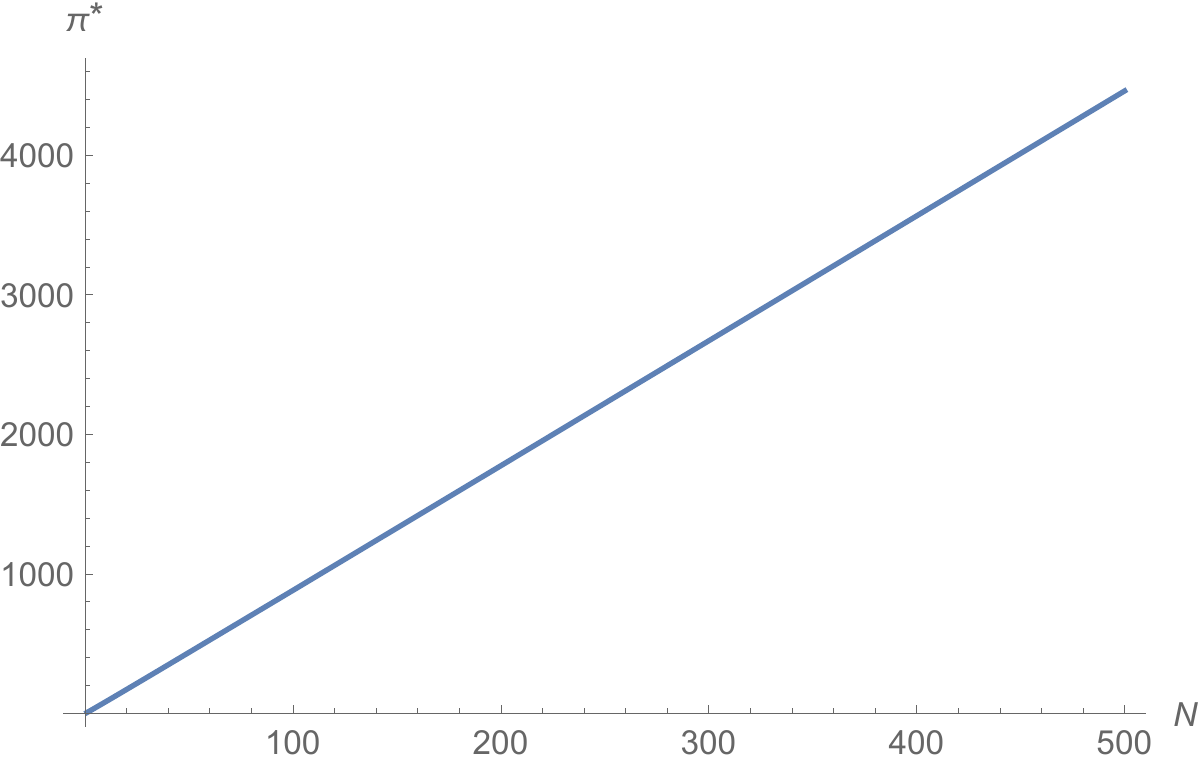}
\caption{$\pi^*_N$ as a function of $N.$ ($v=10, k=2$).
}\label{FNEW9}
\end{center}
\end{figure}

The increase in the profit-maximizing price with the number of units offered for sale is unexpected: an increase in supply usually reduces the price.  However, in our case, an increase in $N$ also increases demand (see Figure~\ref{FNEW8}).

In fact the increase, seen in Figure~\ref{FNEW9} seems to be quite linear. As we show next, for large $N$ it is indeed approximately linear and satisfies  $\pi^*_N\sim (v-1)N.$
The latter result is already demonstrated in Figure~\ref{FNEW7}, which presented $\pi$ for $N=50$. The maximal $\pi^*$ is obtained close to $p=v-1=9,$ with a value that is close to $(v-1)N=9\cdot 50=450.$

Let $\l_i (N)$ be the expected number of arrivals in period $i$ ($i=1,2$), when  $N$ units are for sale in period 1.

Proposition~\ref{Lambda2} in the Appendix proves that  when $p=v-1$
 $$\lim_{N\to\infty}\frac{\l_2(N)}{N}=1.$$ Thus, for large $N$ the expected number of arrivals when the price approaches  $v-1$ tends to the number  of units for sale.

In summary, the firm profits from increased supply, allowing it to charge a price approaching $v-1,$ which is the highest price at which it can expect to sell any units. Formally,

\begin{corollary}\label{Linear}
For all $v>1, k$,  $$\lim_{N\to\infty}\frac{\pi_N^*}{(v-1)N}=1.$$
\end{corollary}

Figure~\ref{FLIMIT} illustrates the result of Theorem~\ref{Linear} for $v=10, k=2.$

\begin{figure}[H]
\captionsetup{width=0.9\textwidth}
\begin{center}
\includegraphics[scale=0.35]{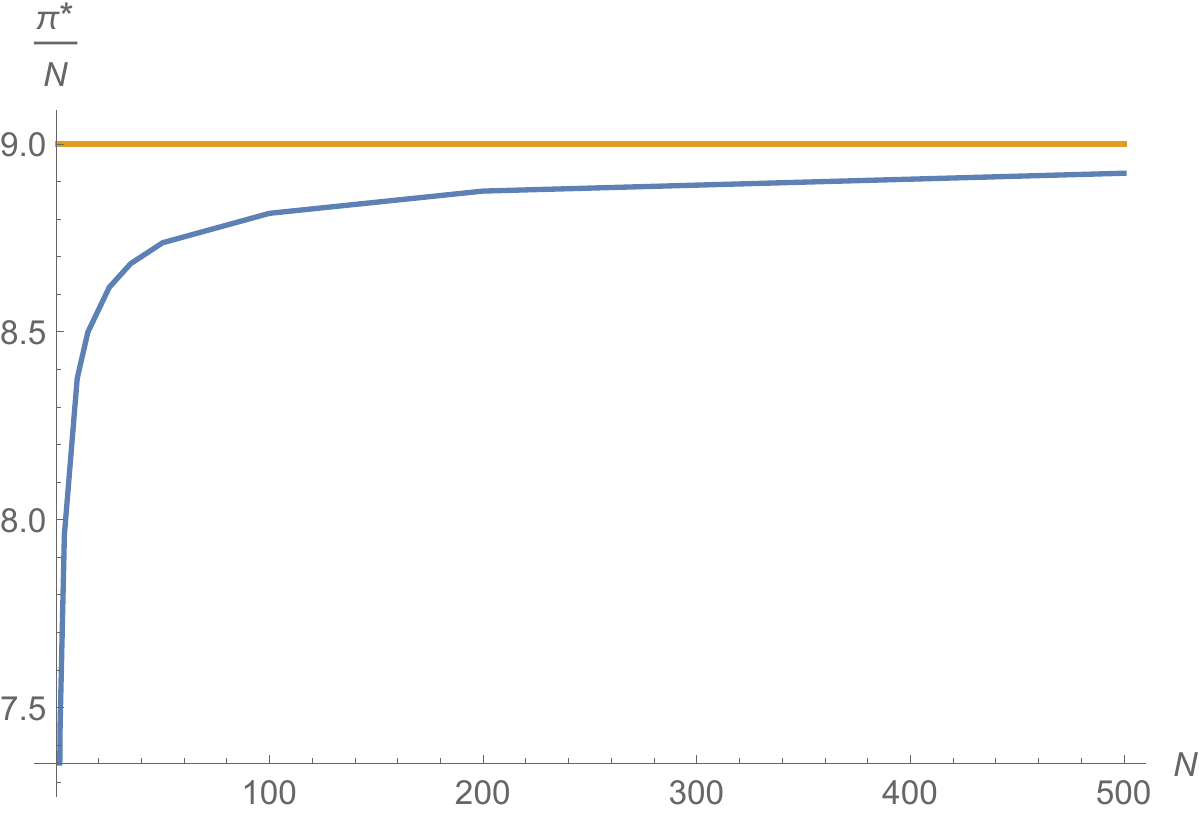}
\caption{ $\pi^* /N$ for $v=10,$ and $k=2.$ The horizontal line is $v-1=9$ for comparison.}\label{FLIMIT}
\end{center}
\end{figure}

\section{Conclusion}

We had considered a single seller and a single type of good. In general, each seller could have multiple goods. Even with multiple goods but one seller, the number of potential buyers can exceed the number of goods. Each consumer may then fear a shortage, and therefore may decide to arrive early. In turn, the seller should decide whether to offer the good in all periods or only in some, and what price to charge.

Consider next multiple sellers, supposing that a consumer can go to only one seller at a time. Similar questions appearing in the model with one seller arise here. A consumer may fear that if he does not get the good in period 1 from the seller he visits in period 1, then in period 2 he may find a shortage at whichever store he visits. But as we saw, examining even the case with one seller having one item is not trivial. We therefore considered only a single seller.

Under high demand and a single unit for sale if the penalty for arriving early is sufficiently large compared to the value $v$ of the good, the firm's normalized expected profit is $v-1-\ln{v}$. This is the best outcome for the firm. The firm never profits from offering the good early; for many parameter values offering early strictly reduces expected profits. This result also means that, even if there is no cost of making the product available early, the firm will not want to offer the good early. A large expected number of arrivals intensifies the fear of shortage, inducing most potential consumers to not arrive at all. The firm may erroneously think that offering the good in an early period may overcome this fear and  induce consumers to arrive early and have better chances to gain the good. Not so. Early opening further intensifies  the fear of shortage, since now a consumer considering
 coming in the period in which he most desires the good competes both with consumers who arrive at the same time and with consumers who come earlier.

We also find that if the firm must also offer the good early, but must charge the same price in both periods, then the profit-maximizing price induces consumers to arrive only in period  1, or under different conditions, induces them to arrive only in period 2. The firm would not set a price which induces consumers to arrive in both periods. In particular, if the firm  controls the penalty for arriving early, then it should set a high penalty so that no one arrives early.
The strategy of consumers when deciding if and when to arrive is more complicated than one may suppose, and can generate some unexpected behavior. For example, for parameter values which induce consumers to arrive in both periods the arrival rate in period 2  declines with the surplus a person gets from buying the good; indeed even the arrival rate summed over both periods  declines with that surplus.

Analyzing the effect of an early arrival penalty we find that a small increase in the penalty reduces the expected profit; but increasing it to the level that eliminates early arrivals brings the expected profit to its upper bound $v-1-\ln{v}$. An increase in the penalty beyond this point does not further affect the expected profit.

An increase in the arrival cost $c$ results in fewer expected arrivals and reduces profit. It also leads to a  shift from early arrivals to later arrivals and finally to no arrivals at all. Mixed arrivals are eliminated when there is no arrival cost or if the arrival cost  is high.

Generalizing for multiple units for sale we find that all the above results hold.
Surprisingly, an increase in the number of units offered for sale induces more consumers to arrive earlier rather than later.  One may expect the opposite, that increased supply will ease the fear of shortage, resulting in more consumers arriving in period 2. The behavior arises because a potential consumer understands that an increase in supply in period 1 induces more consumers to arrive in period 1; fewer units will be left for period 2 and thus a consumer better arrive in period 1 as well.
Additionally, the profit-maximizing price increases with the number of units offered for sale. This too is unexpected, as an increase in supply often results in price reduction.  However, in our case, an increase in supply also increases demand, and the seller may profit by increasing the price. We prove that with a large supply the profit-maximizing price approaches $v-1$.
Therefore ignoring the cost to the firm of producing more units, the firm's profits increase with supply.

\pagebreak

\section{Appendix}\label{App}

\begin{lemma}\label{L1}
 There is at most one equilibrium of each type.
\end{lemma}
\begin{proof} Note that an equilibrium is determined by the values of $(q_1, q_2)$ (or equivalently, by the values of $(\l_1, \l_2)$). Therefore, the claim is straightforward for equilibrium types 1, 6, and 7, where $q_1$ and $q_2$ are explicitly defined.
The type-3 equilibrium (where a consumer may choose never to arrive) requires that $U_{1}=0$, which uniquely defines $\l_1$ since the right-hand side of~\eqref{E9000} is monotonic in $\l_1.$
Substituting the resulting value in $U_2=0$ uniquely defines the candidate $\l_2$. For the type-2 equilibrium, which assumes $\l_1=q_1=0$, $\l_2$ follows as above. Similarly, the type-4 equilibrium already assumes that $q_2=0$, leading to $\l_{1}=0$, with $\l_1$ uniquely determined by $U_{1}=0$.
The type-5 equilibrium (where a consumer is indifferent between arriving in periods 1 and 2) requires that $q_1+q_2=1$, or equivalently that $\l_1+\l_2=\l$. Substituting this condition into the other requirement that $U_{1}=U_2$ gives ${V\over V-K}={e^{\l_1}-1\over\l_1}{1-\l_1\over1-e^{-(1-\l_1)}}$. The right-hand side is a monotonic function of $\l_1$. Hence, $\l_1$  is uniquely defined (and corresponds to an equilibrium).
\end{proof}

\subsection{Characterizing the equilibria}\label{SCHI}

Define $x \equiv \l_1$ as the expected number of people arriving in period 1; and define $y \equiv \l_2$ as the expected number of people arriving  in period 2.

\noindent

\begin{enumerate}

\item
The type-1 equilibrium has $\l_1=\lambda_2=0$. The utilities $u_{1}$ and $u_{2}$ will both be non-positive only if $v\le 1$.

\item
The type-2 equilibrium has $\l_1=0,$  and $0 \le \l_2 \le \l.$  The condition $u_{2}=0$ is equivalent to $v={y\over1-e^{-y}}$, implying that $v \ge 1$. The inequality $u_{1}\le 0$ means that $v-k \le 1$.

\item
The type-3 equilibrium has $0 \le \l_1$ and  $\l_2 \le \l$. The condition $u_{1}=0$ implies  that $v=k+{x\over 1-e^{-x}}$, and $u_{2}=0$ implies that $v=e^x{y\over1-e^{-y}}$.

\item
The type-4 equilibrium has $\l_2=0$ and $0 \le \l_1 \le \l$. The condition $u_{1}=0$ implies that $v=k+{x\over 1-e^{-x}}$. And $u_{2} \le 0$ (with $q_2=0$) amounts to $v \le e^x$.

\item
The type-5 equilibrium has $\l_1+\l_2 = \l$. Then $u_{1}=u_{2}$ reduces to
${k\over v}=1-{x\over 1-x}{1-e^{-(1-x)}\over e^x-1}$.
The non-negativity of $u_{2}$ reduces to $v\ge e^{x}{1-x\over1-e^{-(1-x)}}$.

\item
The type-6 equilibrium has $\l_1=\l$ and $\l_2=0$.  The conditions
$u_{1}\ge u_{2}$ and $u_{1}\ge0$ reduce to $v \ge {e-1\over e-2}k$, and $v \ge k+{e\over e-1}$.

\item The type-7 equilibrium has $\l_1=0$ and $\l_2=\l$.  The conditions
$u_{2}\ge u_{1}$ and $u_{2}\ge0$ reduce to $k\ge{v\over e}$ and $v\ge{e\over e-1}$.

\end{enumerate}

\subsection{High demand }\label{AFIN}
\subsubsection{Equilibria types}\label{AP66}

\begin{proposition} \label{P66}
For any pair $(v,k)$ there exists a unique equilibrium $(\l_1,\l_2).$ This equilibrium is one of the types 1-4, with small  $\l_1, \l_2,$ such that the utilities are still non-negative.
\end{proposition}

\begin{proof}
Note first that equilibrium types 5-7 are not possible when $\l$ is large. For example, in equilibrium type 5 $\l=\l_1+\l_2$. Hence either $\l_1$ or $\l_2$ (or both) are also unbounded. By L'H\^{o}pital's rule, $\lim\limits_{\l\to\infty}{1-e^{-\l}\over\l}=0$. Hence according to \eqref{E9000N} and \eqref{E9001N}
either $u_{1}$ or $u_2,$ (or both) equals $-c,$ contradicting the requirement that $u_{1}, u_2\geq 0$. Similarly, one can verify that when $\l$ is unbounded equilibrium types 6 and 7 are not possible. This is intuitive, as explained in the main text: if many consumers are expected to arrive in one of the periods, then the expected gain of an arriving consumer becomes smaller than the arrival cost $c,$ and thus the utility becomes negative. So the only equilibrium types that are possible when $\l$ is unbounded are types 1-4 and only with bounded (i.e., sufficiently small) values of $\l_1, \l_2$. Namely, `most' potential consumers do not arrive at any time.

The conditions for types 1 and 2 appear in the list (see the beginning of Section~\ref{SCHI}). Namely, for each pair $(v,k)$ in the region where $v\leq 1$ (which will be called Region 1), there exists a type 1 equilibrium. If $1\leq v\leq k+1$ (which will be called Region 2), there exists a type 2 equilibrium. Now, for type 3 we have $v=k+{x\over 1-e^{-x}}.$ Define the pairs $(v,k)$ satisfying the conditions of type 3 as Region 3.  Since ${x\over 1-e^{-x}}$ increases with $x,$ and equals 1 when $x$ goes to zero,   $k+1\leq v$. Hence Region 3 does not overlap with Regions 1 and 2.

By~\eqref{E9001N} in Region 3
\begin{equation}\label{E77}
v=e^x{y\over1-e^{-y}}.
\end{equation}

Since ${y\over1-e^{-y}}$ increases with $y$ and approaches 1 when $y$ approaches 0,  it follows from (\ref{E77}) that in this region $e^x\leq v.$ Since $\frac{\ln{v}}{1-\frac{1}{v}}$ increases with $v$ it follows that
$$ {\ln{e^x}\over 1-\frac{1}{e^{x}}}\leq \frac{\ln{v}}{1-\frac{1}{v}}.$$
Namely,
$$ {x\over 1-e^{-x}}\leq \frac{\ln{v}}{1-\frac{1}{v}}.$$
And so
\begin{equation}\label{E79}
v-{x\over 1-e^{-x}}\geq v-\frac{\ln{v}}{1-\frac{1}{v}}.
\end{equation}
By~\eqref{E9000N} in Region 3
\begin{equation}\label{E78}
v-{x\over 1-e^{-x}}=k.
\end{equation}
Substituting~\eqref{E78} in the left-hand side of~\eqref{E79} gives
\begin{equation}\label{E800}
k\geq  v-\frac{\ln{v}}{1-\frac{1}{v}}.
\end{equation}
 Note that the right-hand side of \eqref{E800}  strictly increases from $0$ to infinity; thus for each $k$ there exists a $v=v(k)$ such that
$$k=v(k)-\frac{\ln{v(k)}}{1-\frac{1}{v(k)}}.$$
Substituting this in~\eqref{E800} gives
$$v(k)-\frac{\ln{v(k)}}{1-\frac{1}{v(k)}}\geq v-\frac{\ln{v}}{1-\frac{1}{v}}.$$
Since $v-\frac{\ln{v}}{1-\frac{1}{v}}$  increases with  $v>1,$
then $v(k)\geq v.$
Similarly, for Region 4, $v\geq v(k).$ The proof of that is very similar to the one given for Region 3, only now $v\leq e^x$ (instead of $v\geq e^x$). That yields the opposite result. Thus Regions 3 and 4 are separated by the line $(v(k),k).$

Because the regions do not overlap, each point $(v, k)$ has only one type of equilibrium. Combining this with Lemma~\ref{L1}, that claimed that any $(v, k)$ has at most one equilibrium of each type, we get the uniqueness of the equilibrium.
\end{proof}

\subsubsection{Equilibrium arrival rates as functions of the net benefit $v-p$ of the purchase}\label{ASAR}

\begin{proposition}\label{P1}
In equilibrium, the arrival rates $\l_1$ in period 1 and $\lambda_2$ in period 2 satisfy:
\begin{itemize}

\item In Region 1, $\l_1=\l_2=0$.

\item In Region 2, $\l_1=0$, and $\l_2$ increases with $v-p$.

\item In Region 3, $\l_1,\l_2>0$, \ $\l_1$ increases with $v-p$,  and $\l_2$ declines with $v-p$.

\item In Region 4,  $\l_1$ increases with $v-p,$ and $\l_2=0.$
\end{itemize}
\end{proposition}

\begin{proof}
Recall that $x \equiv \l_1,\  y \equiv \lambda_2$.

\begin{enumerate}
\item The proof of the first statement of the proposition follows immediately from the definition of Region 1.
\item In Region 2, by definition, $\l_1=0,$ and $u_2=0$, so that (\ref{E9001N}) gives
$$v-p=\frac{y}{1-e^{-y}}.$$
Hence $$\frac{d}{dy}(v-p)= \frac{1-e^{-y}-ye^{-y}}{(1-e^{-y})^2}.$$
Because $y+1\leq e^y$ for $y\geq 0$ the numerator on the right-hand side of the above satisfies $1-e^{-y}(1+y)\geq 0$. Thus $(v-p)' \geq 0$ and $v-p$ increases with $y$. Hence $y$ increases with $v-p$, proving the second claim of the proposition.

\item In Region 3, by definition,  $\l_1$ and $\l_2$ are both positive, and $u_1=u_2=0$.
Similarly to the proof for $y$ in Region 2, it follows from~(\ref{E9000N}) that in Region 3, given $k,$ \  $x$ increases with $v-p$. From~(\ref{E9000N}) and~(\ref{E9001N}) together we have
$$k+\frac{x}{1-e^{-x}}=\frac{ye^x}{1-e^{-y}}.$$
So
\begin{equation}\label{E3}
ke^{-x}+\frac{x}{e^{x}-1}=\frac{y}{1-e^{-y}}.
\end{equation}
The left-hand side of (\ref{E3}) declines with $x \equiv \l_1$ since
$$
\Big( ke^{-x}+\frac{x}{e^{x}-1}\Big)'=-ke^{-x}+\frac{e^x-1-xe^x}{(e^x-1)^2}=-ke^{-x}-\frac{e^x(x-1)+1}{(e^x-1)^2}.
$$
The expression $e^x(x-1)+1$ appearing in the numerator on the right-hand side increases with $x$ and equals $0$ when $x=0$. Hence it is positive and so $-ke^{-x}-\frac{e^x(x-1)+1}{(e^x-1)^2}<0$, implying that the left-hand side of~(\ref{E3}) indeed decreases with $x$. In contrast, the right-hand side of~(\ref{E3}) increases with $y \equiv \l_2$. Thus $y$ decreases with $x $. Because $x \equiv \l_1$ increases with $v-p,$ it follows that $y \equiv \l_2$ declines with $v-p$, proving the third  claim of the proposition.

\item In Region 4, by definition, $y \equiv \l_2 =0$ and $u_1=0$. 

So, as in Region 3, $x \equiv \l_1$  increases with $v-p$, proving the last statement of the proposition.
\end{enumerate}
\end{proof}

\subsubsection{Computing the equilibrium arrival rates}\label{ASARV}
Consider next the equilibrium $(\l_1, \l_2)$ in each of the four regions. It turns out that these values involve the {\it Lambert function} $W[a],$ (also called the omega function or product logarithm). The Lambert function is  the inverse of the function $ze^z$. Namely, if $a=ze^z,$ then $W[a]=z.$
Because $ze^z$ is not injective, we add the restriction that $W[a]\geq -1,$ (without this restriction, $W$ would be multivalued). Note that $W[a]$ is defined only for $a \geq -e^{-1}.$ The Lambert function is presented in Figure~\ref{F10}.

\begin{figure}[h]
\begin{center}
\includegraphics[scale=0.6]{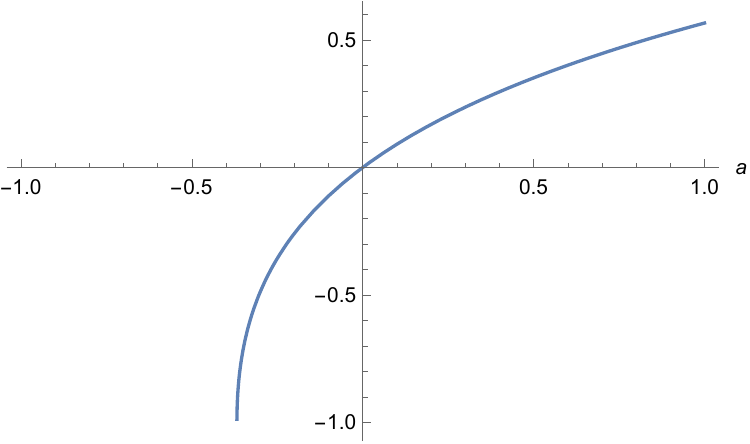}
\caption{The graph of $W[a]$  for $-\frac{1}{e}\leq a\leq 1$.}\label{F10}
\end{center}
\end{figure}

\begin{definition}\label{DW}
For all  $a \geq -e^{-1}$: $W[\cdot] \ \textit{is the inverse function of}  \ ae^{a},$
satisfying $W[a]\geq -1$.
\end{definition}

The following properties of $W$ will be used below. For proofs see Corless et al. (1996).
\begin{itemize}
\item W1. $W[\cdot]$ is an increasing function for all  $a\geq -e^{-1}$.
\item W2. $W[-e^{-1}]=-1$.
\item W3. $W'[a]=\frac{W[a]}{a(W[a]+1)}$.
\item W4. $W[0]=0.$
\item W5. $e^{-W[a]}=\frac{W[a]}{a}$.
\item W6. $W[a e^a]=a$ for all $a\geq -1$.
\end{itemize}

\begin{definition}\label{D2}
For all $a,$
$$R(a) \equiv W[-a e^{-a}].$$
\end{definition}

\begin{definition}\label{DA}
Let $v_1=v-p-k,$ and \  $v_2=v-p$. Then $v_1$ and $v_2$  are a consumer's benefit when buying  at price $p$ in periods 1 and 2 respectively (not taking into account the normalized arrival cost $c$).
Let
$$A=-(v_1+k)e^{-v_1+\frac{k}{v_1}R(v_1)}.$$
\end{definition}

\subsection{Theorem~\ref{T10}}\label{PP3}

\begin{theorem}\label{T10} The unique equilibrium $(\l_1, \l_2)$ is
\
\begin{itemize}
\item $(0 , 0)$ if $(v-p,k)$ is in Region 1.
\item $\Big(0 \ , v_1+k+R(v_1+k)\Big)$ if $(v-p,k)$ is in Region 2.
\item $\left(v_1+R(v_1) \ , W[A]-\left( 1+\frac{k}{v_1} \right)R(v_1)\right)$ if $(v-p,k)$ is in Region 3.
\item $\Big(v_1+R(v_1) \ , 0\Big)$ if $(v-p,k)$ is in Region 4.
\end{itemize}
\end{theorem}

Note that, by Property W6, for all $a \leq 1$ (which is equivalent to $-a\geq -1$)  $R(a)=-a$. But according to Table 1, in Region 2 $v_1+k=v-p>1,$ and in Regions 3 and 4 $v_1=v-p-k>1,$ so that we cannot exploit this attractive property.

\begin{proof}
Since subtracting $p$ from $v$ is just moving to the left of $v,$ we will ignore $p$. The statement in Theorem~\ref{T10} becomes

\begin{itemize}
\item $(0 , 0)$ if $(v,k)$ is in Region 1.
\item $\Big(0 \ , v_1+k+R(v_1+k)\Big)$ if $(v,k)$ is in Region 2.
\item $\left(v_1+R(v_1) \ , W[A]-\left( 1+\frac{k}{v_1} \right)R(v_1)\right)$ if $(v,k)$ is in Region 3.
\item $\Big(v_1+R(v_1) \ , 0\Big)$ if $(v,k)$ is in Region 4.
\end{itemize}

\begin{itemize}
\item The first statement follows immediately from the definition of Region 1 (see Table 1).
\item We wish to prove that the equilibrium value of $\l_2$ in Region 2  equals $v_1+k+R(v_1+k).$
Recall that by \eqref{E9000N} in Region 2 
\begin{equation}\label{E40}
v_1+k=\frac{\l_2e^{\l_1}}{1-e^{-\l_2}}.
\end{equation}
Since in Region 2 $\l_1=0,$ \eqref{E40} becomes
\begin{equation}\label{E41}
v_1+k=\frac{\l_2}{1-e^{-\l_2}}.
\end{equation}

First, we show that $\l_2=v_1+k+R(v_1+k)$ indeed solves (\ref{E41}).

Substituting $\l_2=v_1+k+R(v_1+k)$ in the right-hand side of (\ref{E41}) gives
\begin{equation}\label{E4}
\frac{v_1+k+R(v_1+k)}{1-e^{-(v_1+k)}e^{-R(v_1+k)}}.
\end{equation}

By Property W5
\begin{equation}\label{E790}
e^{-R(v_1+k)}=\frac{R(v_1+k)}{-(v_1+k)e^{-(v_1+k)}}.
\end{equation}
Substituting this in (\ref{E4}) gives
$$\frac{v_1+k+R(v_1+k)}{1+\frac{R(v_1+k)}{v_1+k}}=v_1+k,$$
which is the left-hand side of (\ref{E41}).

Now, because the right-hand side of (\ref{E41}) is strictly monotonic in $\l_2$, then for any given pair $(v_1,k)$ the value $\l_2=v_1+k+R(v_1+k)$ \emph{uniquely} solves (\ref{E41}).

\item To find $\l_1$ in Regions 3 and 4 recall that $u_1=0$ there. Thus by~\eqref{E9000N} 
\begin{equation}\label{E51}
v_1=\frac{\l_1}{1-e^{-\l_1}}.
\end{equation}
Note that substituting $k=0$ in~\eqref{E41} and changing $\l_2$ to $\l_1$ yields  \eqref{E51}. Hence the proof given earlier for $\l_2$ in Region 2 applies for $\l_1$ in Regions 3 and 4, and gives $\l_1=v_1+R(v_1).$
\item Lastly, we have yet to prove that the equilibrium value of $\l_2$ in Region 3 is $W[A]-\left( 1+\frac{k}{v_1} \right)R(v_1).$

Recall that we consider $c$ as the unit value. Recall also that $v_1=v-p-k,$ and that $v_2=v-p.$ Note that given $v$ and $k$, the values $v_1$ and $v_2$ are functions of $p$. We find that most of the expressions appearing in our analysis are functions of $v-p-k.$ Thus we focus
 on $v_1.$ The four regions in terms of $v_1$ are
\begin{itemize}

\item Region 1 is reached  iff \  $-k\leq v_1< 1-k$.
\item Region 2  is reached iff \  $1-k\leq v_1\leq 1$.
\item Region 3 is reached iff \  $1<v_1< v(k)-k$.
\item Region 4  is reached iff \  $v(k)-k\leq v_1$,
\end{itemize}
where $v(k)$ was defined as the point separating Regions 3 and 4 (when the regions were expressed in terms of $v$).

What is left to prove is that the equilibrium value of $\l_2$ in Region 3 is $W[A]-\left( 1+\frac{k}{v_1} \right)R(v_1).$

Recall that
$A=-(v_1+k)e^{-v_1+\frac{k}{v_1}R(v_1)}.$

Using the two central equations~\eqref{E9000N} and~\eqref{E9001N} that hold in Region 3 and appear in the main text we have
\begin{equation}\label{E01}
v_1=\frac{\l_1}{1-e^{-\l_1}},
\end{equation}
and
\begin{equation}\label{E02}
v_2=\frac{\l_2e^{\l_1}}{1-e^{-\l_2}}.
\end{equation}

We already proved that in Regions 3 and 4 the equilibrium value of $\l_1$ is  $v_1+R(v_1),$ and that it uniquely solves (\ref{E01}). Substituting this in  (\ref{E02}) gives
$$v_2=\frac{\l_2e^{v_1+R(v_1)}}{1-e^{-\l_2}},$$
which is the same as
\begin{equation}\label{E5}
v_2e^{-v_1-R(v_1)}=\frac{\l_2}{1-e^{-\l_2}}.
\end{equation}
Since the right-hand side of (\ref{E5}) is monotonic in $\l_2$, for any given pair $v_1, v_2$ at most one value of $\l_2$  satisfies (\ref{E5}). We now show that the proposed solution $\l_2=W[A]-\frac{v_2R(v_1)}{v_1}$ satisfies~(\ref{E5}). Substituting this $\l_2$ in the right-hand side of   (\ref{E5}) gives
$$\frac{W[A]-\frac{v_2R(v_1)}{v_1}}{1-e^{\frac{v_2R(v_1)}{v_1}}e^{-W[A]}}.$$
Because of Property W5 this expression equals 
$$\frac{W[A]-\frac{v_2R(v_1)}{v_1}}{1-\frac{e^{\frac{v_2R(v_1)}{v_1}}W[A]}{A}},$$ which equals 
$$\frac{\Big(W[A]-\frac{v_2R(v_1)}{v_1}\Big)A}{A- e^{\frac{v_2R(v_1)}{v_1}}W[A]}.$$
Since $v_2=v_1+k$ the above  equals
$$\frac{\Big(\frac{v_2R(v_1)}{v_1}-W[A]\Big)v_2e^{-v_1+\frac{kR(v_1)}{v_1}}}{-v_2e^{-v_1+\frac{kR(v_1)}{v_1}}- e^{\frac{v_2R(v_1)}{v_1}}W[A]}$$
$$=\frac{\Big(\frac{v_2R(v_1)}{v_1}-W[A]\Big)e^{R(v_1)+\frac{kR(v_1)}{v_1}}}{-v_2e^{-v_1+\frac{kR(v_1)}{v_1}}- e^{\frac{v_2R(v_1)}{v_1}}W[A]}\cdot v_2e^{-v_1-R(v_1)}$$

Hence we need to prove that the quotient above equals 1, namely that
\begin{equation}\label{E70}
\Big(\frac{v_2R(v_1)}{v_1}-W[A]\Big)e^{R(v_1)+\frac{kR(v_1)}{v_1}}=-v_2e^{-v_1+\frac{kR(v_1)}{v_1}}- e^{\frac{v_2R(v_1)}{v_1}}W[A].
\end{equation}
Note that the expression $e^{R(v_1)+\frac{kR(v_1)}{v_1}}$ appearing on the left-hand side satisfies
$$
e^{R(v_1)+\frac{kR(v_1)}{v_1}}=e^{\big(\frac{v_1+k}{v_1}\big)R(v_1)}=e^{\frac{v_2R(v_1)}{v_1}}.
$$
Substituting this in the left-hand side of~(\ref{E70}) gives
$$\frac{v_2R(v_1)}{v_1}e^{\frac{v_2R(v_1)}{v_1}}-e^{\frac{v_2R(v_1)}{v_1}}W[A].$$
Hence we only need to prove that
$$\frac{v_2R(v_1)}{v_1}e^{\frac{v_2R(v_1)}{v_1}}=-v_2e^{-v_1+\frac{kR(v_1)}{v_1}}.$$
This is equivalent to proving that
$$
\frac{R(v_1)}{v_1}e^{\frac{v_2R(v_1)}{v_1}+v_1-\frac{kR(v_1)}{v_1}}=-1.
$$
Note that the left-hand side of the above equation equals
$$\frac{R(v_1)}{v_1}e^{\frac{(v_2-k)R(v_1)}{v_1}+v_1}=\frac{R(v_1)}{v_1}e^{R(v_1)+v_1}.$$
By~(\ref{E790}) the above equals
$$\frac{R(v_1)}{v_1}\frac{(-v_1e^{-v_1})e^{v_1}}{R(v_1)}=-1.$$

\end{itemize}
\end{proof}

\subsection{Proof of Theorem~\ref{T1}}\label{PT1}
To prove Theorem~\ref{T1}, we first need to establish several results, using the Lambert function $W[x],$ (see Definition~\ref{DW} in Section~\ref{SAR}).  We also presented  a list of W1-W6 properties of the Lambert function, that we use in our analysis below (see Section~\ref{SAR}).

Recall that we denote
$$R(a)=W[-a e^{-a}].$$

The following lemma proves essential properties of the function $R(v_1)$.

\begin{lemma}\label{T2}
\
\begin{itemize}
\item R1. \ $R(v_1)$ is negative and increasing for all $ v_1>1$.
\item R2. \ $R(1)=-1$.
\item R3. \ $R(v(k)-k)=\frac{k}{v(k)}-1$.
\item R4. \ $R'(v_1)=\frac{-R(v_1)}{R(v_1)+1}\cdot\frac{v_1-1}{v_1}.$
\end{itemize}
\end{lemma}

\begin{proof}
\
\begin{enumerate}
\item Since $-v_1<-1<0,$ then $-v_1e^{-v_1}<0.$ Hence by  Properties W1 (that $W$ is increasing) and W4 (that $W[0]=0$), we get $R(v_1)=W[-v_1e^{-v_1}]<0$. 
Now $(-v_1e^{-v_1})'=e^{-v_1}(v_1-1)>0$ (since $1<v_1$). The Lambert function is increasing and so $W[-v_1e^{-v_1}]$ increases with $v_1$, proving R1.
\item To prove R2, note that by Property W2 $R(1)=W[-e^{-1}]=-1$.
\item We wish to prove that $R(v(k)-k)=\frac{k}{v(k)}-1$. Recall that $v(k)$ satisfies $k=v(k)-\frac{\ln{v(k)}}{1-\frac{1}{v(k)}}$. Hence $(k-v(k))(1-\frac{1}{v(k)})=-\ln{v(k)}$, and so $k-v(k)=-\ln{v(k)}+(\frac{k}{v(k)}-1)$.
Adding $\ln{(k-v(k))}$ to both sides of the equation gives
$$(k-v(k)) +\ln{(k-v(k))}=\ln{\Big(\frac{k}{v(k)}-1\Big)}+\Big(\frac{k}{v(k)}-1\Big).$$
Thus
$$(k-v(k))e^{k-v(k)}=\Big( \frac{k}{v(k)}-1\Big)e^{\frac{k}{v(k)}-1}.$$
Applying the Lambert function to both sides of the equation gives
\begin{equation}\label{E6}
W[(k-v(k))e^{k-v(k)}]=W\Big[\Big( \frac{k}{v(k)}-1\Big)e^{\frac{k}{v(k)}-1}\Big].
\end{equation}
The left-hand side of (\ref{E6}) is $R(v(k)-k)$. Now by W6, since $\frac{k}{v(k)}-1\geq -1$ then
$$
W \Big[\Big( \frac{k}{v(k)}-1\Big)e^{\frac{k}{v(k)}-1}\Big]=\frac{k}{v(k)}-1,
$$
proving R3.
\item
\begin{equation}\label{E7}
R'(v_1)=W'[-v_1e^{-v_1}]e^{-v_1}(v_1-1).
\end{equation}
Hence by Property W3,
$$R'(v_1)=\frac{R(v_1)e^{-v_1}(v_1-1)}{-v_1e^{-v_1}(R(v_1)+1)}=\frac{-R(v_1)}{R(v_1)+1}\cdot\frac{v_1-1}{v_1},
$$
proving R4.

\end{enumerate}
\end{proof}

For the next result we need the following lemma.

\begin{lemma}\label{L2}
The expression $v_1+(v_1+k)R(v_1),$ is negative in Region 3.
\end{lemma}

\begin{proof}
Substituting $v_1=v(k)-k,$ which is the right end of Region 3, in $v_1+(v_1+k)R(v_1),$ and recalling that by R3 $R(v(k)-k)=\frac{k}{v(k)}-1$, gives $0$. Additionally, by R4, the derivative \ $1+R(v_1)+(v_1+k)R'(v_1)$ of this expression is $$1+R(v_1)+\frac{-R(v_1)(v_1+k)(v_1-1)}{v_1(R(v_1)+1)}.$$
By R1 and R2, in Region 3 $1+R(v_1)>0,$ and \ $-R(v_1)>0.$
 Hence  the derivative is positive and so the expression $v_1+(v_1+k)R(v_1)$ is negative in Region 3.
\end{proof}

Recall Definition~\ref{DA} that $A=A(v_1)=-(v_1+k)e^{-v_1+\frac{kR(v_1)}{v_1}},$

\begin{lemma}\label{T3}
\
\begin{enumerate}
\item $W[A(v(k)-k)]=-1.$
\item $A(v_1)$ decreases with $v_1$ for $1< v_1< v(k)-k.$
\end{enumerate}
\end{lemma}

\begin{proof}
\
\begin{enumerate}
\item Since $W[\cdot]$ is strictly monotonic and by Property W2 $W[-e^{-1}]=-1$, we must prove that $A(v(k)-k)=-e^{-1}$. Now`
\begin{equation}\label{E8}
A(v(k)-k)=-v(k)e^{k-v(k)+\frac{kR(v(k)-k)}{v(k)-k}}.
\end{equation}
From R3, \ $R(v(k)-k)=\frac{k}{v(k)}-1$.
Substituting this in the right-hand side of (\ref{E8}) gives
$$
-v(k)e^{k-v(k)+\frac{k(k/v(k) -1)}{v(k)-k}}=-v(k)e^{k-v(k)-\frac{k}{v(k)}}=-e^{k-v(k)-\frac{k}{v(k)}+\ln{v(k)}},
$$

Substituting $\ln{v(k)}=-(k-v(k))\Big( 1-\frac{1}{v(k)}\Big)$ gives
$$
-e^{k-v(k)-\frac{k}{v(k)}-(k-v(k))\Big( 1-\frac{1}{v(k)}\Big)}=-e^{-1}.
$$
\item To prove that $A(v_1)$ is decreasing note that
\begin{equation}\label{E789}
A'(v_1)=\frac{e^{-v_1+\frac{kR}{v_1}}\Big(v_1(v_1+k-1)+kR \Big)\big( v_1+(v_1+k)R\big)}{v_1^2(R+1)},
\end{equation}
where $R=R(v_1)$. By properties R1 and R2, the denominator is positive. Thus we need to prove that the numerator is negative. 
In Region 3, \ $v_1>1.$ Thus by properties R1 and R2, \ $v_1+R>1+R>0.$  Hence the expression in the first parentheses of~\eqref{E789}, $v_1(v_1+k-1)+kR$ which equals $k(v_1+R)+v_1(v_1-1),$ is positive.
The expression in the second parentheses is negative by Lemma~\ref{L2}. Thus, $A'(v_1)$ is negative
 implying that $A(v_1)$ is indeed decreasing there.
\end{enumerate}
\end{proof}

Recall that $R=R(v_1),$ and that $v(k)$ is defined by the equation $k=v(k)-\frac{\ln{v(k)}}{1-\frac{1}{v(k)}}.$
We examine $W[A]$ as a function of $k.$
Given $v_1$ denote by $k_0$ the value of $k$ that satisfies  $v_1=v(k_0)-k_0$. Denote $u_0=v(k_0).$
Also denote by $A_0$ the value of $A$ that corresponds to $k_0,$ namely
 $ A_0=-(v_1+k_0)e^{-v_1+\frac{k_0R}{v_1}}.$

\begin{lemma}\label{P50}
Given $v_1:$
\begin{enumerate}
\item In Region 3 $W[A]$ increases with $k$.
\item In Region 3 $\frac{d}{dk}W[A]$ decreases with $k$.
\item $\frac{d}{dk}W[A]\Big|_{k=k_0}=\frac{1}{u_0}.$
\end{enumerate}
\end{lemma}

\begin{proof}
\
\begin{enumerate}
\item To prove that $W[A]$ increases with $k$ in Region 3 note that 
$$
\frac{dA}{dk}=-\frac{e^{-v_1+\frac{kR}{v_1}}}{v_1}(v_1+(v_1+k)R).
$$ 
By Lemma~\ref{L2}, \ $v_1+(v_1+k)R<0$ in Region 3. Thus $\frac{dA}{dk}>0$ and $A$ increase with $k$ in Region 3. Since $W$ is an increasing function, $W[A]$ also increases with $k$ in Region 3.
\item Using Property W3 yields
\begin{equation}\label{E34}
\frac{d}{dk}W[A]=\frac{(v_1+(v_1+k)R)W[A]}{v_1(v_1+k)(W[A]+1)}.
\end{equation}
The above equals
$$\frac{v_1+(v_1+k)R}{v_1(v_1+k)}-\frac{(v_1+(v_1+k)R)}{v_1(v_1+k)(W[A]+1)}.$$
Differentiating the above expression with respect to $k$ gives
\begin{equation}\label{E35}
\frac{d^2}{dk^2}W[A]= \frac{-v_1^2}{(v_1(v_1+k))^2}+\frac{v_1(v_1+k)\frac{d}{dk}W[A](v_1+(v_1+k)R)}{(v_1(v_1+k)(W[A]+1))^2}.
\end{equation}
The first term in~(\ref{E35}) is negative. The second term is also negative since we proved in Part 1 that $\frac{d}{dk}W[A]$ is positive, and by Lemma~\ref{L2} that $v_1+(v_1+k)R$ is negative.
\item
By~(\ref{E34})
$$\frac{d}{dk}W[A]=\frac{(v_1+(v_1+k)R)W[A]}{v_1(v_1+k)(W[A]+1)}.$$
Note that for $k=k_0$  the numerator vanishes according to Lemma~\ref{L2}, and the denominator vanishes according to Part 1 of Lemma~\ref{T3}.

By L'H\^{o}pital's rule
$$\frac{d}{dk}W[A]\Big|_{k=k_0}=\frac{\frac{d}{dk}\Big((v_1+(v_1+k)R)W[A]\Big)\Big|_{k=k_0}}{\frac{d}{dk}\Big(v_1(v_1+k)(W[A]+1)\Big)\Big|_{k=k_0}}.$$
Hence
\begin{equation}\label{E30}
\frac{d}{dk}W[A]\Big|_{k=k_0}=\frac{RW[A_0]+(v_1+(v_1+k_0)R)\frac{d}{dk}W[A]\Big|_{k=k_0}}{v_1(W[A_0]+1)+v_1(v_1+k_0)\frac{d}{dk}W[A]\Big|_{k=k_0}}.
\end{equation}

Denote $s=\frac{d}{dk}W[A]\Big|_{k=k_0}$. Then (\ref{E30}) becomes
$$s=\frac{RW[A_0]+(v_1+(v_1+k_0)R)s}{v_1(W[A_0]+1)+v_1(v_1+k_0)s}.$$
Recalling that $v_1+(v_1+k_0)R$ appearing in the numerator equals $0,$ and that $W[A_0]=-1$ yields
\begin{equation}\label{E31}
s=\frac{-R}{v_1(v_1+k_0)s}.
\end{equation}
Note that $s=0$ does not solve~(\ref{E31}); hence the above expression is well defined.
From~(\ref{E31}) we get
\begin{equation}\label{E32}
s^2=\frac{-R}{v_1(v_1+k_0)}.
\end{equation}
By R3, $R=R(v_1)=R(u_0-k_0)=\frac{k_0}{u_0}-1.$ Substituting this and also $v_1=u_0-k_0$  in~(\ref{E32}) gives
\begin{equation}\label{E111}
s^2=\frac{1-\frac{k_0}{u_0}}{(u_0-k_0)u_0}=\frac{u_0-k_0}{(u_0-k_0)u_0^2}=\frac{1}{u_0^2}.
\end{equation}

Now, from the proof of the first part of Lemma~\ref{P50} (that in Region 3 $W[A]$ increases with $k$), in Region 3 $\frac{d}{dk}W[A]$ is  non-negative. In particular, $s\geq 0.$ Combining this with the fact that $s\ne 0$, which we established earlier, gives $s>0$. This, together with~(\ref{E111}), implies that $s=\frac{1}{u_0}.$
\end{enumerate}
\end{proof}

We can now prove Theorem~\ref{T1}.

\subsubsection{Proof of Theorem~\ref{T1}}\label{PPT1}
\

Recall that  $v=v(k)$ is the border between Regions 3 and 4, and is the unique solution for $k=v-\frac{\ln{v}}{1-\frac{1}{v}}$ (see the proof of Proposition~\ref{P66} in Section~\ref{AP66} in the Appendix).
\newline

In Region 3, by Theorem~\ref{T10}
$$\l_1+\l_2=v_1+R(v_1)+W[A]-\frac{v_2R(v_1)}{v_1}.$$
We use R4 and~(\ref{E789}) to arrive at
\begin{equation}\label{E9}
(\l_1+\l_2)'=\frac{v_1^2W[A](R+1)+(v_1+k)\Big( v_1^2+v_1(v_1+k)R+kR^2\Big)}{v_1^2(v_1+k)\Big(W[A]+1\Big)(R+1)},
\end{equation}
where the derivative is with respect to $v_1.$
According to R1 and R2, \ $R+1>0$ in Region 3 and according to Lemma~\ref{T3} in Region 3 $W[A]+1>0$. Hence, in Region 3 the denominator is positive . Thus it is sufficient to prove that the numerator is negative.

At the end points of Region 3, namely $v_1=1$ and $v_1=v(k)-k,$ the numerator vanishes. We wish to prove that for all $k\geq 0$ and $1<v_1<v(k)-k$ (i.e., for all $v_1$ in Region 3) the numerator is negative. This is  illustrated in Figure~\ref{F201} which presents the numerator as a function of $v,$ for $k=2.$ Note that $v(2)=3.81449$; hence Region 3 is \ $1<v_1<1.81449.$

\begin{figure}[h]
\includegraphics[scale=0.5]{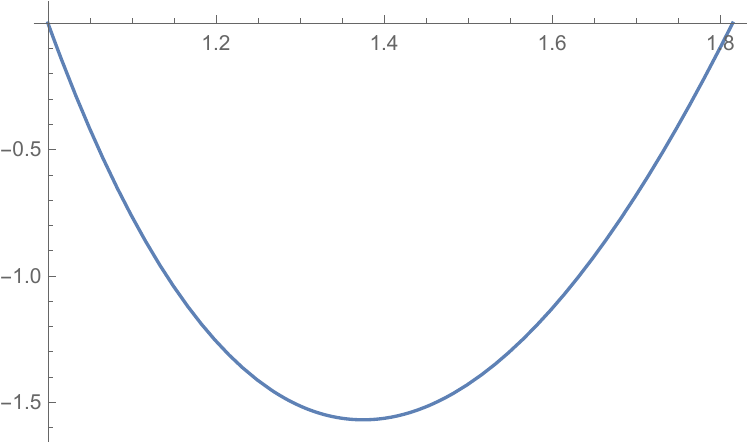}
\centering
\caption{The numerator of $(\l_1+\l_2)'$ as a function of $v_1,$ in Region 3, for $k=2.$}\label{F201}
\end{figure}

Recall that we denoted $v_1=u_0-k_0.$
For any given $v_1,$ the numerator evaluated at $k_0=u_0-v_1$ (corresponding to the right end point $v_1=u_0-k_0$ of Region 3 when $k=k_0$) equals $0.$ We wish to prove that, given $v_1$, the numerator for all $k>k_0$ decreases with $k$ and is therefore negative. This will be proved shortly and is demonstrated in Figure~\ref{F202}.

\begin{figure}[h]
\begin{center}
\includegraphics[scale=0.3]{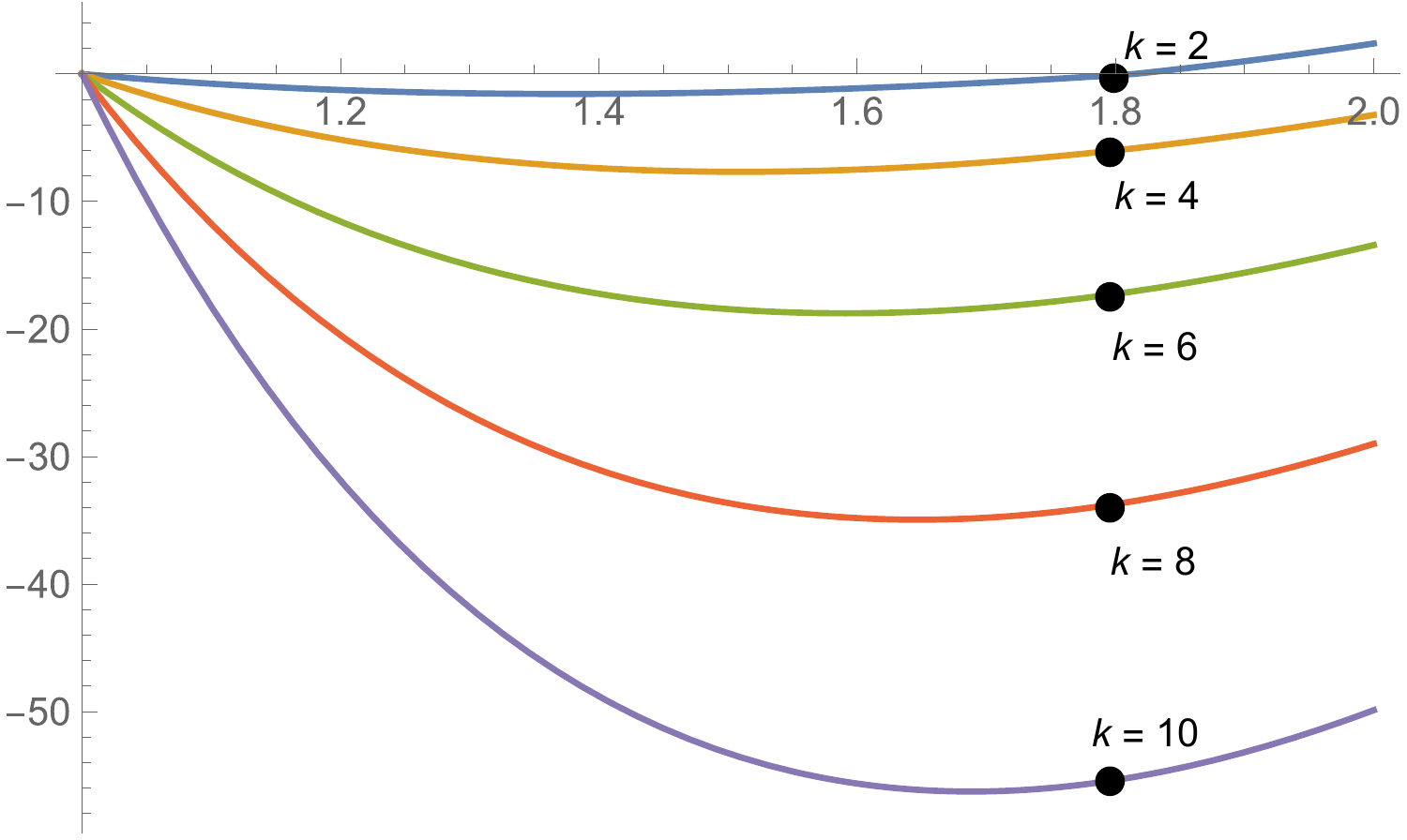}
\caption{The numerator of $(\l_1+\l_2)'$ in Region 3, for $k=2,4,6,8,10.$}\label{F202}
\end{center}
\end{figure}

Figure~\ref{F202} which presents the numerator for $k=2,4,6,8,10$ shows that at $v_1=1.81449$ the numerator equals zero for $k=k_0=2$ (upper line); for increasing $k$ the value of the numerator declines (and is thus negative as claimed). Note also that the length of Region 3 increases with $k.$ The increase arises because for all $k>0$ the left end of Region 3 is $1,$ and the right end is $v(k)-k=\frac{\ln{v(k)}}{1-\frac{1}{v(k)}}$ which increases  with $v(k)$ and thus with $k$. In particular when $k\to 0,$ the right end  $v(k)-k\to 1,$ and so when $k=0$ Region 3 is empty (since it consists of $1<v_1<1$).

Formally, given $v_1$, denote by $N(k)$ the numerator of $(\l_1+\l_2)'$ (when the derivative is with respect to $v_1$), namely
$$N(k)=v_1^2W[A](R+1)+(v_1+k)\Big( v_1^2+v_1(v_1+k)R+kR^2\Big).$$

We  now differentiate $N(k)$ with respect to $k.$ Note that $v_1$ and $k_0$ are fixed in $k$; consequently $R,$ which is a function only of $v_1$, is also fixed.

\begin{equation}\label{E33}
\frac{d}{dk}N(k)= v_1^2(R+1)\frac{d}{dk}W(k)+v_1^2+v_1(v_1+k)R+kR^2+R(v_1+k)(R+v_1).
\end{equation}

We will first prove that $\frac{d}{dk}N(k)$ for $k=k_0$ is negative. We will then prove that $\frac{d}{dk}N(k)$ decreases with $k$ for all $k>k_0$; hence for all $k>k_0$ $\frac{d}{dk}N(k)<0.$ And since $N(k_0)=0$ it follows that $N(k)<0$ for all $k>k_0.$

Recall that for $k=k_0:$ \  $v_1=u_0-k_0,$ and \ $R=R(v_1)=R(u_0-k_0)=\frac{k_0}{u_0}-1.$ Also, by Part 3 of Lemma~\ref{P50}, $\frac{d}{dk}W[A]\Big|_{k=k_0}=\frac{1}{u_0}.$ Substituting all this in~(\ref{E33}) gives
$$(u_0-k_0)^2\frac{k_0}{u_0^2}+(u_0-k_0)^2+(u_0-k_0)u_0\left(\frac{k_0}{u_0}-1\right)+k_0\left(\frac{k_0}{u_0}-1\right)^2+\left(\frac{k_0}{u_0}-1\right)u_0\left(\frac{k_0}{u_0}-1+u_0-k_0\right).$$
This equals
$$\left(\frac{k_0}{u_0}-1\right)^2k_0+\left(\frac{k_0}{u_0}-1\right)(2k_0-u_0+u_0^2-u_0k_0),$$ which equals
\begin{equation}\label{E6666}
\left(\frac{k_0}{u_0}-1\right)\left( (u_0-1)(u_0-k_0)+\frac{k_0^2}{u_0}\right).
\end{equation}

Since for all $k> 0$ we have that $v(k)-k> 1,$ then the first factor in~(\ref{E6666}) is negative, and the second factor is  positive in Region 3, proving that indeed in Region 3 $\frac{d}{dk}N(k)\big|_{k=k_0}<0$ .

We now prove that $\frac{d}{dk}N(k)$ decreases with $k$ for all $k>k_0.$
By~(\ref{E33})
$$\frac{d}{dk}N(k)= v_1^2(R+1)\frac{d}{dk}W(k)+v_1^2+v_1(v_1+k)R+kR^2+R(v_1+k)(R+v_1).$$
This equals

\begin{equation}\label{E37}
v_1^2(R+1)\frac{d}{dk}W[A]+2R(R+v)k+v_1^2(R+1)+v_1R(R+v_1).
\end{equation}
The first term $v_1^2(R+1)\frac{d}{dk}W[A]$decreases with $k,$ since by R1 and R2 $R+1>0$ in Region 3 and by the second statement in Lemma~\ref{P50},
$\frac{d}{dk}W[A]$ decreases with $k.$ The expression $2R(R+v)k+v_1^2(R+1)+v_1R(R+v_1)$ is a linear function of $k$ with a negative multiplier $2R(R+v)$ for $k$; hence it also declines with $k.$
So indeed, for all $k>k_0$, $\frac{d}{dk}N(k)$ decreases with $k$.  Since we proved that $\frac{d}{dk}N(k)\big|_{k=k_0}<0,$  it follows that for all $k>k_0$, $\frac{d}{dk}N(k)<0.$
It follows from $N(k_0)=0$ that, for all $k>k_0$, 	`$N(k)<0$. This means that the numerator of $(\l_1+\l_2)'$ (where the derivative is with respect to $v_1$) is negative for all $k>k_0.$ Since $v(k)-k=\frac{\ln{v(k)}}{1-\frac{1}{v(k)}}$ increases with $k$, for all $k>0 $ \  $v(k)-k > u_0-k_0=v_1.$ Hence, for all $k>0,$ and for all $1<v_1<v(k)-k,$ (i.e., Region 3) $(\l_1+\l_2)'<0,$ implying that in Region 3 $\l_1+\l_2$ decreases with $v_1$ .

\subsection{Profit maximization - elaborating on the results in Section~\ref{PM}}\label{APMN}

The first step is to find an explicit expression for  the local maximum of $\pi$ in each region. However, these local maxima may appear outside the region in which they are valid; in that case they are irrelevant. In other words, the prices $p$ that are relevant candidates for the global maximum $\pi$ are the maxima for which $(v-p,k)$ still lies in the original region of $(v,k)$. For instance, if $(v,k)$ lies in Region 4, then the local maximum for Region 4 may have $(v-p,k)$ lie in Region 3. That local maximum is not relevant, since in Region 3 there is a different expression for the local maximum of $\pi$.
For a region with no relevant local maximum we need to check its end points. Among these candidates, the $p$ that attains the maximum value for $\pi(p)$ is the price that maximizes the seller's expected profit.

\subsubsection{Proof of Theorem~\ref{T4}}\label{APM}

Recall that the right end of Region 3 was denoted as $p_3.$ Then $p_3=v-k-1$. Recall that $p_2$ and $p_4$ denote the local maxima of Regions 2 and 4 respectively. Then $p_2$, $p_3$, and $p_4$ are the only candidates for the global maximum $p^*$ of the expected profit $\pi$. Recall that if $p_2$ does not belong to Region 2 it is irrelevant. Similarly for $p_4$ and Region 4. In contrast, $p_3=v-k-1,$ which is the point separating between Regions 2 and 3, always exists (when Region 3 exists).
Next, we wish to find explicit expressions for $p_i$ and $\pi_i$ in each region.
Denote $\pi_i=\pi(p_i) for i=2,3,4.$

\begin{definition}\label{D200}
$$W=W[-(k+1)e^{-(k+1)}].$$
\end{definition}

\begin{proposition}\label{P5000}
The candidates $p_2, p_4$ and $p_,$ for the global maxima $p^*$ of the expected profit $\pi$ and their associated expected profits $\pi_i$ are
\begin{enumerate}
\item  \begin{equation}\label{E901}
p_2=v-{\ln\left(v\right)\over1-{1\over v}}, \ \  \
 \pi_2=v-1-\ln\left(v\right).
 \end{equation}
\item
\begin{equation}\label{E891}
p_3=v-k-1, \ \ \
\pi_3=(v-k-1)\left( 1-e^{-(W+k+1)}\right).
\end{equation}
\item \begin{equation}\label{E700}
p_4=v-k-{\ln\left(v-k\right)\over1-{1\over v-k}}, \  \  \
\pi_4=v-k-1-\ln\left(v-k\right),
\end{equation}
\end{enumerate}

\end{proposition}
\begin{proof}
\
\begin{enumerate}
\item Recall that in Region 2, $\l_1=0$. Hence by~(\ref{E9001N}),
$(v-p){1-e^{-\l_2}\over \l_2}=1$, and so in particular,
\begin{equation}\label{EE2}
p_2= v-{\l_2\over1-e^{-\l_2}}.
\end{equation}
Substituting this in~\eqref{E5000} gives
\begin{equation}\label{EEE2A}
\pi_2= \left( v-{\l_2\over1-e^{-\l_2}}\right)\left( 1-e^{-\l_2} \right)=v(1-e^{-\l_2})-\l_2.
\end{equation}
To find the local maximum in Region 2 compute
\begin{equation}\label{E903}
\frac{d}{d\l_2}\pi=ve^{-\l_2}-1.
\end{equation}
The unique solution for $\frac{d}{d\l_2}\pi=0$ is $\l_2^*=\ln{v}.$
Substituting this in~\eqref{EE2} and~\eqref{EEE2A} gives $p_2$ and $\pi_2,$ as stated in~(\ref{E901}).

\item Because $p_3=v-k-1$ is on the border between Regions 2 and 3, $\l_1=0$, and so
\begin{equation}\label{EE3}
\pi_3=(v-k-1)\left( 1-e^{-\l_2}\right).
\end{equation}
Now,  $p=v-k-1$ implies that $v_1=1$.
Hence, by Theorem~\ref{T10}, in Region 3
$$\l_2=W[A]-(k+1)R(1).$$
By Property W2, $R(1)=W[-e^{-1}]=-1.$ Note also that by Definition~\ref{DA}, for $v_1=1,$
$$A=-(k+1)e^{-(k+1)}.$$ Hence

$$\l_2=W[-(k+1)e^{-(k+1)}]-(k+1)(-1)=W+k+1.$$
 Substituting $\l_2=W+k+1$ in~\eqref{EE3} gives~(\ref{E891}).

\item Recall that in Region 4, $\l_2=0$. By (\ref{E9000N})
\begin{equation}\label{EE4}
p_4=v-k-{\l_1\over1-e^{-\l_1}}.
\end{equation}
So by~(\ref{E5000})
\begin{equation}\label{EEE4}
\pi_4 =(v-k)(1-e^{-\l_1})-\l_1.
\end{equation}
Now,
\begin{equation}\label{E900}
\frac{d}{d\l_1}\pi=(v-k)e^{-\l_1}-1,
\end{equation}
giving the profit-maximizing value $\l_1^*=\ln\left(v-k\right)$. Substituting this in~\eqref{EE4} and~\eqref{EEE4}, gives $p_4$ and $\pi_4,$ as stated in~(\ref{E700}).

\end{enumerate}
\end{proof}

\begin{proposition}\label{P6}
For all $k> 0,$ and $v>1$,
\begin{enumerate}
\item $\pi_4 < \pi_2$

\item $\pi_3< \pi_2$,
\end{enumerate}
\end{proposition}

Proposition~\ref{P6} is important since it implies that whenever $p_2\in {\textit{Region 2}}$ (namely, the local maximum of Region 2 still lies in Region 2) then the global maximum $\pi^*$ is attained at $p_2,$ and so $\pi^*=\pi_2.$

\begin{proof}

\
\begin{enumerate}
\item By (\ref{E700}), \  $\pi_4=v-k-1-\ln\left(v-k\right),$ and by (\ref{E901}), \  $\pi_2=v-1-\ln{\left(v\right)}.$
Note that for $y\geq 1$, the value of $y-1-\ln{y}$ is increasing. This result and the inequalities $1\leq v-k<v$ imply that $\pi_4 < \pi_2.$
\item
To prove that $\pi_3< \pi_2$  for $k>0$ we must prove that
$$(v-k-1)\left( 1-e^{-(W+k+1)}\right)\leq v-1-\ln{v},$$ which is equivalent to proving that
\begin{equation}\label{E13}
(v-k-1)\left( 1-e^{-(W+k+1)}\right)- v+\ln{v}\leq -1.
\end{equation}
We will find the maximum value  of the left-hand side of (\ref{E13}) and show that it equals $-1$. To find the maximum we solve
$$
\left((v-k-1)\left( 1-e^{-(W+k+1)}\right)- v+\ln{v}\right)'=1-e^{-W-k-1}-1+\frac{1}{v}=0.
$$
The solution is $v=e^{W+k+1}$. The  second derivative of the left-hand side of (\ref{E13})  is 
$(1-e^{-W-k-1}-1+\frac{1}{v})'=-\frac{1}{v^2}$. Hence $e^{W+k+1}$ is a local maximum. 
Note that
\begin{equation}\label{E14}
e^{W+k+1}=e^We^{k+1}=\frac{-(k+1)e^{-k-1}e^{k+1}}{W}=\frac{k+1}{-W}.
\end{equation}
We now show that for $v=e^{W+k+1}$ the left-hand side of (\ref{E13}) is $-1$.
By Property W5 we obtain
$$=\left(-\frac{k+1}{W}-k-1\right)\left(1+\frac{We^{-(k+1)}}{(k+1)e^{-(k+1)}}\right)
 +\frac{k+1}{W}+\ln\left(\frac{k+1}{-W}\right)$$
$$=\left(-\frac{k+1}{W}-k-1\right)\left(1+\frac{W}{k+1}\right)+\frac{k+1}{W}+\ln\left(\frac{k+1}{-W}\right),$$
which gives
\begin{equation}\label{E15}
-k-2-W+\ln\left(-\frac{k+1}{W}\right).
\end{equation}
Since by (\ref{E14}) $\frac{k+1}{-W}=e^{W+k+1}$
$$
\ln\left(\frac{k+1}{-W}\right)=\ln{e^{(W+k+1)}}=W+k+1.
$$
Substituting this in (\ref{E15}) gives
$$-k-2-W+\ln\left(\frac{k+1}{-W}\right)=-k-2-W+W+k+1=-1.$$
\end{enumerate}
\end{proof}

{\bf{When does $p_2$ belong to Region 2?}}

\

As explained earlier, Proposition~\ref{P6}  implies that whenever $p_2\in {\textit{Region 2}}$ (namely, the local maximum of Region 2 still lies in Region 2), then the global maximum $\pi^*$ is attained at $p_2.$
We now find the conditions guaranteeing that $p_2$ lies in Region 2. Recall that $W=W[-(k+1)e^{-(k+1)}].$

\begin{lemma}\label{P5}

If \ $1\leq v\leq e^{W+k+1},$ then $p_2$ lies in Region 2.
\end{lemma}
\begin{proof}
We will first prove that $p_2$ lies in Region 2, iff $\frac{\ln{v}}{1-\frac{1}{v}}\leq k+1.$ Then, we will prove that $\frac{\ln{v}}{1-\frac{1}{v}}\leq k+1$ iff $v\leq e^{W+k+1}.$
Recall that Region 2 refers to all $p$ satisfying $v-k-1\leq p\leq v-1$.
First, by L'H\^{o}pital's rule
$$\lim_{v\to 1}\frac{\ln{v}}{1-\frac{1}{v}}=\frac{\lim_{v\to 1}\frac{1}{v}}{\lim_{v\to 1}\frac{1}{v^2}}=\frac{1}{1}=1.$$
Because $\frac{\ln{v}}{1-\frac{1}{v}}$  is increasing for all $v\geq 1,$
$$\frac{\ln{v}}{1-\frac{1}{v}}\geq 1,$$
and so
$$v-\frac{\ln{v}}{1-\frac{1}{v}}\leq v-1,$$ 
proving that $p_2\leq v-1$.
We now prove that $p_2\geq v-k-1,$ iff
$\frac{\ln{v}}{1-\frac{1}{v}}\leq k+1$. The latter is equivalent to
$$v-\frac{\ln{v}}{1-\frac{1}{v}}\geq v-k-1.$$ 
Hence in this case $p_2\geq v-k-1,$ and so $p_2$ lies in Region 2.
We will show that the unique solution for
\begin{equation}\label{E16}
\frac{\ln{v}}{1-\frac{1}{v}}=k+1
\end{equation}
is \ $e^{W+k+1}.$
By (\ref{E14}) $$e^{W+k+1}=-\frac{k+1}{W}.$$ Hence
$$We^{W+k+1}=-(k+1),$$ implying that
$$(W+k+1)e^{W+k+1}-(k+1)e^{W+k+1}=-(k+1),$$ and so
$$(W+k+1)e^{W+k+1}=(k+1)\left(e^{W+k+1}-1\right).$$
This is equivalent to
$$\frac{(W+k+1)e^{W+k+1}}{e^{W+k+1}-1}=k+1,$$ which implies
$$\frac{W+k+1}{1-\frac{1}{e^{W+k+1}}}=k+1.$$

Hence $e^{W+k+1}$ solves (\ref{E16}). Because the left-hand side of (\ref{E16}) is strictly increasing, $e^{W+k+1}$ uniquely solves (\ref{E16}).
\end{proof}

\begin{corollary}\label{C333}
If $1\leq v\leq e^{W+k+1},$ then $p^*=p_2=v-{\ln\left(v\right)\over1-{1\over v}},$ \ and \ $\pi^*=\pi_2=v-1-\ln\left(v\right).$
\end{corollary}

For larger $v$, namely $v \geq e^{W+k+1}$, we need to find where  $\pi_3 \geq  \pi_4$ is satisfied. In that case, $p^*=p_3$. Recall that $p_3=v-k+1$ is the point that separates Regions 2 and 3. Therefore it always exists and is relevant. But $p_4,$ which is the local maximum of the expression for $\pi$ in Region 4, may not belong to Region 4, and in that case it is not relevant.

Recall that
$\pi_4=v-k-1-\ln\left(v-k\right)$, and that $\pi_3=(v-k-1) \left( 1-e^{-(W+k+1)}\right))$.

Let $f(v)=\pi_4-\pi_3$.  Hence,

\begin{definition}\label{D99}
$f(v) \equiv \left( 1-\frac{v}{k+1}\right)W-\ln{(v-k)}.$
\end{definition}

Also, denote $v_m$ as the minimum of $f$ in Region 3.
\begin{lemma}\label{P8}
Given $k$, $f(v)$ is a strictly convex function which has exactly two roots, the smaller of which is  $k+1.$
\end{lemma}

\begin{proof}
First, note that $f(k+1)=0.$ Now,
\begin{equation}\label{E66}
f'(v)=\left(\left( 1-\frac{v}{k+1}\right)W-\ln{(v-k)} \right)'=-\frac{W}{k+1}-\frac{1}{v-k},
\end{equation}
and
we have 
$$
f''(v)=\frac{1}{(v-k)^2}>0.
$$
Hence $f$ is indeed strictly concave, and so has at most two roots. To see that  $k+1$ is not the only root, we need to find $v_m$, the minimum of $f$ in Region 3, and show that $k+1\ne v_m.$
To solve $f'(v)=0,$ we use ~(\ref{E66}) obtaining
$$-\frac{W}{k+1}-\frac{1}{v-k}=0.$$ This is equivalent to
$$v-k=-\frac{k+1}{W}.$$ Hence
$$v=k-\frac{k+1}{W},$$ and so
 $v_m=k-\frac{k+1}{W}$ is the minimum of $f$ in Region 3.
We now show that $k+1<v_m.$ Since  $-(k+1)e^{-(k+1)}\geq -e^{-1}$\  for all $k>0$, and $W[\cdot]$ is an increasing function, then $W=W[-(k+1)e^{-(k+1)}]\geq W[-e^{-1}]=-1$. Thus for all $k>0,$ \ $W>-(k+1)$ and so $$1<-\frac{k+1}{W},$$ and
$$k+1<k-\frac{k+1}{W}=v_m.$$
Hence $f$  has exactly two roots.

The function $f$ is strictly convex, and $f(v_m)< 0$ (since $f(k+1)=0,$ and $v_m$ is the minimum of $f$). Therefore $v_m$ lies between the two roots. Because $k+1<v_m,$ $k+1$ is the smaller root. Denote $v_f$ as the larger root (a formal definition is presented after the end of the proof). Then
\begin{equation}\label{E777}
k+1<v_m\leq v_f.
\end{equation}

\end{proof}

\begin{definition}\label{D100}
$v_f$ is the unique value that satisfies both $f(v)=0,$  and $v> k+1.$
\end{definition}

Because of Lemma~\ref{P8}, $v_f$ is well defined. See Figure~\ref{F100}  that presents $f$ as a function of $v$ for $k=1$.

\begin{figure}[H]
\begin{center}
\includegraphics[scale=0.25]{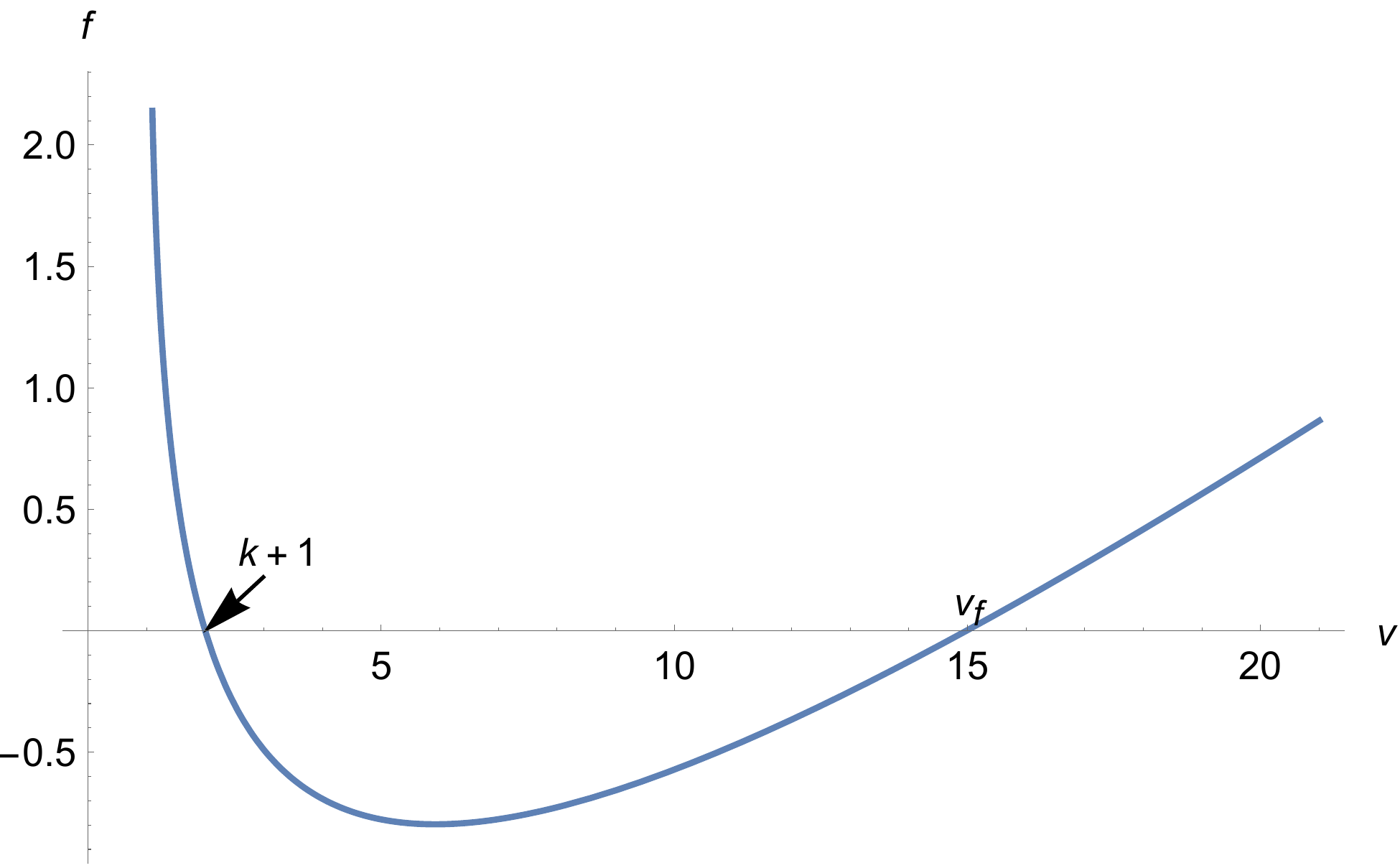}
\caption{$f$ as a function of $v$}\label{F100}
\end{center}
\end{figure}

Since $f(v)=\pi_4-\pi_3$, it follows that for all $v$ satisfying $k+1< v< v_f$ $\pi_4<\pi_3$.
Recall that by Proposition~\ref{P6} $\pi_3 \leq \pi_2,$ and that  by Lemma~\ref{P5} $p_2$ is relevant for all $v$ satisfying $1\leq v\leq e^{W+k+1}.$

\begin{corollary}\label{C1}
For every pair $(v,k)$
$$e^{W+k+1}\leq v_f.$$
\end{corollary}

\begin{proof}
Recall that by (\ref{E14}) we have  $$e^{W+k+1}=-\frac{k+1}{W}.$$ Thus
$$e^{W+k+1}=-\frac{k+1}{W}\leq k-\frac{k+1}{W}=v_m\leq v_f,$$
where the last inequality follows from~(\ref{E777}).
\end{proof}

\begin{corollary}
\
\begin{itemize}
\item For all $1\leq v\leq e^{W+k+1},$ \ \ $p^*=p_2,$ \ and \ $\pi^*=\pi_2$.
\item For all $e^{W+k+1}<v\leq v_f,$ \ \  $p^*=p_3,$ \ and \ $\pi^*=\pi_3$.
\end{itemize}
\end{corollary}

But what happens for $v > v_f$?  By Lemma~\ref{P8} $\pi_4>\pi_3$ \ for \ $v> v_f.$ But is $p_4$ relevant (i.e., does it belong to Region 4) for $v> v_f$? The following Lemma says Yes.

\begin{lemma} \label{P9}
\
\begin{enumerate}
\item If  $v\geq k+v(k)$ then $p_4$ lies in Region 4.

\item For every pair $(v,k)$ \
$v_f\geq k+v(k).$
\end{enumerate}
\end{lemma}

\begin{proof}

\
\begin{enumerate}
\item Recall that Region 4 refers to all $p$ satisfying  $0\leq p\leq v-v(k).$

Proving that $p_4\geq 0$ requires proving that
$$v-k-{\ln\left(v-k\right)\over1-{1\over v-k}}\geq 0.$$
This is equivalent to
$$(v-k)\left( 1-\frac{1}{v-k}\right)-\ln(v-k)\geq 0.$$
Hence we need to prove that
\begin{equation}\label{E80}
(v-k) -\ln(v-k)\geq 1.
\end{equation}
Note that $v-k\geq v(k)\geq 1$ (where the last inequality follows from the fact that for all $k\geq 0$ $v(k)\geq 1$.)
For $v-k\geq 1$  (\ref{E80}) always holds since the left-hand side of~(\ref{E80}) equals $1$ for $v-k=1$ and is increasing for $v-k \geq 1$. Hence $p_4$ is indeed non-negative.
We now prove that $p_4\leq v-v(k)$ if $v\geq k+v(k)$.
Since $\frac{\ln y}{1-\frac{1}{y}}$ increases with $y$, $v-k\geq v(k)$ implies that
$$\frac{\ln (v-k)}{1-\frac{1}{v-k}}\geq \frac{\ln v(k)}{1-\frac{1}{v(k)}}.$$
Hence
$$v-k-\frac{\ln (v-k)}{1-\frac{1}{v-k}}\leq v-k-\frac{\ln v(k)}{1-\frac{1}{v(k)}}.$$
The left-hand side of the above equals $p_4$, so we have
$$
p_4\leq v-k-\frac{\ln v(k)}{1-\frac{1}{v(k)}}.
$$
Recall that $v(k)=k+\frac{\ln{v(k)}}{1-\frac{1}{v(k)}}$, so that $p_4 \leq v-v(k)$.

\item We will prove that $k+v(k)\leq v_m$, where $v_m$ was defined as the minimum of $f.$ This result and the observation that $v_m\leq v_f$ (see~(\ref{E777}) in the proof of Lemma~\ref{P8}) complete the proof. Recall that $$v_m=k-\frac{(k+1)}{W}.$$ Hence we need to prove that $$v(k)\leq -\frac{(k+1)}{W}.$$
By (\ref{E14}) $$-\frac{(k+1)}{W}=e^{W+k+1},$$ so we must prove that
$$
v(k)\leq e^{W+k+1}.
$$
Both sides of the inequality are functions of $k$. For $k=0$, we obtain equality with $1$ on both sides. As seen in Figure~\ref{F101} from then on $v(k)<e^{W+k+1}$.
\end{enumerate}
\end{proof}

\begin{figure}[H]
\begin{center}
\includegraphics[scale=0.2]{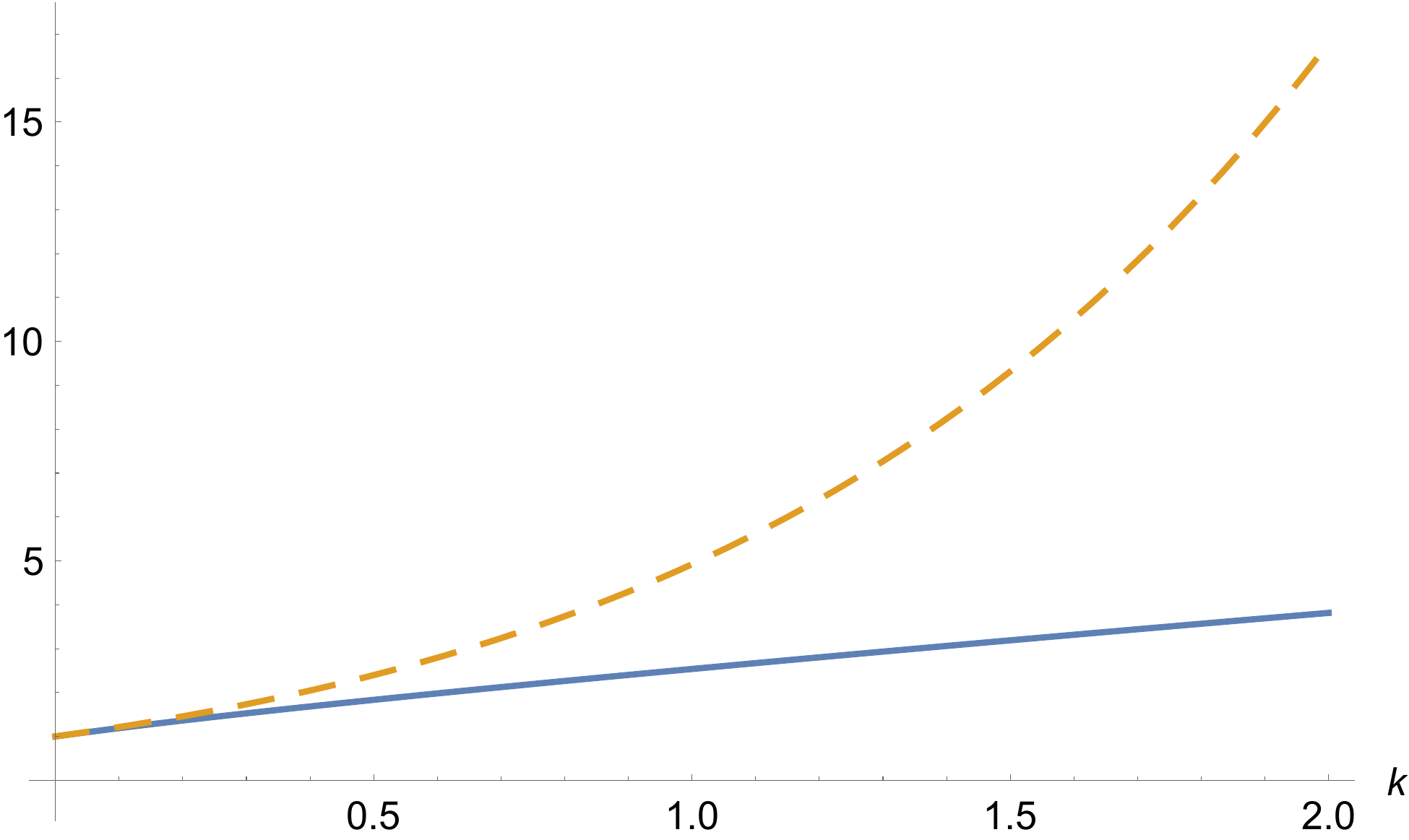}
\caption{$v(k)$ (solid line) and $e^{W+k+1}$ (dashed line) as functions of $k$}\label{F101}
\end{center}
\end{figure}

The proof of Theorem~\ref{T4} follows immediately from the results in this section.
Note that if $k$ is sufficiently large then the condition $1 \leq v\leq e^{W+k+1}$ (appearing in case 2 of Theorem~\ref{T4}) is indeed satisfied and the upper bound $v-1-\ln{v}$ of $\pi$ is thus obtained.
Therefore the profit-maximizing penalty $k$ should be large enough to satisfy the condition.

\subsubsection{Proof of Corollary~\ref{SUP}}\label{ASUP}
According to Proposition~\ref{P6} and Theorem~\ref{T4}, for all $(v, k)$ the upper bound of $\pi$ is  $ v-1-\ln{v}.$   This upper bound is  realized only when  $1 \leq v\leq e^{W+k+1}$ (see case 2 in  Theorem~\ref{T4}). In that case $\pi^*=\pi_2 = v-1-\ln{v},$ and all consumers arrivals are in period 2. If $(v, k)$ does not satisfy this condition, then the firm's expected profit will be strictly smaller than $v-1-\ln{v}$ (since according to Proposition~\ref{P6} $\pi_2>\pi_3, \pi_4$).

\subsection{More on the effects of changing the arrival cost $c$ - l on the results appearing in Section~\ref{EC}}\label{AEC}

Technically, Region 3 disappears when its left and right borders coincide, or when its right border equals zero. Figure~\ref{F3000} presents the borders of Region 3, as functions of the arrival cost $c.$
As shown by Figure~\ref{F3000}, the left and right borders coincide when $c=0.$
At the other extreme, when $c=8,$ the right end of Region 3 becomes zero; thus Region 3 disappears.
When $c=6.5$ the left border of Region 3 equals zero; thus Region 4 disappears.

\

 \begin{figure}[H]
\begin{center}
\includegraphics[scale=0.4]{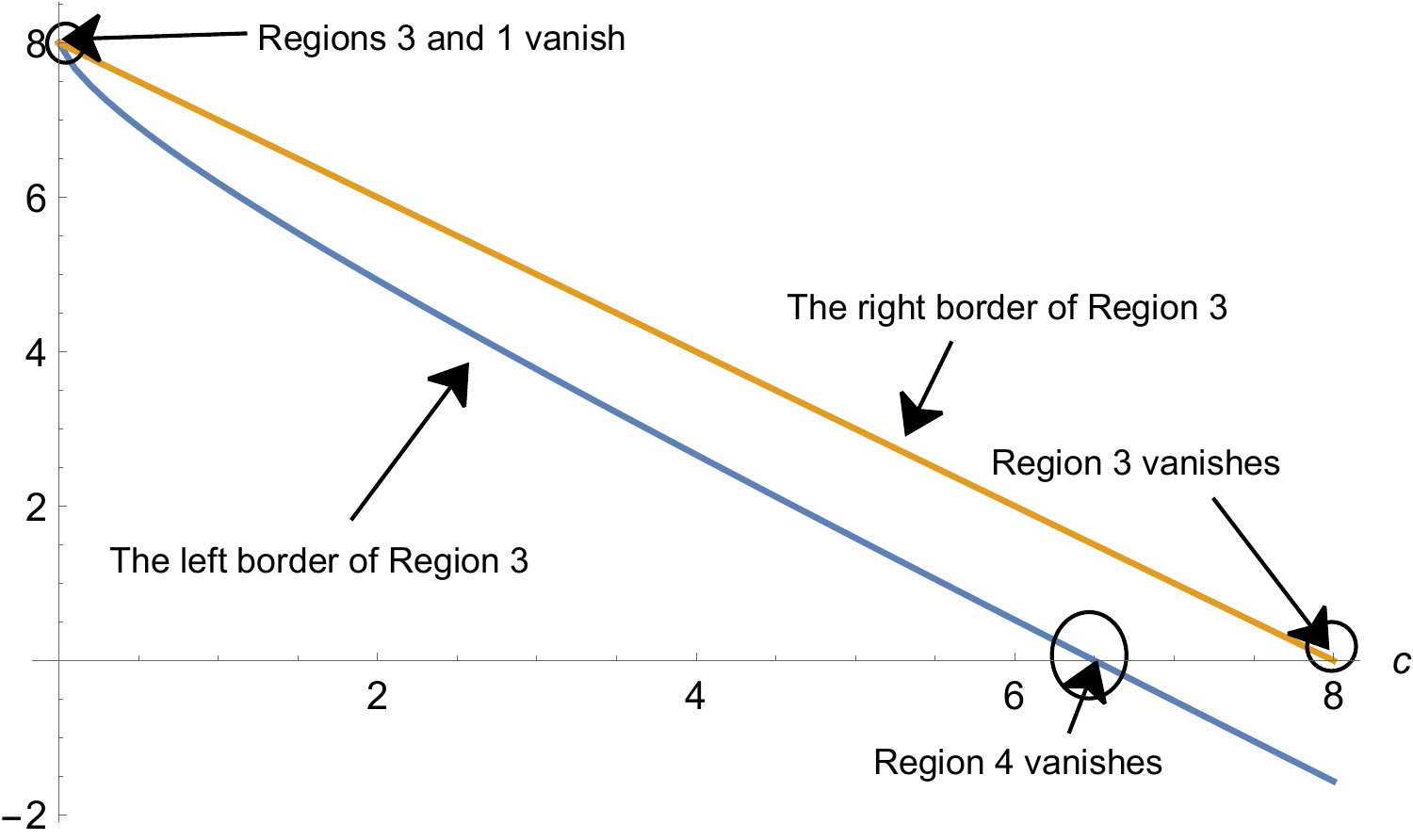}
\caption{The borders of Region 3 as functions of the arrival cost $c$. The upper line shows the right border, and the bottom line shows the left border.}\label{F3000}
\end{center}
\end{figure}

\subsection{Results for Section~\ref{Munits} -Generalization to multiple units when $\l$ is high}\label{AMunits}

\subsubsection{Proposition~\ref{NP1}}\label{ANP1}
\begin{proposition}\label{NP1}
For all $v, k, p, N,$
\begin{equation}\label{EN1}
u_1=-1+(v-p-k)\left[  \left( 1-\frac{N}{\l_1}\right)\sum_{i=0}^{N-1}\frac{e^{-\l_1}\l_1^i}{{i!}}+ \frac{N}{\l_1} -\frac{e^{-\l_1}\l_1^{N-1}}{{(N-1)!}}\right]
\end{equation}

\begin{equation}\label{EN2}
u_2=-1+(v-p)\sum_{s=1}^{N}\frac{e^{-\l_1}\l_1^{N-s}}{{(N-s)!}}\left[  \left( 1-\frac{s}{\l_2}\right)
\sum_{i=0}^{s-1}\frac{e^{-\l_2}\l_2^i}{{i!}}+ \frac{s}{\l_2} -\frac{e^{-\l_2}\l_2^{s-1}}{{(s-1)!}}\right].
\end{equation}
\end{proposition}

Note first that
substituting $N=1$  in~\eqref{EN1} and \eqref{EN2} yields
$u_1= -1+(v-k-p)\frac{\left(1-e^{-\l_1}\right)}{\l_1},$
and
$u_{2}=-1+e^{-\l_1}(v-p)\frac{\left(1-e^{-\l_2}\right)}{\l_2},$
which are indeed what we had established for $N=1$ (see~\eqref{E9000N} and~\eqref{E9001N}).

\begin{proof}
Dividing~\eqref{EN10} by $c,$
 gives
$$u_1=-1+(v-k-p)\left[\P(X_1\le N-1)+\sum_{i=N}^{\infty} \P(X_1=i){N\over i+1}\right].$$
Hence
$$u_1=-1+(v-k-p)\left[  \sum_{i=0}^{N-1}\frac{e^{-\l_1}\l_1^i}{{i!}} +N\sum_{i=N}^{\infty} \frac{e^{-\l_1}\l_1^i}{{i!}}{1\over i+1}\right]$$
$$=-1+(v-k-p)\left[  \sum_{i=0}^{N-1}\frac{e^{-\l_1}\l_1^i}{{i!}} +\frac{N}{\l_1}\sum_{i=N}^{\infty} \frac{e^{-\l_1}\l_1^{i+1}}{{(i+1)!}}\right].$$
Now, since  $\sum_{i=0}^{\infty}\frac{e^{-\l_1}\l_1^i}{{i!}}=1,$ we get
$$u_1=-1+(v-k-p)\left[  \sum_{i=0}^{N-1}\frac{e^{-\l_1}\l_1^i}{{i!}} +\frac{N}{\l_1}\left( 1-\sum_{i=0}^{N} \frac{e^{-\l_1}\l_1^{i}}{{i!}}\right)\right]$$
$$=-1+(v-k-p)\left[  \sum_{i=0}^{N-1}\frac{e^{-\l_1}\l_1^i}{{i!}} +\frac{N}{\l_1}  -\frac{N}{\l_1}\left( \sum_{i=0}^{N-1} \frac{e^{-\l_1}\l_1^{i}}{{i!}} + \frac{e^{-\l_1}\l_1^{N}}{{N!}}  \right)\right]$$
$$=-1+(v-k-p)\left[  \sum_{i=0}^{N-1}\frac{e^{-\l_1}\l_1^i}{{i!}} +\frac{N}{\l_1}  -\frac{N}{\l_1} \sum_{i=0}^{N-1} \frac{e^{-\l_1}\l_1^{i}}{{i!}} - \frac{e^{-\l_1}\l_1^{N-1}}{{(N-1)!}} \right].$$
Thus,
$$u_1=-1+(v-p-k)\left[  \left( 1-\frac{N}{\l_1}\right)\sum_{i=0}^{N-1}\frac{e^{-\l_1}\l_1^i}{{i!}}+ \frac{N}{\l_1} -\frac{e^{-\l_1}\l_1^{N-1}}{{(N-1)!}}\right].$$
The proof of~\eqref{EN2} is very similar.
\end{proof}

\

\subsubsection{The expected profit of the seller in the multi-unit case}\label{APN10}

\begin{proposition}\label{PN10}
\begin{equation}\label{EN3}
\pi=p\left( N\left( 1-  \sum_{i=0}^{N-1}\frac{e^{-\l_1}\l_1^i}{{i!}}\right) + \sum_{s=1}^{N}\frac{e^{-\l_1}\l_1^{N-s}}{{(N-s)!}}\left[ N +\l_2\sum_{i=0}^{s-1}\frac{e^{-\l_2}\l_2^i}{{i!}} -s\sum_{m=0}^{s}\frac{e^{-\l_2}\l_2^m}{{m!}}  \right]    \right).
\end{equation}
\end{proposition}

Note that substituting $N=1$ in~\eqref{EN3} gives  $\pi=p(1-e^{-\l_1-\l_2})$  which is indeed expected profit for the single-unit case (see~\eqref{E5000}).
\begin{proof}
The expected profit,  $\pi$,  equals $p$ times the expected total number of units sold.
The first expression in the brackets $N\left( 1-  \sum_{i=0}^{N-1}\frac{e^{-\l_1}\l_1^i}{{i!}}\right)$ applies when $s=0,$ namely when all $N$ units were sold in period 1. In that case sales are $N.$ This is multiplied by the probability that indeed $s=0,$ namely $Pr(X_1\geq N),$ to give  $N\sum_{i=N}^{\infty}\frac{e^{-\l_1}\l_1^i}{{i!}}= N\left( 1-  \sum_{i=0}^{N-1}\frac{e^{-\l_1}\l_1^i}{{i!}}\right).$
The second  expression consists of a sum with inner brackets.
The expression in the brackets gives the conditional expected number of units sold if $s=1,2,\dots , N$ units remain for sale in period 2. This means that $N-s>0$ units were sold in period 1.
The probability that $N-s>0$ units were sold in period 1 is $\frac{e^{-\l_1}\l_1^{N-s}}{{(N-s)!}},$ / $s=1,2,\dots N,$ which is therefore the expression multiplying the inner brackets.
To understand the expressions in the inner brackets, note that sales in period 1 are $N-s,$ which is the first expression in the brackets. The number of arrivals in period 2 may be equal to or less than the number $s$ of units available, or it may exceed that number. In the first case, expected sales in period 2 are 
$\sum_{i=0}^{s}i\frac{e^{-\l_2}\l_2^i}{{i!}}$ which equals $\l_2\sum_{i=0}^{s-1}\frac{e^{-\l_2}\l_2^i}{{i!}}.$ In the second case, namely where the number of arrivals in period 2 exceeds $s,$ expected sales in this period are $s\sum_{m=s+1}^{\infty}\frac{e^{-\l_2}\l_2^m}{{m!}}= s\left( 1 -\sum_{m=0}^{s}\frac{e^{-\l_2}\l_2^m}{{m!}}\right).$

Substituting $N=1$  in~\eqref{EN1} and \eqref{EN2} verifies that 
$u_1= -1+(v-k-p)\frac{\left(1-e^{-\l_1}\right)}{\l_1},$
and
$u_{2}=-1+e^{-\l_1}(v-p)\frac{\left(1-e^{-\l_2}\right)}{\l_2},$
which are indeed what we had established for $N=1$ (see~\eqref{E9000N} and~\eqref{E9001N}).
\end{proof}

\subsubsection{Results of the model with a single unit hold for multiple units}\label{AGER}

The same seven types of equilibria found in the single-unit case (see Section~\ref{SCHI}) also hold with multiple units. As in the single-unit case (see Section~\ref{SU}), when $\l$ is large, only equilibria of types 1-4 exist,  and here too, most potential consumers avoid arriving. We confirmed (numerically) that as in the single-unit model the conditions defining the border between Regions 3 and 4 is unique. We denote it as $p_0.$ Formally,

\begin{definition}\label{DP0}
$p_0$ is the unique solution for $u_1=u_2=0,$ and $\l_2=0.$
\end{definition}

Hence,

\begin{itemize}
\item Region 1  consists of all $p$ satisfying \ $v-1\leq p\leq v,$
\item Region 2  consists of all $p$ satisfying \ $v-1-k\leq p< v-1,$
\item Region 3  consists of all $p$ satisfying \ $p_0\leq p< v-1-k,$
\item Region 4  consists of all $p$ satisfying \ $0\leq p< p_0,$
\end{itemize}

This section demonstrates that all the results proved in the single-unit case still hold for any number of units $N$. Additionally we find some interesting effects of $N$ on consumer behavior. Most of the mathematical analysis becomes too difficult when $N>1,$ hence we used  {\it{Mathematica}} software for our investigations. All the figures below were executed for $k=2$ and $v=10.$
The main results of the single unit model were:
\begin{enumerate}
\item In Region 3, $\l_2$ increases with the price $p$ . This increase exceeds the decrease of $\l_1,$ such that the sum $\l_1+\l_2$ increases with $p$.
\item Expected profit $\pi$ increases monotonically with $p$ in Region 3, and has at most one local maximum in each of Regions 2 and 4.
\item If there exists a local maximum in Region 2, then its value is always the global maximum.

The above results yield the following:
\item The firm does not profit by selling early (and usually loses by doing so).
\item If, however, the firm is forced to also offer the good early, then the firm maximizes profits by setting a price which induces all consumers to arrive early, or all to arrive late, depending on the good's value to the customers.
\end{enumerate}

As stated for the single-unit case (i.e., $N=1$), the first item in the list of main results above is surprising: in Region 3 (where consumers arrive in both periods) total arrivals increase with the price.
Recall that Region 3 is defined by $p_0\leq p<v-1-k.$ Substituting $k=2$ and $v=10$ means that Region 3  is defined by $p_0\leq p<7.$ Indeed, in the multi-unit case too, $\l_1$ decreases and $\l_2$ increases in Region 3 (not presented here). Moreover, $\l_1+\l_2$  also increases in Region 3 (Item 1). To see that, see Figure~\ref{FNEW8} in Section~\ref{Munits} in the main text, which presents $\l_1+\l_2$ for $N=1,2,4,10$ (the graph for $N=50$ is not shown because it is in a different scale, but it has the same characteristics as $N=1,2,4,10$).

\

As shown in Figure~\ref{FNEW8} and as expected, in all other regions $\l_1+\l_2$ decreases with $p$.
Figures~\ref{FNEW4}-\ref{FNEW7} in Section~\ref{Munits} in the main text,  present $\pi$ for  $v=10, k=2$  and several values of $N.$ The Figures illustrate that items 2,3  above hold for $N>1$. First, the profit $\pi$ increases monotonically in Region 3 and has at most one local maximum in Regions 2 and 4 (Item 2). In the single-unit case, the monotonicity of $\pi$ in Region 3 followed immediately from the first item, since for $N=1$ we have $\pi=p(1-e^{-\l_1-\l_2}),$ which increases with $\l_1+\l_2.$ But for $N>1$ the expression for $\pi$ is much more complicated (see Proposition~\ref{PN10}). Still, the figures below showing $\pi$ for $v=10, k=2$  and $N=1,2,4,10, 50$ illustrate that the monotonicity of $\pi$ in Region 3 continues to hold when $N>1$. We solved numerically with many values of $v, k, N$, and always observe this behavior.

In all figures the local maximum in Region 2 exists and is therefore the global maximum (Item 3).
Item 5 follows immediately from Item 2. The three candidates for global maximum (the local maximum $p_2$ in Region 2, the local maximum $p_4$ in region 4, and the point $p=v-k-1=7$ separating Regions 3 and 2) all have consumers arriving in only one period. With $p_4$ consumers arrive only in period 1. With $p_2$ and  $p_3,$ consumers arrive only in period 2. Item 4 follows immediately from Item 3, since not offering the good in period 1 is equivalent to $k$ approaching infinity. In that case   Regions 3 and 4 do not exist (since Regions 3 and 4 require that $v-p-1>k$). Thus the global maximum must be in Region 2, which according to Item 3 is the best.

\subsubsection{The effect of supply on profit, profit-maximizing price, behavior of consumers}

Note that the left border $p_0$ of Region 3 increases with $N,$ while  the right border  $v-1-k = 10-1-2=7$ is constant with $N.$ This is shown in Figure~\ref{FNEW8} in Section~\ref{Munits} in the main text. The vertical lines (dashed lines at $p_0$ for $N=1,2,4,10$ separately, and solid line at the right border $7$ of Region 3).
Thus increasing $N$ reduces the size of Region 3.  Figures~\ref{FNEW4}-\ref{FNEW7} in Section~\ref{Munits} in the main text show that Region 4 increases at the expense of Region 3. Take for example $p=6.5.$ When $N=1$ it lies in Region 3, but when $N=50$ it lies in Region 4.
Denote $\l_2(N, p)$ as the arrival rate in period 2 when the price for a unit is $p$ and $N$ units are available in the first period.
\begin{proposition}\label{Lambda2}
$$\lim_{N\to\infty}\frac{\l_2(N, v-1)}{N}=1.$$
\end{proposition}

\begin{proof}

Recall that by~\eqref{EN2} we have that

\begin{equation}\label{EN2again}
u_2=-1+(v-p)\sum_{s=1}^{N}\frac{e^{-\l_1}\l_1^{N-s}}{{(N-s)!}}\left[  \left( 1-\frac{s}{\l_2}\right)
\sum_{i=0}^{s-1}\frac{e^{-\l_2}\l_2^i}{{i!}}+ \frac{s}{\l_2} -\frac{e^{-\l_2}\l_2^{s-1}}{{(s-1)!}}\right].
\end{equation}

The point $p=v-1$ is the right end of Region 2  in which for all $N,$ $\l_1(N, p)=0,$ and $u_2=0.$

Substituting this in~\eqref{EN2again} yields

\begin{equation}\label{EN222}
u_2=-1+(v-p)\left[  \left( 1-\frac{N}{\l_2}\right)
\sum_{i=0}^{N-1}\frac{e^{-\l_2}\l_2^i}{{i!}}+ \frac{N}{\l_2} -\frac{e^{-\l_2}\l_2^{N-1}}{{(N-1)!}}\right]=0,
\end{equation}
 where $\l_2=\l_2(N, p)$ for all $p$ in Region 2.

We wish to find a price $p$ in Region 2 for which
\begin{equation}\label{PPLIM}
\lim_{N\to\infty}\frac{\l_2(N, p)}{N}=1.
\end{equation}

Substituting $\l_2=N$ in~\eqref{EN222} gives
$$(v-p)\left[ 1-\frac{e^{-N}N^{N-1}}{(N-1)!} \right]=1.$$
Hence,
$$(v-p)\left[ 1-\frac{e^{-N}N^{N}}{N!} \right]=1.$$
By Stirling's formula, for large $N,$ \ $\frac{e^{-N}N^{N}}{N!}\sim \frac{1}{\sqrt{2\pi N}}.$

Thus,
$$(v-p)\left[ 1-\lim_{N\to\infty}\frac{1}{\sqrt{2\pi N}} \right]=1,$$
implying that $v-p=1,$ and so $p=v-1$ satisfies~\eqref{PPLIM}, namely,
\begin{equation}\label{Stir}
\lim_{N\to\infty}\frac{\l_2(N, v-1)}{N}=1.
\end{equation}
\end{proof}

\subsection{Low demand}\label{AFIN}

As explained in the main text, when demand is low, there may be pairs $(v,k)$ with multiple equilibria. Thus we do not explore profit maximization in the this case. However, regarding our result on the profit maximizing price, the following shows that  also in the low demand case the monopolist wants all consumers to arrive at the same time, either all in period 1 (i.e., $\l_2=0$), or all in period 2 (i.e., $\l_1=0$), depending on the parameters of the model.

Given an arrival rate $\lambda$, equilibrium types 5-7 do exist in general.  In these cases, the sum of arrival rates  $\l_1+\l_2$ equals $\l$ by definition (see Section~\ref{SCHI}).  This implies that the expected profit $\pi=p\Big( 1-e^{-(\l_1+\l_2)} \Big)=p\Big( 1-e^{-\lambda}\Big)$ is linear and increasing in $p$. Hence the candidates for the global maximum point are the border-points between the regions. For the other (former) equilibria types 1-4 (in which by definition $\l_1+\l_2<\l$), we proved that for equilibrium type 3, \ $\l_1+\l_2$ increases with $p.$ This holds also in the low demand case. So the points on the left border of Region 3 (see Figure~\ref{f2}) are candidates for the maximum profit. Recall also that in Region 2 $\l_1=0,$ and that in Region 4 $\l_2=0.$ Hence the result that the profit-maximizing price induces all consumers to arrive at the same time also holds in the low demand case.

In the next section, we  normalize $V-P=1$ and $\l=1$ (instead of $c$ being the unit value).

\subsubsection{Equilibria for the low demand case when normalizing $V-P=1$ and $\l=1$}\label{SMR}

Recall that the expected number of people arriving in period 1 is $x \equiv \l_1$; the expected number of people arriving in period 2 is $y \equiv \l_2$. For each type of equilibrium we present the conditions on $c$ and $K$ for having an equilibrium of that type.  Note that we describe the feasibility region for equilibria of types 3-5 using parametric representation of the conditions on $c$ and $K$.

\begin{enumerate}
\item
A type-1 equilibrium (where no consumer ever arrives) has $\l_1=\l_2=0$, $U_{1}\le 0$, and $U_2\le 0$, requiring that $1-K \le c$ and that $c \ge 1$.

\item
A type-2 equilibrium (where no consumer arrives in period 1, but may arrive in period 2) has $\l_1=0$. The condition $U_2=0$ is equivalent to $c={1-e^{-y}\over y}$, implying that $c\in(1-1/e,1)$. The condition $U_{1}\le 0$ means that $K+c \ge 1$.

\item
For the type-3 equilibria (where a consumer chooses to arrive in period 1, or in period 2, or not to arrive at all, all with positive probabilities), for $0\le x\le 1$ and $0\le y\le1-x$, the condition $U_{1}=0$ reduces to
$c=(1-K){1-e^{-x}\over x}$. The condition $U_2=0$ amounts to
$c=e^{-x}{1-e^{-y}\over y}\in\left(e^{-x}{1-e^{(1-x)}\over1-x},e^{-x}\right)$.

\item
For the type-4 equilibrium (where no consumer arrives in period 2, but consumers may arrive in period 1), for $0 \le x \le 1$, the condition $U_{1}=0$ implies that $K=1-c{x \over 1-e^{-x}}$, and $U_2 \le 0$ with $q_2=0$ implies that $c \ge e^{-x}$.

\item
In the type-5 equilibrium (where a consumer is indifferent about when to arrive), $\l_1+\l_2=1$. Define $x \equiv \l_1$. For $0 \le \l_1 \le 1$, the condition $U_{1}=U_2$ reduces to $K=1-{x\over1-x}{1-e^{-(1-x)}\over
e^x-1}$ . And the non-negativity of $U_2$ reduces to $c\le e^{-x}{1-e^{-(1-x)}\over1-x}$. This condition also implies that $c \le 1-{1\over e}$.

\item
In a type-6 equilibrium (where consumers may arrive in period 1), $\l_1=1$ and $\l_2=0$. The conditions $U_{1} \ge U_2$ and $U_{1} \ge 0$ reduce to $(1-K)(1-{1\over e})\ge {1\over e}$, or $K\le {e-2\over e-1}$, and $c \le (1-K)(1-{1\over e})$.

\item
In a type-7 equilibrium (where consumers may arrive in period 2), $\l_1=0$ and $\l_2=1$.  The conditions $U_2\ge U_{1}$ and $U_2\ge0$ reduce to $1-{1\over e}\ge 1-K$, or $K \ge{1\over e}$ and $c \le 1-{1\over e}$.

\end{enumerate}

Figure \ref{f1} describes the regions corresponding to the different types of equilibria when $\l=1$. The point at the intersection of types 7,5,3 and 2 has $c=1-1/e$ and $K=1/e$.  The point at the intersection of types 7, 6, 5, 4, and 3 has $c=1/e$ and $K=(e-2)/(e-1)$.

\begin{figure}[h]
\centering
\vspace{-5cm}
\includegraphics[scale=0.6]{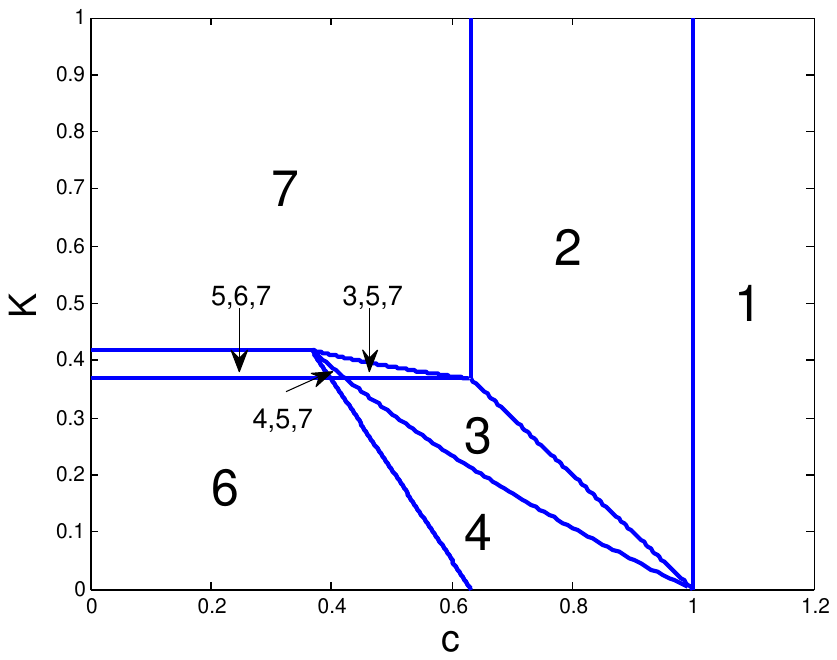}
\vspace{-5cm}
\caption{Types of equilibria in the $(c,K)$ space.}\label{f1}
\end{figure}

From Figure \ref{f1} we conclude that for most combinations of $c$ and $K$ a unique equilibrium exists. However, in the narrow region where a type-5 equilibrium exists, other equilibria also exist.

In the following, denote $\lambda_e \equiv (\l_1,\l_2)$.
\medskip

\noindent{\bf Example 1}  Let a consumer's cost of going to the store be $c=0.2$. Let a consumer's penalty for buying the good earlier than at his ideal period be $K=0.4$. Then the pair of arrival rates at the store $\l_1=0$ and $\lambda_2=1$ is a type-7 equilibrium (where consumers may arrive in period 2). The pair $\l_1=1$ and $\lambda_2=0$ (denoted by $\lambda_e=(1,0)$) is a type-6 equilibrium (an equilibrium where consumers may arrive in period 1). And the pair $\l_1=0.6305$ and $\lambda_2=0.3995$ is a type-5 equilibrium. Note that the type-7 equilibrium is efficient. A consumer who wants to consume the good in period 2 arrives in period 2. A type-6 equilibrium suggests inefficiency. Although consumers want the good in period 2, they arrive in period 1. So if the waiting cost is positive, a consumer who arrives early may incur a cost without increasing the benefit he gets from the good. Early arrival reflects rent-seeking behavior. The same applies for a type-5 equilibrium. In terms of consumer welfare (the aggregated utilities), note that the consumer welfare function, denoted by $SW$, satisfies
$$SW=(1-e^{-\l_1})(V-P-K)+e^{-\l_1}(1-e^{-\l_2})(V-P)-c(\l_1+\l_2)=$$
$$(1-e^{-\l_1})0.6+e^{-\l_1}(1-e^{-\l_2})-0.2(\l_1+\l_2).$$

Thus consumer welfare for the equilibria of types 7, 6 and 5 are 0.432, 0.179 and 0.056, respectively. Put differently, if the store refused to sell in period 1 consumers would be better off.

Early arrival reflects behavior of {\it strategic complements} or {\it follow the crowd} (see Hassin and Haviv 2003), where a consumer's best response tends to follow the strategy of the others. Typically, such situations have two extreme equilibria in pure strategies, and one equilibrium with a strategy.

\medskip

\noindent{\bf Example 2}  Let $c=0.4$ and $K=0.37$. Then  $\lambda_e=(1, 2)$ is a type-7 equilibrium, $\lambda_e=(0.041,0.959)$ is a type 5 equilibrium, and $\lambda_e=(0.989,0)$ is a type-4 equilibrium.
\medskip

\noindent{\bf Example 3}  Let $c=0.4$ and $K=0.4$. Then $\lambda_e=(1, 2)$ is a type-7 equilibrium, $\lambda_e=(0.63,0.37)$ is a type-5 equilibrium; $\lambda_e=(0.8742,0.085)$ is a type-3 equilibrium.


\section{Notation}
\begin{center}

\begin{tabular}{|m{1cm}|m{15cm}|}
\hline
$c$ & Cost of going to store \\
\hline
$V$ & Value of good to consumer if bought at his ideal period \\
\hline
$v$  & $\frac{V}{c}$\\
\hline
$K$ & Reduction in consumer's utility if he buys the good  early \\
\hline
$k$  & $\frac{K}{c}$ \\
\hline
$P$ & Price of good \\
\hline
$p$ & $\frac{P}{c}$ \\
\hline
$\Pi$ & Profits \\
\hline
$\pi$ & $\frac{\Pi}{c}$ \\
\hline
$q_{t}$ & Probability that a consumer chooses to arrive in period $t$ \\
\hline
$U_{t}$ & Expected utility of a consumer who arrives in period $t$ \\
\hline
$u_t$ & $\frac{U_t}{c}$ \\
\hline
$\l$ & Arrival rate of potential consumers \\
\hline
$\l_t$ & Arrival rate of consumers at the store in period t \\
\hline
$x$ &  $\l_1$ \\
\hline
$y$&  $\l_2$ \\
\hline
$v(k)$ & Border between Regions 3 and 4 in terms of $v$ \\
\hline
$p_0$ & Border between Regions 3 and 4 in terms of $p$ \\
\hline
$W[a]$ &  Lambert function \\
\hline
$W$ & $W[-(k+1)e^{-(k+1)}]$ \\
\hline
$p_2, p_4$ & Local maximum points of Regions 2 and 4 respectively \\
\hline
$p_3$  &  $v-k-1$ The point separating between Regions 3 and 2 (i.e., the right end of Region 3) \\
\hline
$p^*$ & Global maximum point of $\pi$ \\
\hline
$\pi^*$ & Maximum value of $\pi$ \\
\hline
$\pi_N^*$ & Maximum expected profit when $N$ units are for sale \\
\hline
\end{tabular}
\bigskip

\end{center}
\

\


\singlespacing

\pagebreak

\end{document}